\documentclass{article}
\usepackage{matharticle}
\usepackage{ifthen}
\usepackage{mdwlist}
\usepackage{bm}
\usepackage{xcolor}
\usepackage[colorlinks=true,linkcolor=blue,citecolor=blue,urlcolor=blue]{hyperref}
\usepackage{enumitem}
\usepackage{graphicx}
\usepackage{xspace}
\usepackage{verbatim}
\usepackage[margin=1in]{geometry}
\usepackage{thm-restate}
\usepackage{latexsym}
\usepackage{epsfig}
\usepackage{mleftright}
\usepackage[colorinlistoftodos,textsize=scriptsize]{todonotes}
\def\withcolors{1}
\def\withnotes{1}
\def\eps{\ve}
\renewcommand{\epsilon}{\ve}
\def\ve{\varepsilon}

\newcommand{\pr}[2][]{\mathrm{Pr}\ifthenelse{\not\equal{}{#1}}{_{#1}}{}\!\left[#2\right]}

\ifnum\withcolors=1
  \newcommand{\gcolor}[1]{{\color{red}#1}} 
\else

  \newcommand{\gcolor}[1]{{#1}}
\fi

\ifnum\withnotes=1
  \newcommand{\gnote}[1]{\par\gcolor{\textbf{G: }\sf #1}} %
  \newcommand{\gfootnote}[1]{\footnote{{\bf \gcolor{Gautam}}: {#1}}}

\else
  \newcommand{\gnote}[1]{}
  \newcommand{\gfootnote}[1]{}

\fi
\newcommand{\ignore}[1]{\leavevmode\unskip} %

\newcommand{\lp}[1]{\ell_{#1}}
\newcommand{\unif}{\mathrm{Unif}}
\newcommand{\poi}{\mathrm{Poi}}
\newcommand{\calU}{\mathcal{U}}

\newcommand{\mas}[2]{\|#1_#2\|_\infty}
\newcommand{\rat}[2]{\rho_{#1,#2}}
\newcommand{\samcom}[4]{\operatorname{SC}(#1,#2,#3,#4)}

\title{Towards a Total Tabulation of the Tolerance of Tolerant Testers} %
\title{The Toll for Tolerance in Testing}
\title{The Price of Tolerance in Distribution Testing\thanks{Authors are listed in alphabetical order.}} %
\author{
Cl\'ement L. Canonne\thanks{University of Sydney. Email: \url{clement.canonne@sydney.edu.au}.}
\and
Ayush Jain\thanks{UC San Diego. Email: \url{ayjain@eng.ucsd.edu}.}
\and
Gautam Kamath\thanks{Cheriton School of Computer Science, University of Waterloo. Email: \url{g@csail.mit.edu}. Supported by an NSERC Discovery Grant.}
\and
Jerry Li\thanks{Microsoft Research. Email: \url{jerrl@microsoft.com}.}
}

\begin{document}

\maketitle
\thispagestyle{empty}
\setcounter{page}{0}
\begin{abstract}
We revisit the problem of tolerant distribution testing. That is, given samples from an unknown distribution $p$ over $\{1, \dots, n\}$, is it $\varepsilon_1$-close to or $\varepsilon_2$-far from a reference distribution $q$ (in total variation distance)?
Despite significant interest over the past decade, this problem is well understood only in the extreme cases.
In the noiseless setting (i.e., $\varepsilon_1 = 0$) the sample complexity is $\Theta(\sqrt{n})$, strongly sublinear in the domain size.
At the other end of the spectrum, when $\varepsilon_1 = \varepsilon_2/2$, the sample complexity jumps to the barely sublinear $\Theta(n/\log n)$.
However, very little is known about the intermediate regime.
We fully characterize the price of tolerance in distribution testing as a function of $n$, $\varepsilon_1$, $\varepsilon_2$, up to a single $\log n$ factor.
Specifically, we show the sample complexity to be 
\[\tilde \Theta\mleft(\frac{\sqrt{n}}{\ve_2^{2}} + \frac{n}{\log n} \cdot \max \mleft\{\frac{\ve_1}{\ve_2^2},\mleft(\frac{\ve_1}{\ve_2^2}\mright)^{\!\!2}\mright\}\mright),\]
providing a smooth tradeoff between the two previously known cases.
We also provide a similar characterization for the problem of tolerant equivalence testing, where both $p$ and $q$ are unknown.
Surprisingly, in both cases, the main quantity dictating the sample complexity is the ratio $\varepsilon_1/\varepsilon_2^2$, and not the more intuitive $\varepsilon_1/\varepsilon_2$.
Of particular technical interest is our lower bound framework, which involves novel approximation-theoretic tools required to handle the asymmetry between $\varepsilon_1$ and $\varepsilon_2$, a challenge absent from previous works.
\end{abstract}

\newpage

\section{Introduction}
Upon observing independent samples from an unknown probability distribution, can we determine whether it possess some property of interest?
This natural question, known as distribution testing or statistical hypothesis testing, has enjoyed significant study from several communities, including theoretical computer science, statistics, information theory, and machine learning.
The prototypical problem in this area is \emph{identity testing} (sometimes called \emph{goodness-of-fit} or \emph{one-sample} testing): given samples from an unknown probability distribution $p$ over $[n]$, test whether it is equal to some reference distribution $q$, or $\ve$-far in $\lp{1}$-distance. 
It is now well understood that $\Theta\mleft({\sqrt{n}}/{\ve^2}\mright)$ samples are necessary and sufficient to solve this problem~\cite{Ingster94,GoldreichR00,BatuFFKRW01,Paninski08,ValiantV14,DiakonikolasKN15a,AcharyaDK15,DiakonikolasGPP18}.
Quite surprisingly, this sample complexity is strongly sublinear in $n$, enabling sample-efficient testing even over large domains.

The drawback of this formulation is that it is very particular in terms of the relationship between $p$ and $q$.
More precisely, it prescribes only that one must distinguish between the cases where $p$ and $q$ are far versus when they are \emph{exactly} equal -- no guarantees are provided for any intermediate case, e.g., for when $p$ and $q$ are \emph{close} but not identical.
This restriction limits the relevance of solutions to this problem, as it is unrealistic to assume \emph{precise} knowledge of a distribution due to a number of reasons, including model misspecification, imprecise measurements, or dataset contamination. 

\noindent To address these concerns, the problem of \emph{tolerant identity testing} was introduced~\cite{ParnasRR06}, which is the main focus of our work.

\medskip \framebox{
\begin{minipage}{15.5cm}
{\em Tolerant Identity Testing}: Given an explicit description of a distribution $q$ over $[n]$, sample access to a distribution $p$ over $[n]$, and bounds $\ve_2 > \ve_1 \ge 0$, and $\delta >0$, distinguish with probability at least $1-\delta$ between $\norm{p-q}_1 \le \ve_1$ and $\norm{p-q}_1 \ge \ve_2$, whenever $p$ satisfies one of these two inequalities.
\end{minipage}}\medskip

\noindent We will also study the problem of \emph{tolerant equivalence testing} (sometimes called tolerant \emph{closeness} or \emph{two-sample} testing):

\medskip \framebox{
\begin{minipage}{15.5cm}
{\em Tolerant Equivalence Testing}: Given sample access to distributions $p$ and $q$ over $[n]$, and bounds $\ve_2 > \ve_1 \ge 0$, and $\delta >0$, distinguish with probability at least $1-\delta$ between $\norm{p-q}_1 \le \ve_1$ and $\norm{p-q}_1 \ge \ve_2$, whenever $p, q$ satisfy one of these two inequalities.
\end{minipage}} \medskip

Focusing our attention on tolerant identity testing and constant $\varepsilon_2$, it is natural to consider the strong tolerance requirement of $\ve_1 = \ve_2/2$, in which the two cases are separated only by a constant factor. 
One would ideally like to maintain the strongly sublinear sample complexity of $\mathcal{O}(\sqrt{n})$, as in the non-tolerant case where $\ve_1=0$.
Unfortunately, this is impossible: as shown by Valiant and Valiant~\cite{ValiantV10a,ValiantV10b,ValiantV11a}, \smash{$\Theta\big(\frac{n}{\log n}\big)$} samples are necessary and sufficient, see also~\cite{JiaoHW18,JiaoVHW17,HanJW16}.
On the other end of the spectrum, it is known that mild tolerance of \smash{$\ve_1 = \frac{\ve_2}{2\sqrt{n}}$} is achievable with the same strongly-sublinear sample complexity of $\mathcal{O}(\sqrt{n})$, by converting $\lp{2}$-tolerance to $\lp{1}$-tolerance~\cite{GoldreichR00,BatuFFKRW01,BatuFRSW13,DiakonikolasKN15a, DiakonikolasK16, DaskalakisKW18}.
However, existing results only capture these two extremes, and we have very little understanding of the intermediate landscape of tolerant testing. 
Does there exist a smooth hierarchy of increasingly difficult testing problems, or is there a sharp transition in the sample complexity from strongly to barely sublinear?

\subsection{Results and Techniques}
We provide a complete characterization of the sample complexity of tolerant identity and equivalence testing (up to a single logarithmic factor in the domain size $n$). 
Our main results are as follows:
\begin{theorem}[Identity testing (Informal; see Theorem~\ref{th:ubU} and Corollary~\ref{cor:lbU})]\label{th:identity:informal}
The sample complexity of tolerant identity testing over $[n]$ with parameters $0\leq \ve_1 < \ve_2 \leq 1$ is
\[
\Omega\mleft(\frac{\sqrt{n}}{\ve_2^{2}} + \frac{n}{\log n} \cdot \max \mleft\{\frac{\ve_1}{\ve_2^2},\mleft(\frac{\ve_1}{\ve_2^2}\mright)^{\!\!2}\mright\}\mright) \text{ and }
\mathcal{O}\mleft(\frac{\sqrt{n}}{\ve_2^{2}} + n \cdot \max \mleft\{\frac{\ve_1}{\ve_2^2},\mleft(\frac{\ve_1}{\ve_2^2}\mright)^{\!\!2}\mright\}\mright)\,.
\]
\end{theorem}
\begin{theorem}[Equivalence testing (Informal; see Theorem~\ref{th:ubC} and Corollary~\ref{cor:lbC})]\label{th:closeness:informal}
The sample complexity of tolerant equivalence testing over $[n]$ with parameters $0\leq \ve_1 < \ve_2 \leq 1$ is
\[
\Omega\mleft(\max\mleft\{\frac{\sqrt{n}}{\ve_2^{2}}, \frac{n^{2/3}}{\ve_2^{4/3}}\mright\}
+ \frac{n}{\log n} \cdot \max \mleft\{\frac{\ve_1}{\ve_2^2},\mleft(\frac{\ve_1}{\ve_2^2}\mright)^{\!\!2}\mright\}\mright) \text{ and }
\mathcal{O}\mleft(\max\mleft\{\frac{\sqrt{n}}{\ve_2^{2}}, \frac{n^{2/3}}{\ve_2^{4/3}}\mright\}
+ n \cdot \max \mleft\{\frac{\ve_1}{\ve_2^2},\mleft(\frac{\ve_1}{\ve_2^2}\mright)^{\!\!2}\mright\}\mright)\,.
\]
\end{theorem}
\noindent In both cases, we give computationally-efficient algorithms which achieve the upper bounds.
Moreover, one interesting feature of our algorithms is that they require no knowledge of $\ve_1$, which only arises in the sample complexity: that is, our algorithm automatically achieves the best possible $\eps_1$, for a given target $\eps_2$ and number of samples.

It is worth noting that prior to our work, only two extreme points of the full tradeoff we show were known:
\begin{itemize}
    \item the ``non-tolerant'' (noiseless) case where $\ve_1=0$, for which the $\Theta(\sqrt{n}/\ve_2^2)$ sample complexity (or, for the equivalence testing version, $\Theta(\max\{\sqrt{n}/\ve_2^{2}, {n^{2/3}}/{\ve_2^{4/3}}\})$)~\cite{Paninski08,Valiant11,ChanDVV14, ValiantV14}. In the case of identity testing, it is further known that some of the optimal testers (namely, those based on testing in the $\lp{2}$ distance as a proxy) achieve a weak tolerance of $\ve_1 = \ve_2/\sqrt{n}$ ``for free'', due to the relation between $\lp{1}$ and $\lp{2}$ norm along with the Cauchy--Schwarz inequality.
    \item the maximally noisy case where $\ve_1 = \Theta(\ve_2)$, for which results of Valiant and Valiant~\cite{ValiantV10a,ValiantV10b,ValiantV11a} as well as follow-up works~\cite{JiaoHW18,JiaoVHW17,HanJW16} show that the sample complexity must grow as $\Theta(n/\log n)$. Interestingly, the dependence on $\ve_1,\ve_2$ was not fully understood, even in this case, as most lower bounds dealt with \emph{estimation} of the distance between $p,q$ to an additive $\ve$, which is a related yet different problem (essentially, showing that $\Omega(n/(\ve_2-\ve_1)^2\log n)$ samples are required, when $\ve_1=\Theta(1)$ and $\ve_2-\ve_1$ can be arbitrarily small). The lower bound from~\cite{ValiantV10a} does imply, by ``scaling,'' an $\Omega(n/(\ve_2\log n))$ lower bound for arbitrary $\ve_2$ and $\ve_1=\Theta(\ve_1)$, but   it is still far from the upper bound of $\mathcal{O}(n/(\ve_2^2\log n))$ in this regime that both~\cite{ValiantV10b} and~\cite{JiaoHW18} prove in this setting. Our result shows that this upper bound is tight in this parameter regime, as our lower bound is then $\Omega(n/(\ve_2^2\log n))$.
\end{itemize}
We emphasize that our results go beyond those two extreme points, and essentially settles the landscape of tolerant testing. As just one example, the question of testing $1/n^{1/10}$-close vs. $1/n^{1/5}$-far was left completely open by previous work; our results imply that the sample complexity is $\tilde{\Theta}(n)$. We depict in~\autoref{fig:phasediagram} the different regimes of sample complexity this leads to, for both identity and closeness testing.

\begin{figure}[ht!]
    \centering
    \includegraphics[width=0.45\textwidth]{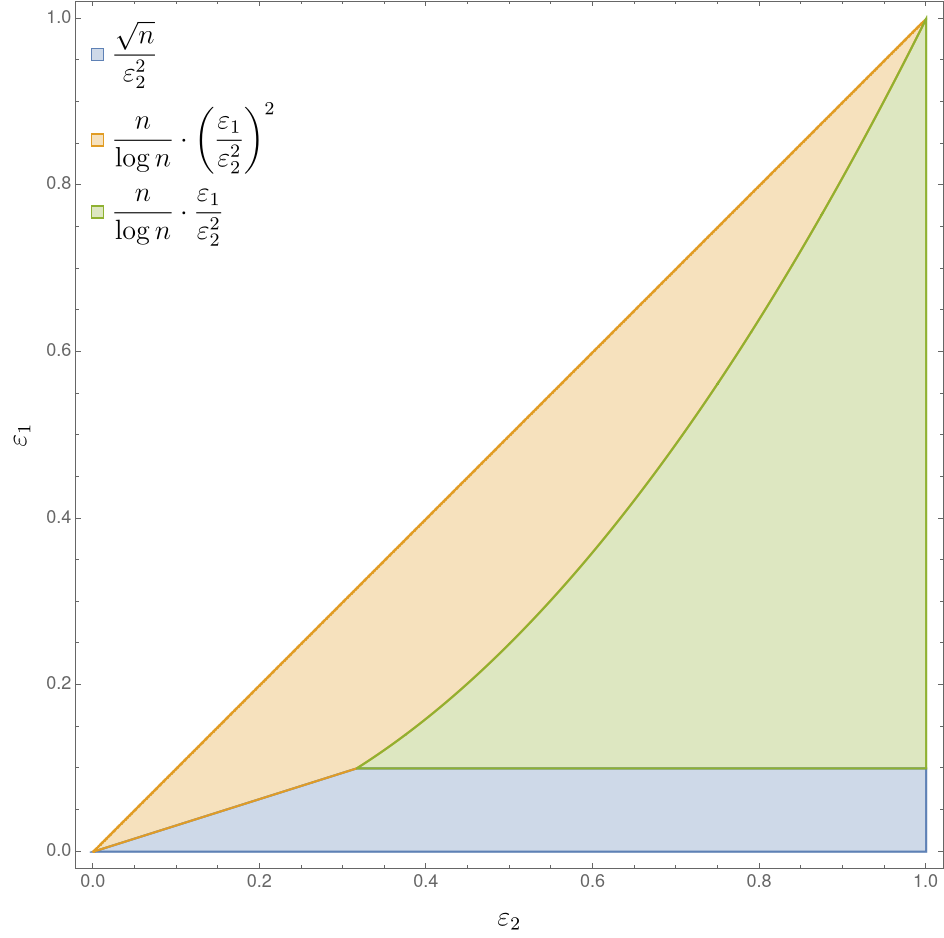} 
    \includegraphics[width=0.45\textwidth]{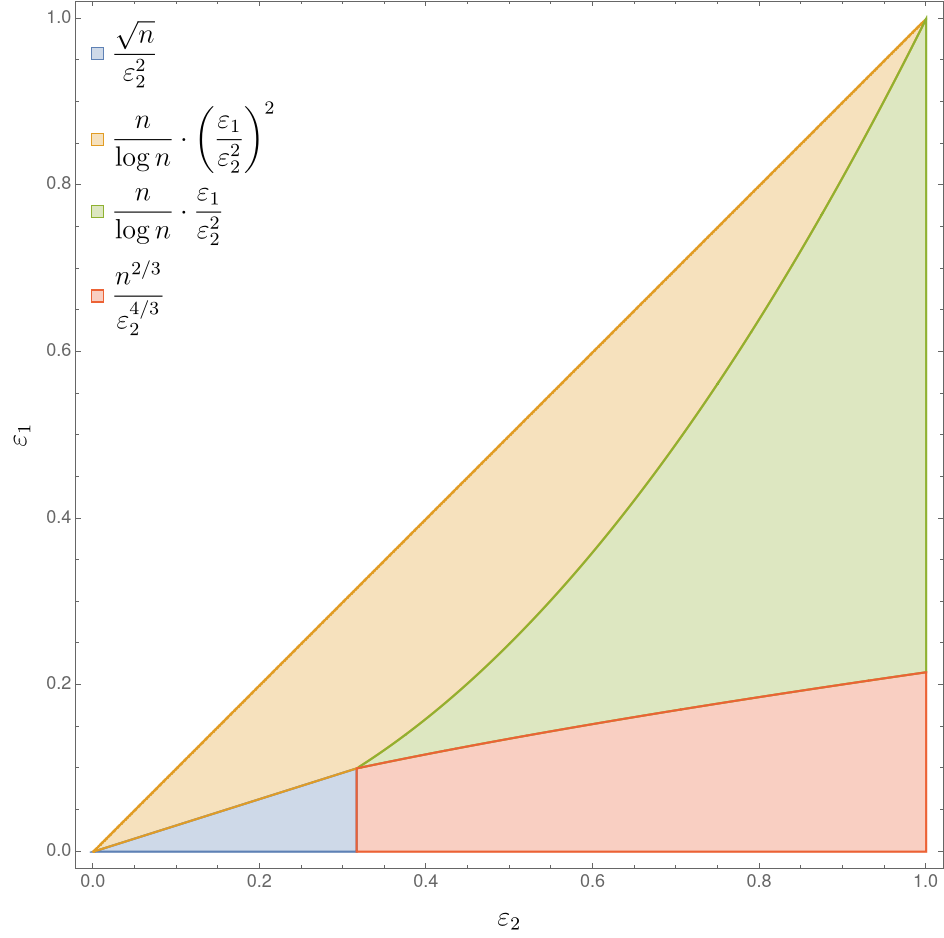}
    \caption{The different regimes of sample complexity corresponding to~\autoref{th:identity:informal} (identity testing, left) and~\autoref{th:closeness:informal} (closeness testing, right), as a function of $0\leq \eps_1<\eps_2 \leq 1$ (for fixed $n$), depicting in both cases which of the terms of the sample complexity bound dominates. }
    \label{fig:phasediagram}
\end{figure}

Surprisingly, our results for both tolerant identity and equivalence testing show that the relevant quantity governing the ``price of tolerance'' is not the ratio $\ve_1/\ve_2$, as one might na\"ively think; but instead is the (inhomogeneous!) ratio $\rho\eqdef \ve_1/\ve_2^2$, which might seem counterintuitive~--~especially in view of the two different regimes the $\max(\rho,\rho^2)$ scaling implies.

Another interesting and unexpected byproduct of our result is to show that even the known ``weak tolerance'' of the standard $\lp{2}$-based testers, which allow to test identity with tolerance $\ve_1=\ve_2/\sqrt{n}$ with the same $\mathcal{O}(\sqrt{n}/\ve_2^2)$ sample complexity as the non-tolerant case, is \emph{not} the best one can do with this sample complexity. Indeed, our results imply that one can actually achieve tolerance up to $\ve_1 = \min(1/\sqrt{n}, \ve_2/\sqrt[4]{n})$ ``for free,'' a significant improvement over $\ve_2/\sqrt{n}$. One can rephrase this as saying that the Cauchy--Schwarz inequality, from which this ``natural'' weak tolerance provided by $\lp{2}$-based testers stems from, is (oddly) not the right way to look at the problem.

Finally, our techniques allow us to derive an analogue of~\autoref{th:identity:informal} for the ``instance-optimal'' setting~\cite{ValiantV17a} (see also~\cite{BlaisCG17,DiakonikolasK16}), where the sample complexity is expressed as a function of the known reference distribution $q$ instead of the domain size $n$ (which corresponds to a worst-case over all possible reference distributions). Specifically, we show the following:\footnote{Here and in~\autoref{app:io}, we slightly abuse the $\tilde{\Theta}$ notation to also hide logarithmic factors in $n$, not just in the argument.}
\begin{theorem}[Instance-optimal identity testing (Informal; see Theorem~\ref{th:ubU:io} and Theorem~\ref{th:lbU:io})]\label{th:identity:io:informal}
For any fixed $q$ over $\N$, the sample complexity of tolerant identity testing with reference distribution $q$ with parameters $0\leq \ve_1 < \ve_2 \leq 1$ is
\[
\tilde{\Theta}\mleft(\frac{\|q_{-\Theta(\ve_2)}\|_{2/3}}{\ve_2^{2}} + \|q_{-\Theta(\ve_2)}\|_{1/2}\cdot\mleft(\frac{\ve_1}{\ve_2^2}\mright)^{\!\!2} + \|q_{-\Theta(\ve_2)}\|_0 \cdot\frac{\ve_1}{\ve_2^2} \mright)\,,
\]
where $q_{-\alpha}$ denotes the (sub)distribution obtained by removing as many of the smallest elements of $q$ as possible, without removing more than a total of $\alpha$ probability mass overall.
\end{theorem}
We defer the details and proof of this result to~\autoref{app:io}; and discuss some of its aspects here. First, note that by choosing $q$ to be the uniform distribution, we see that $\|q_{-\Theta(\ve_2)}\|_{2/3}\approx \sqrt{n}$, $\|q_{-\Theta(\ve_2)}\|_{1/2}\approx n$, and $\|q_{-\Theta(\ve_2)}\|_{0}\approx n$, so that~\autoref{th:identity:io:informal} retrieves~\autoref{th:identity:informal} up to logarithmic factors.
In particular, this gives a refined perspective  on~\autoref{th:identity:informal}, showing that the $\tilde \Theta(n)$ term actually arises due to two separate costs, which happen to coincide for the uniform distribution.
This brings us to our second point: the term $\|q_{-\Theta(\ve_2)}\|_{2/3}$ corresponds to the \emph{non-tolerant} instance-optimal identity testing bound established in~\cite{ValiantV17a}, i.e., a \emph{testing} term; while the quantity $\|q_{-\Theta(\ve_2)}\|_{1/2}$ can be interpreted as capturing the difficulty of \emph{learning}, as the $1/2$-quasinorm is known to capture the sample complexity of learning a probability distribution (see, e.g.,~\cite{KamathOPS15,CanonneLearningNote20}). Thus, the bound of~\autoref{th:identity:io:informal} can be read as saying the sample complexity of tolerant identity testing is (nearly) characterized by three aspects of the reference distribution: how hard it is to \emph{test}, how hard it is to \emph{learn}, and how large its \emph{effective support size} is.

\paragraph{Relation to the Statistics literature.} 
Despite their pervasive use, the Statistics community is outspoken about the pitfalls associated with point nulls (i.e., $\ve_1 = 0$) for statistical hypothesis testing~\cite{BergerS87, RaoL16, Abadie20}.
Instead, the community advocates for composite nulls, where the null hypothesis is a set of distributions rather than a single one. 
This more-general problem is often reduced to our tolerant testing problem (sometimes called the \emph{imprecise null} by the Statistics community) by assuming the null holds and performing estimation to obtain a candidate distribution $q$.
Thus, we believe our results may be a useful tool for solving more challenging composite-versus-composite hypothesis testing problems.
While some classic work provides minimax rates for certain related tolerant testing problems~\cite{Ingster00}, results in this direction have been relatively hard to come by.
In a recent survey paper~\cite{BalakrishnanW18}, Balakrishnan and Wasserman specifically highlight the problem of designing non-conservative thresholds for imprecise null hypothesis tests, which we believe to be an interesting direction for future work.

\paragraph{Overview of our techniques.}
Given the extensive literature on distribution testing, the community has developed a rich set of tools for problems in this space. 
However, the techniques used for the two extreme cases appear to be qualitatively quite different.
In the non-tolerant case, algorithms usually take the form of simple $\lp{2}$- or $\chi^2$-test statistics, and lower bounds are established via either Ingster's method~\cite{Ingster94} or mutual information arguments.
On the other hand, analysis for the maximally noisy case depends on results from the literature on best-polynomial approximation.
Given the contrasting approaches for these two cases, it is natural to wonder which set of techniques will be effective for the problems which lie between the two.
Interestingly, our results borrow from both: our algorithms are more similar to those from the non-tolerant setting, while our lower bound techniques resemble those in the maximally noisy case.

Our main algorithm thresholds a rescaled $\lp{2}$-statistic (in certain cases called a $\chi^2$-statistic), similar to testing algorithms in the past (see, e.g.,~\cite{ChanDVV14,ValiantV14,AcharyaDK15,DiakonikolasKN15a}).
Specifically, our statistic takes the form $Z = \sum_i \left((X_i - Y_i)^2 - X_i - Y_i\right)/\hat f_i$,
where $X_i$ and $Y_i$ are the number of occurrences of symbol $i$ drawn from distribution $p$ and $q$, respectively, and $\hat f_i$ are symbol-dependent rescaling factors. 
While prior works either computed these factors for identity testing based on the reference distribution $q$, or used the same set of samples for both the numerator as well as the rescaling factor in the denominator, we use sample-splitting to separately obtain empirical estimates $\hat f_i$ for the relevant quantities.
We show multiplicative concentration for these factors to ensure they are close to the values for which we are using them as a proxy.
These rescaling factors are empirical estimates of two terms. A typical choice, now common in the literature, is based on $p_i + q_i$, which limits fluctuations in the estimator caused by individual terms. 
Our approach crucially introduces an additional novel rescaling term based on  $|p_i - q_i|$, which prevents the statistic from placing too much emphasis on the the $\lp{2}$-norm of the distribution.
The contributions of both terms are crucial for making the analysis work out.

We note that our test statistic only involves the first two moments of the distribution.
This is in contrast to previous upper bounds for tolerant testing in the maximally noisy case, which instead inspected $\log n$ moments.
Thus, we show that considering only two moments suffices for near-optimal tolerant testing.
Interestingly, our algorithm achieves the optimal sample complexity (up to constants) for the non-tolerant case, but loses a $\log n$ factor in the maximally noisy case.
Removing this final logarithmic term may require a statistic which exploits higher-order moments, and is an interesting question for future work.\smallskip

Our lower bounds are obtained via the generalized two-point method.
At a high level, we follow the moment matching approach pioneered by~\cite{WuY16}. We construct two priors over distributions, where distributions drawn from the two priors are $\ve_1$-close to and $\ve_2$-far from uniform, respectively.
To prove lower bounds, we must choose these priors such that the process of drawing a distribution and then $m$ samples from it has low total variation distance between the two priors.
By considering priors over product distributions, we can further reduce our task to simply constructing a pair of univariate random variables with properties described in Theorem~\ref{thm:lb-main}.
By appealing to results from polynomial approximation (Lemma~\ref{lem:mm}), it suffices to construct this pair such that their low-order moments match.

Prior works construct this pair of random variables by expressing this moment matching problem as an infinite dimensional convex program and analyzing its dual.
Our approach follows the same recipe, however, the analysis of the dual convex program is much more involved in our case.
Prior lower bounds only considered the special case where $\eps_1$ and $\eps_2$ differ by a fixed, constant factor.
In this regime, tolerant testing becomes essentially equivalent to learning the $\ell_1$ distance between $p$ to $q$ to error $\eps_1$.
This is a setting which is much easier for this formalism to handle; indeed, the moment matching paradigm was initially designed for estimation problems.
Importantly, this induces an key symmetry in the lower bound construction, and consequently, the dual has a very nice interpretation in terms of the best uniform approximation of a given function by a low-degree polynomial.

In our case we must handle general $\eps_1$ and $\eps_2$, and this symmetry is lost. 
As a result, we analyze a convex program which directly captures the testing problem.
However, the dual has a much more complex interpretation.
At a high level, the goal is now to approximately fit a low-degree polynomial within a ``wedge'' of minimal arc length.
Interestingly, in our formulation of the dual, instead of having to prove that there is a good approximating polynomial, we must demonstrate that no low degree polynomial can achieve this task.
This is the main technical difficulty in the lower bound, and we do so from first principles by leveraging classic tools from polynomial approximation theory to prove new approximation-theoretic results in our setting.

\subsection{Related Work}
Distribution testing was first considered in the theoretical computer science community by Goldreich and Ron~\cite{GoldreichR00}, who analyzed and applied an algorithm for uniformity testing towards the problem of testing whether a graph is an expander. Batu, Fischer, Fortnow, Kumar, Rubinfeld, and White~\cite{BatuFFKRW01} studied the general problem of identity testing. 
A number of results have discovered and rediscovered optimal bounds for identity testing~\cite{Paninski08,ValiantV14,AcharyaDK15,DiakonikolasKN15a,Goldreich16,DiakonikolasK16,DiakonikolasGPP18,DaskalakisKW18,DiakonikolasGPP19}, even with optimal dependence on the failure probability $\delta$ and on an instance-by-instance basis.
The harder problem of equivalence testing was studied in~\cite{BatuFRSW00}, and optimal upper and lower bounds were given in~\cite{Valiant11,ChanDVV14,DaskalakisKW18,DiakonikolasGKPP21}. Some work has also studied the case where an unequal number of samples are received from the two distributions~\cite{AcharyaJOS14c,BhattacharyaV15,DiakonikolasK16}.

Tolerant testing has been previously considered, in a few different regimes.
Strong tolerance, or equivalently, estimating distance between distributions, was studied first by Valiant and Valiant~\cite{ValiantV10a,ValiantV10b,ValiantV11a,ValiantV11b}, and in more recent works by Han, Jiao, Venkat, and Weissman~\cite{JiaoHW18,JiaoVHW17,HanJW16}. 
Tolerance in distances besides $\lp{1}$ (including chi-squared, KL, Hellinger, and $\lp{2}$) has also been considered~\cite{GoldreichR00,BatuFRSW00,ChanDVV14,AcharyaDK15,DaskalakisKW18}.
An interesting direction for future work is to understand the sample complexity of tolerant testing for these other distances in a fine-grained manner, as we do for $\ell_1$ distance.
Moreover, results with $\lp{2}$-tolerance imply testers with weak $\lp{1}$-tolerance, through the relation between $\lp{1}$ and $\lp{2}$ norms and the Cauchy--Schwarz inequality.  Finally, very recent work sets out to understand whether, for general properties of distributions, the (near)-quadratic gap between tolerant and non-tolerant testing achievable for identity testing is the worst possible~\cite{ChakrabortyFGMS21}.
For additional background on distribution testing, see surveys and related work in~\cite{Rubinfeld12, BalakrishnanW18,Kamath18,Canonne20}.

Techniques involving moment matching and best-polynomial approximation are useful for tolerant distribution testing, but also play a key role in estimation of distributional properties, including entropy, support size, support coverage, and distance to uniformity~\cite{WuY16,JiaoHW18,AcharyaDOS17,OrlitskySW16,WuY18}
See~\cite{WuY20} for a survey on applications of polynomial methods in statistics.

\subsection{Preliminaries}
We identify a probability distribution $p$ over a known discrete domain $[n]\eqdef \{1,2,\dots,n\}$ with its probability mass function (pmf), i.e., a nonnegative vector $p =(p_1,p_2,\dots,p_n)$ such that $\sum_{i=1}^n p_i = 1$. Given two distributions $p,q$, their total variation distance (also known as statistical distance) is defined as
\[
    \operatorname{TV}(p,q) = \sup_{S\subseteq [n]} (p(S)-q(S)) = \frac{1}{2} \sum_{i=1}^n |p_i-q_i| = \frac{1}{2}\norm{p-q}_1\,.
\]
Due to this equivalence with the $\lp{1}$ norm, we will interchangeably use the TV and $\lp{1}$ norms in our paper. We will also extensively use the $\lp{2}$ distance between probability distributions, which is just the $\lp{2}$ norm $\norm{p-q}_2$ between their pmfs and, by Cauchy--Schwarz, satisfies $\frac{1}{\sqrt{n}}\norm{p-q}_1 \leq \norm{p-q}_2 \leq \norm{p-q}_1$. 

Let $p,q$ be two distributions over the domain $[n]$. For given $\ve_1$ and $\ve_2$ such that $0\leq \ve_1 < \ve_2$, we want to understand the sample complexity (i.e., minimum number of i.i.d.\ samples required) to distinguish between:
\begin{description}
    \item[Yes:]  $\norm{p - q}_1 \leq \ve_1$,
    \item[No:] $\norm{p - q}_1 > \ve_2$. 
\end{description}
with probability at least $4/5$.\footnote{The exact constant here is immaterial, and by standard amplification arguments one can achieve a probability of success of $1 - \delta$ at the cost of a multiplicative $\mathcal{O}(\log (1/\delta))$ factor in the sample complexity.}

We consider the problems of tolerant uniformity, identity, and equivalence testing.
In \emph{identity testing} the distribution $q$ is explicitly known in advance, while $p$ is unknown: the sample complexity is then the number of i.i.d.\ samples from $p$. \emph{Uniformity testing} is a special case of identity testing, where $q=(\frac{1}{n},\frac{1}{n},\dots,\frac{1}{n})$ is the uniform distribution, denoted $\unif_n$. 
In \emph{equivalence testing}, both $p$ and $q$ are unknown and we get samples from each. The sample complexity is then the total number of samples obtained from both $p$ and $q$. We will typically denote the number of i.i.d.\ samples used by an algorithm by $m$.
 Note that uniformity testing is a special case of identity testing (and hence lower bounds for the former imply lower bounds for the latter), and that equivalence testing is at least as hard as identity testing, in terms of sample complexity.

\section{Algorithms for Tolerant Testing}

In this section, we describe our testing algorithm (Algorithm~\ref{algo:tester}), before analyzing its performance. As a preliminary simplification, instead of assuming the algorithm is provided with $m$ independent samples we will rely on the so-called ``Poissonization trick'' and assume we obtain $\operatorname{Poi}(m)$ samples each from both $p$ and $q$. 
The benefit of Poissonization is that the number of occurrences of each domain element will be an independent Poisson, eliminating correlations between symbols which arise with a fixed budget.
This is without loss of generality, as by standard arguments about concentration of Poisson random variables this changes the sample complexity by at most a (small) constant factor. Moreover, losing again a factor 2 in the sample complexity, our algorithms will take as input two sets of $\operatorname{Poi}(m)$ samples for each of $p$ and $q$.

Let $\tilde{X}_i$ and $X_i$ be the count of occurrences of symbol $i\in[n]$ in the first and the second set of the samples from $p$, respectively.
Similarly, let $\tilde{Y}_i$ and $Y_i$ be the count of symbol $i$ in the first and the second set of the samples from $q$, respectively.
Let 
\[
f_i \eqdef
\begin{cases}
\max\{ \sqrt{mn}\cdot |p_i-q_i|, n\cdot  (p_i+q_i), 1 \} & \text{if $m\geq n$}\\
\max\{ m\cdot  (p_i+q_i), 1 \} & \text{if $m < n$}.
\end{cases}
\]
We will use the first set of counts $\tilde{X}_i$ and $\tilde{Y}_i$ to estimate $f_i$ with $\widehat{f}_i$, defined as
\[
  \widehat{f}_i \eqdef
  \begin{cases}
\max\mleft\{\frac{|\tilde{X}_i-\tilde{Y}_i|}{\sqrt{m/n}},\frac{\tilde{X}_i+\tilde{Y}_i}{m/n}, 1 \mright\} & \text{if $m\geq n$}\\
\max\{\tilde{X}_i+\tilde{Y}_i, 1 \} & \text{if $m < n$}.
\end{cases}
\]
Let $Z_i \eqdef (X_i-Y_i)^2-X_i-Y_i$ and let
\begin{equation}
  \label{eq:defz}
  Z\eqdef \sum_{i=1}^n \frac{Z_i}{\widehat{f}_i}.
\end{equation}
and $\tau \eqdef c\cdot\min\mleft(\frac{m^{3/2}{\ve_2}}{n^{1/2}}, \frac{m^2\ve_2^2}{n}\mright)$, 
where $c>0$ is an absolute constant determined in the course of the analysis.
Our tester is then as follows:
\begin{algorithm}[H]
  \begin{algorithmic}
    \Require{$0 \leq \ve_1 < \ve_2 \leq 1$, $m$, $n$, two sets of $\operatorname{Poi}(m)$ samples from both $p$ and $q$}
    \State Set the threshold
      \[
          \tau \gets c\cdot\min\mleft(\frac{m^{3/2}{\ve_2}}{n^{1/2}}, \frac{m^2\ve_2^2}{n}\mright) 
      \]
    \State Compute 
    $
        Z
    $
    from the sets of samples, as per~\eqref{eq:defz}.
    \If{$Z \geq \tau$} 
      \Return $\norm{p-q}_1\geq \ve_2$
    \Else 
      \ \Return $\norm{p-q}_1\leq \ve_1$
    \EndIf
  \end{algorithmic}
  \caption{\label{algo:tester}Tolerant testing algorithm.}
\end{algorithm}
\noindent Note that the algorithm itself requires no knowledge of $\ve_1$, and thus as the number of samples $m$ increases, the same test statistic (with appropriate substitution of $m$) becomes more and more tolerant.

To gain some intuition, we first remark that our tester is a modification of the $\lp{2}$ testers in~\cite{ChanDVV14,DiakonikolasK16}, and akin to the chi-square tester from~\cite{AcharyaDK15}. The main difference lies in the choice of normalizing factor $f_i$ (of which $\widehat{f}_i$ is merely the natural estimator). The goal of this denominator is twofold: the relatively standard term $n\cdot  (p_i+q_i)$ (which is comparable to the standard deviation of $Z_i$; for $m<n$, we use $m(p_i+q_i)$ to make up for larger imprecision in our estimates) ensures that no single term of the sum will make the estimator fluctuate too much. The term $\sqrt{mn}\cdot |p_i-q_i|$  (which is only needed the regime $m\geq n$, since for $m<n$ the best accuracy we can get for $|p_i-q_i|$ is $\approx 1/m$, but scaling by $m|p_i-q_i|$ would be unnecessary as the $m(p_i+q_i)$ term already dominates) is a crucial difference with previous work; its goal is to ``tamper down'' the numerator $Z_i$ when the $\lp{2}$ contribution $(p_i-q_i)^2$ is too large, which is key for our $\lp{2}$-based tester to work. Indeed, in the ``far'' case where $\norm{p-q}_1\geq \ve_2$, this is not a problem; however, in the ``close'' case where $\norm{p-q}_1\leq \ve_1$, the relation between $\lp{2}$ and $\lp{1}$ does not preclude an individual element to have a large contribution $(p_i-q_i)^2$, which could cause the statistic to be too large and the tester to incorrectly reject. To avoid this, the term $\sqrt{mn}\cdot |p_i-q_i|$ in the denominator will ``kick in'' for any such element $i$, and make the ratio $Z_i/f_i$ behave proportionally to $|p_i-q_i|/\sqrt{mn}$ instead of $(p_i-q_i)^2$, ensuring that the algorithm does not mistakenly reject ``close'' distributions due to any single large element contribution. %

\begin{remark}
\label{rem:split}
For the identity testing problem, where the reference distribution $q$ is known, we use the now-standard ``splitting operation'' of Diakonikolas and Kane~\cite{DiakonikolasK16} (see Section~\ref{app:splitting} for details) to obtain distributions $p'$ and $q'$ over a domain of size $2n$ such that $\norm{p'-q'}_1 = \norm{p-q}_1$ and $\norm{q'}_2\leq 1/\sqrt{n}$. Moreover, samples from $p$ and $q$ can be used to simulate the same number of samples from distributions $p'$ and $q'$, respectively. 
We apply our tester on the modified distributions $p'$ and $q'$, instead of using it for $p$ and $q$ directly. 
As the new reference distribution $q'$ is over a domain of size $2n$ and satisfies $\norm{q'}_2\leq {\sqrt{2}}/{\sqrt{2n}}$, 
this transformation lets us assume without loss of generality that the reference distribution $q$ over $[n]$ in the identity testing problem is such that $\norm{q}_2\leq \sqrt{2/n}$.
\end{remark}\medskip

We now formally state the performance of Algorithm~\ref{algo:tester} for tolerant identity and equivalence testing, i.e., that it achieves near-optimal sample complexity in both cases.
\begin{theorem}[Identity testing]\label{th:ubU}
Let $q$ be a known reference distribution and $p$ be an unknown distributions, both over $[n]$. There exists an absolute constant $c >0$ such that, for any $0 \leq  \ve_2\leq 1$ and $0\leq \ve_1 \leq c \ve_2$, given
\[
      \mathcal{O}\mleft( n\Big(\frac{\ve_1}{\ve_2^2}\Big)^2+n\Big(\frac{\ve_1}{\ve_2^2}\Big)+\frac{\sqrt n}{\ve_2^{2}} \mright)
\]
samples from each of $p$ and $q$ Algorithm~\ref{algo:tester} (after the splitting operation of Remark~\ref{rem:split}) distinguishes between $\norm{p-q}_1\leq \ve_1$ and $\norm{p-q}_1\geq \ve_2$ with probability at least $4/5$.
\end{theorem}
\begin{theorem}[Equivalence testing]\label{th:ubC}
Let $p$ and $q$ be two unknown distributions over $[n]$. There exists an absolute constant $c >0$ such that, for any $0 \leq  \ve_2\leq 1$ and $0\leq \ve_1 \leq c \ve_2$, given
\[
      \mathcal{O}\mleft( n\Big(\frac{\ve_1}{\ve_2^2}\Big)^2+n\Big(\frac{\ve_1}{\ve_2^2}\Big)+\frac{\sqrt n}{\ve_2^{2}} +\frac{n^{2/3}}{\ve_2^{4/3}} \mright)
\]
samples from each of $p$ and $q$ Algorithm~\ref{algo:tester} distinguishes between $\norm{p-q}_1\leq \ve_1$ and $\norm{p-q}_1\geq \ve_2$ with probability at least $4/5$.
\end{theorem}
Note that for a unified exposition, we assumed in Theorem~\ref{th:ubU} that the algorithm is provided with samples even from the explicitly known reference distribution $q$. This is not a restriction, as given this explicit knowledge it is possible to efficiently sample from the distribution $q$.

\subsection{Analysis of Algorithm~\ref{algo:tester}}
This section is devoted to the proofs of Theorems~\ref{th:ubU} and~\ref{th:ubC}, which are both established in a similar manner.

Observe that, following the Poissonization, all $Z_i$'s and $\widehat{f}_i$'s are independent random variables. From the properties of Poisson distribution, it is not hard to check that the expectation and variance of the $Z_i$'s are given by
\begin{align}
\E[ Z_i ] &=  m^2|p_i-q_i|^2, \label{eq:expzi}\\
\Var [ Z_i ] &= 4m^3(p_i-q_i)^2(p_i+q_i)+2m^2(p_i+q_i)^2. \label{eq:varzi}
\end{align}

\noindent Next, using the independence of $Z_i$ and $\widehat{f}_i$'s, we get that the conditional expectation of $Z$ is
\begin{align}
\E\mleft[ Z\;\middle|\;\widehat{f}_i \text{ for }i\in[n] \mright]  = \E\mleft[ \sum_{i=1}^n \frac{Z_i}{\widehat{f}_i}\;\middle|\;\widehat{f}_i \text{ for }i\in[n] \mright]  = \sum_{i=1}^n \frac{\E[ Z_i ]}{\widehat{f}_i},\label{eq:mean}
\end{align}
while its conditional variance is given by
\begin{align}
\Var\mleft[ Z \big|\widehat{f}_i \text{ for }i\in[n] \mright]  = \Var\mleft[ \sum_{i=1}^n \frac{Z_i}{\widehat{f}_i}\bigg|\widehat{f}_i \text{ for }i\in[n] \mright]  = \sum_{i=1}^n \frac{\Var(Z_i)}{\widehat{f}_i^2}.\label{eq:var}
\end{align}

To prove the optimality of the tester, we first bound the conditional expectation and variance of $Z$. These bounds differ for the regimes $m\geq n$ and $m\leq n$, and are characterized in Lemmas~\ref{lem:mainmgn} and~\ref{lem:mainmgn2}, respectively.
\begin{lemma}\label{lem:mainmgn}
There exist absolute constants $c_1,c_2,c_3>0$ such that the following holds. For $m\geq n$, and any distributions $p$ and $q$ over $[n]$, the following bounds simultanously hold with probability at least $9/10$:
\[
c_1\min\mleft(\frac{m^{3/2}{\norm{p-q}_1}}{n^{1/2}}, \frac{m^2{\norm{p-q}_1}^2}{n}\mright)
\leq 
\E\mleft[ Z\;\middle|\;\widehat{f}_i \text{ for }i\in[n] \mright]
\leq 
c_2\frac{m^{3/2}\norm{p-q}_1}{n^{1/2}},
\]
and
$
\Var\mleft[ Z \;\middle|\;\widehat{f}_i \text{ for }i\in[n] \mright] \leq 
c_3\frac{m^2}{n}.
$
\end{lemma}
\begin{lemma}\label{lem:mainmgn2}
There exist absolute constants $c_1,c_2,c_3>0$ such that the following holds. For $m\leq n$, and any distributions $p$ and $q$ over $[n]$, the following bounds simultanously hold with probability at least $9/10$:
\[
c_1\frac{m^2{\norm{p-q}_1}^2}{n}
\leq 
\E\mleft[ Z\;\middle|\;\widehat{f}_i \text{ for }i\in[n] \mright]
\leq 
c_2m\norm{p-q}_1,
\]
and
$
\Var\mleft[ Z \;\middle|\;\widehat{f}_i \text{ for }i\in[n] \mright] \leq 
c_3 m.
$
Additionally, 
\[
\Var\mleft[ Z \;\middle|\;\widehat{f}_i \text{ for }i\in[n] \mright]\le \frac{1}{40}( \E[ Z\;|\;\widehat{f}_i \text{ for }i\in[n] ] )^2 + 324\E[ Z\;|\;\widehat{f}_i \text{ for }i\in[n] ]+648 m^2\norm{q}_2^2.
\]
\end{lemma}
We prove these lemmata in Section~\ref{sec:cond-proofs}.
We now show that, assuming these statements, we can establish Theorems~\ref{th:ubU} and~\ref{th:ubC}.
We handle the cases $m\geq n$ and $m<n$ separately.

\paragraph{Proof of the theorems for $m\geq n$:} Using Lemma~\ref{lem:mainmgn}, we show that for any $m \geq n$ such that $m =  \Omega\left(n\left(\frac{\ve_1}{\ve_2^2}\right)^2+\frac{\sqrt n}{\ve_2^{2}}\right)$ the estimator correctly distinguishes between $\norm{p-q}\leq \ve_1$ vs $\norm{p-q}\geq \ve_2$ with probability at least $8/10$.

Applying Chebyshev’s inequality to the conditional expectation and variance and using Lemma~\ref{lem:mainmgn}, we get that, with probability $\geq 8/10$,
\begin{equation}\label{eq:ell2lb}
 Z  \geq c_1\min\mleft(\frac{m^{3/2}{\norm{p-q}_1}}{n^{1/2}}, \frac{m^2{\norm{p-q}_1}^2}{n}\mright)- \sqrt{10c_3}\frac{m}{\sqrt n},   
\end{equation}
and
\begin{equation}\label{eq:ell2ub}
Z  \leq c_2\frac{m^{3/2}\norm{p-q}_1}{n^{1/2}} + \sqrt{10c_3}\frac{m}{\sqrt n}.
\end{equation}

\begin{itemize}
  \item In the case $\norm{p-q}_1\geq \ve_2$, the lower bound in Equation~\eqref{eq:ell2lb} reduces to
\begin{align*}
 Z &\geq c_1\min\mleft(\frac{m^{3/2}{\ve_2}}{n^{1/2}}, \frac{m^2\ve_2^2}{n}\mright) - \sqrt{10c_3}\frac{m}{\sqrt n}
    \geq
 \frac{c_1}{2}\min\mleft(\frac{m^{3/2}{\ve_2}}{n^{1/2}}, \frac{m^2\ve_2^2}{n}\mright),
\end{align*}
the last step as long as $m \geq C\sqrt{n}/\ve_2^2$ for $C \eqdef \max(2\sqrt{10c_3}/c_1, 40c_3/c_1^2)$.
Therefore, with probability at least $8/10$ the tester correctly outputs that $\norm{p-q}_1\geq \ve_2$. 

\item In the case, $\norm{p-q}_1\leq \ve_1$, the upper bound in Equation~\eqref{eq:ell2ub} reduces to
\[
Z \leq  c_2\frac{m^{3/2}\ve_1}{n^{1/2}}+\sqrt{10c_3}\frac{m}{\sqrt n} \leq \frac{c_1}{4}\min\mleft(\frac{m^{3/2}{\ve_2}}{n^{1/2}}, \frac{m^2\ve_2^2}{n}\mright)
\]
where we used that $m \geq C'\sqrt{n}/\ve_2^2$ for $C \eqdef \max(\frac{8\sqrt{10c_3}}{c_1}, \frac{640c_3}{c_1^2})$
to ensure that $\sqrt{10c_3}\frac{m}{\sqrt n} \leq \frac{c_1}{8}\min\mleft(\frac{m^{3/2}{\ve_2}}{n^{1/2}}, \frac{m^2\ve_2^2}{n}\mright)$, 
and that (i)~$\ve_1 \leq \frac{c_1}{8c_2}\ve_2$ and (ii)~$m \geq C'' \cdot n\Big(\frac{\ve_1}{\ve_2^2}\Big)^2$ with $C' = 64c_2^2/c_1^2$ 
to ensure that $c_2\frac{m^{3/2}\ve_1}{n^{1/2}} \leq \frac{c_1}{8}\min\mleft(\frac{m^{3/2}{\ve_2}}{n^{1/2}}, \frac{m^2\ve_2^2}{n}\mright)$.
Therefore, with probability at least $8/10$ the tester correctly outputs that $\norm{p-q}_1\leq \ve_1$. 
\end{itemize}
This proves the two theorems for the case $m\geq n$. We next turn to the case $m\leq n$.

\paragraph{Proof of the theorems for $m< n$:} The argument for this case is similar to the previous, using Lemma~\ref{lem:mainmgn2} instead of Lemma~\ref{lem:mainmgn}. We show that for any $m \leq n$ such that $m =  \Omega\left(n\frac{\ve_1}{\ve_2^2}+\min\left\{\frac{n\norm{q}_2}{\ve_2^{2}},\frac{n^{2/3}}{\ve_2^{4/3}}\right\}\right)$ the estimator correctly distinguishes between $\norm{p-q}\leq \ve_1$ and $\norm{p-q}\geq \ve_2$ with probability at least $8/10$. This in turn follows from computations nearly identical to the ones above, which we thus omit in the interest of space. 

The proofs for the two cases, combined with the fact that for identity testing we can as discussed before assume without loss of generality that $\norm{q}_2 \leq \sqrt{2/n}$, establish Theorems~\ref{th:ubU} and~\ref{th:ubC}. \qed

\subsection{Proof of Lemmas~\ref{lem:mainmgn} and~\ref{lem:mainmgn2}}
\label{sec:cond-proofs}

In this section, we give the proof of the remaining two pieces in our analysis of Algorithm~\ref{algo:tester}, Lemmas~\ref{lem:mainmgn} and~\ref{lem:mainmgn2}. 
The following lemma will be useful to lower bound the conditional expectation $\E\mleft[ Z\;\middle|\;\widehat{f}_i \text{ for }i\in[n] \mright]$. 
\begin{lemma}\label{cor:ficon5}
There exist absolute constants $c_1,c_2,c_3>0$ such that, for every $m$, $n$, and $i\in [n]$,
\[
\E[\widehat{f}_{i}] \leq c_1 f_i, 
\qquad 
\E[\widehat{f}_{i}^{-1}] \leq \frac{c_2}{f_i}, 
\quad \text{and }
\E[\widehat{f}_{i}^{-2}] \leq \frac{c_3}{f_i^2}.
\]
\end{lemma}

\begin{proof}
We use the following two concentration bounds, which provide exponential tail bounds on our estimates $\widehat{f}_i$ of the of $f_i$'s. The proofs of those two claims are quite technical, and rely on a careful case distinction along with standard concentration properties of Poisson random variables. We provide them in Section~\ref{app:conlem}.
\begin{lemma}\label{lem:ficon}
There exists $c>0$ such that, for every $m$, $n$, $t > 3$, and $i\in [n]$,
$
\Pr[ \widehat{f}_{i} > t f_i ]\leq e^{-ct}.
$
\end{lemma}

\begin{lemma}\label{lem:ficon3}
There exists $c'>0$ such that, for every $m$, $n$, $t > 2$, and $i\in [n]$,
$
\Pr[ \widehat{f}_{i} < \frac{f_i}{t} ]\leq e^{-c't}.
$
\end{lemma}

Given the above two results, the proof is straightforward: indeed, for every $i\in[n]$ we have, using Lemma~\ref{lem:ficon},
\[
    \E[\widehat{f}_{i}] = \int_0^\infty \Pr[ \widehat{f}_{i} > u ]du
    = f_i\int_0^\infty \Pr[ \widehat{f}_{i} > t f_i ]dt
    \leq f_i\mleft( \int_0^3 dt + \int_3^\infty e^{-ct}dt \mright)
    = f_i\mleft( 3 + \frac{e^{-3c}}{c} \mright)
\]
while, from Lemma~\ref{lem:ficon3},
\[
    \E[\widehat{f}_{i}^{-1}] = \int_0^\infty \Pr[ \widehat{f}_{i}^{-1} > u ]du
    = \frac{1}{f_i}\int_0^\infty \Pr[ \widehat{f}_{i} < f_i/t ]dt
    \leq \frac{1}{f_i}\mleft( 2 + \int_2^\infty e^{-c't}dt \mright)
    = \frac{1}{f_i}\mleft( 2 + \frac{e^{-2c'}}{c'} \mright)
\]
and, similarly,
\[
    \E[\widehat{f}_{i}^{-2}] = \int_0^\infty \Pr[ \widehat{f}_{i}^{-2} > u ]du
    = \frac{2}{f_i^2}\int_0^\infty \Pr[ \widehat{f}_{i} < f_i/t ]t\,dt
    \leq \frac{2}{f_i^2}\mleft( 2 + \int_2^\infty te^{-c't}dt \mright)
    = \frac{2}{f_i^2}\mleft( 2 + \frac{(2c'+1)e^{-2c'}}{c'^2} \mright)
\]
which, given that $c,c'$ are just positive constants, is what we set out to prove.
\end{proof}

We also require the following simple inequality.

\begin{fact}
For any real $(a_i)_{i=1}^n$ and positive $(b_i)_{i=1}^n$, 
\[
\sum_{i=1}^n \frac{a_i^2}{b_i}\geq \frac{(\sum_{i=1}^n|a_i|)^2}{\sum_{i=1}^n b_i}.
\]
\end{fact}
\begin{proof}
The result follows from applying Cauchy--Schwarz to $\sum_{i=1}^n|a_i| = \sum_{i=1}^n\sqrt{b_i}|a_i|/\sqrt{b_i}$.
\end{proof}
From the above fact and~\eqref{eq:expzi}, it follows that
\[
\E\mleft[ Z\;\middle|\;\widehat{f}_i \text{ for }i\in[n] \mright] = \sum_{i=1}^n \frac{\E[ Z_i ]}{\widehat{f}_i} = \sum_{i=1}^n \frac{m^2(p_i-q_i)^2}{\widehat{f}_i} \geq  \frac{m^2(\sum_{i=1}^n|p_i-q_i|)^2}{\sum_{i=1}^n\widehat{f}_i} =\frac{m^2\norm{p-q}_1^2}{\sum_{i=1}^n\widehat{f}_i}.
\]
Moreover, by definition the random variables $\widehat{f}_i$ are non-negative, and thus, applying the Markov inequality we get that
\[
\sum_{i=1}^n \widehat{f}_i \leq 30 \sum_{i=1}^n \E[\widehat{f}_i]
\]
with probability at least $1-1/30$. Combined with Lemma~\ref{cor:ficon5}, this means that, 
with probability at least $1-1/30$,
\begin{equation}
  \label{eq:markov:1}
 \E\mleft[ Z\;\middle|\;\widehat{f}_i \text{ for }i\in[n] \mright] \geq \frac{m^2\norm{p-q}_1^2}{30c_1\sum_{i=1}^n f_i}.
\end{equation}
Next, applying the Markov's inequality for the non-negative random variable $\E\mleft[ Z\;\middle|\;\widehat{f}_i \text{ for }i\in[n] \mright]$, we get that, with probability at least $1-1/30$,
\begin{equation}
  \label{eq:markov:2}
    \E\mleft[ Z\;\middle|\;\widehat{f}_i \text{ for }i\in[n] \mright] 
    \leq 30 \E\mleft[ \E\mleft[ Z\;\middle|\;\widehat{f}_i \text{ for }i\in[n] \mright]\mright] 
    = 30 \E\mleft[\sum_{i=1}^n \frac{\E[Z_i]}{\widehat{f}_i}\mright]
    \leq 30c_2\sum_{i=1}^n \frac{m^2(p_i-q_i)^2}{f_i}.
\end{equation}
Finally, considering the non-negative random variable $\Var\mleft[ Z\;\middle|\;\widehat{f}_i \text{ for }i\in[n] \mright]$, we again get that, with probability at least $1-1/30$,
\begin{equation}
  \label{eq:markov:3}
    \Var\mleft[ Z\;\middle|\;\widehat{f}_i \text{ for }i\in[n] \mright] 
    \leq 30 \E\mleft[ \Var \mleft[ Z\;\middle|\;\widehat{f}_i \text{ for }i\in[n] \mright]\mright] 
    = 30 \E\mleft[\sum_{i=1}^n \frac{\Var(Z_i)}{\widehat{f}_i^2}\mright]
    \leq 30c_3 \sum_{i=1}^n \frac{\Var(Z_i)}{f_i^2}.
\end{equation}
By a union bound, we get that the guarantees of~\eqref{eq:markov:1},~\eqref{eq:markov:2}, and~\eqref{eq:markov:3} simultaneously hold with probability at least $1-3\cdot\frac{1}{30}=\frac{9}{10}$.
Importantly, the RHS in all three bounds only depend on the deterministic quantities $f_i$'s, instead of the random variables $\widehat{f}_i$'s.
We bound each of these RHS in the next two lemmas, for $m\geq n$ and $m\leq n$, respectively.
\begin{lemma}
For any $m\geq n$ and distributions $p$ and $q$ over $[n]$ the following holds:
(1)~
$
\sum_{i=1}^n \frac{\Var(Z_i)}{{f}_i^2} \leq \frac{10m^2}{n},
$
(2)~
$
\sum_{i=1}^n \frac{m^2(p_i-q_i)^2}{{f}_i} \leq \frac{m^{3/2}\norm{p-q}_1}{n^{1/2}},
$ 
and
(3)~
$
\frac{m^2\norm{p-q}_1^2}{\sum_{i=1}^n{f}_i}\geq  \min\Big(\frac{m^{3/2}{\norm{p-q}_1}}{2n^{1/2}}, \frac{m^2{\norm{p-q}_1}^2}{6n}\Big).
$
\end{lemma}
\begin{proof}
First, we upper bound $\sum_{i=1}^n \frac{\Var(Z_i)}{ f_i^2}$: from~\eqref{eq:varzi}, we get
\begin{align}
\sum_{i=1}^n \frac{\Var(Z_i)}{ f_i^2}&=\sum_{i=1}^n \frac{4m^3(p_i-q_i)^2(p_i+q_i)+2m^2(p_i+q_i)^2}{f_i^2} \nonumber\\
&= \sum_{i=1}^n \frac{4m^3(p_i-q_i)^2(p_i+q_i)+2m^2(p_i+q_i)^2}{\Big(\max\{ \sqrt{mn}\cdot |p_i-q_i|, n\cdot  (p_i+q_i), 1 \}\Big)^2}\nonumber\\
&\leq \sum_{i=1}^n \frac{4m^3(p_i+q_i)}{mn}+\sum_{i=1}^n\frac{2m^2}{n^2}= \frac{8m^2}{n}+\frac{2m^2}{n} = \frac{10m^2}{n}.\nonumber
\end{align}
Next, we prove the second inequality:
\begin{align*}
\sum_{i=1}^n \frac{m^2(p_i-q_i)^2}{f_i} &= \sum_{i=1}^n \frac{m^2|p_i-q_i|^2}{\max\{ \sqrt{mn}\cdot |p_i-q_i|, n\cdot  (p_i+q_i), 1 \}} \\    
&\leq  \sum_{i=1}^n \frac{m^{3/2}|p_i-q_i|}{n^{1/2}} = \frac{m^{3/2}\norm{p-q}_1}{n^{1/2}}.
\end{align*}
Finally, we prove the last inequality:
\begin{align*}
\frac{m^2\norm{p-q}_1^2}{\sum_{i=1}^n{f}_i}&=\frac{m^2\norm{p-q}_1^2}{\sum_{i=1}^n \max\{ \sqrt{mn}\cdot |p_i-q_i|, n\cdot  (p_i+q_i),1 \}}\\
&\geq \frac{m^2\norm{p-q}_1^2}{\sum_{i=1}^n \mleft( \sqrt{mn}\cdot |p_i-q_i| + n\cdot  (p_i+q_i) + 1\mright)} \\
&= \frac{m^2\norm{p-q}_1^2}{ \sqrt{mn}\cdot \norm{p-q}_1 + 2n + n} \\
&\geq  \min\Big(\frac{m^{3/2}{\norm{p-q}_1}}{2n^{1/2}}, \frac{m^2{\norm{p-q}_1}^2}{6n}\Big).\qedhere
\end{align*}
\end{proof}

\begin{lemma}
For any $m\leq n$ and distributions $p$ and $q$ over $[n]$ the following holds:
$
\sum_{i=1}^n \frac{\Var(Z_i)}{{f}_i^2} \leq 24m
$,
(2)~
$
\sum_{i=1}^n \frac{m^2(p_i-q_i)^2}{{f}_i} \leq m\norm{p-q}_1
$, 
and (3)~
$
\frac{m^2\norm{p-q}_1^2}{\sum_{i=1}^n{f}_i}\geq  \frac{m^2{\norm{p-q}_1^2}}{3n}.
$
\end{lemma}

\begin{proof}
As before, we first upper bound $\sum_{i=1}^n \frac{\Var(Z_i)}{ f_i^2}$:
\begin{align*}
\sum_{i=1}^n \frac{\Var(Z_i)}{ f_i^2}&=\sum_{i=1}^n \frac{4m^3(p_i-q_i)^2(p_i+q_i)+2m^2(p_i+q_i)^2}{f_i^2} \\
&= \sum_{i=1}^n \frac{4m^3(p_i-q_i)^2(p_i+q_i)+2m^2(p_i+q_i)^2}{\Big(\max\{ m\cdot  (p_i+q_i), 1 \}\Big)^2}\\
&\leq \sum_{i=1}^n \frac{4m^3(p_i-q_i)^2(p_i+q_i)+4m^2(p_i-q_i)^2+8m^2q_i^2}{ \max\{ m^2\cdot  (p_i+q_i)^2, 1 \}} \tag{as $(a+b)^2 \leq 2(a-b)^2+4b^2$}\\
&\leq \sum_{i=1}^n \frac{4m^3(p_i-q_i)^2(p_i+q_i)}{m^2\cdot(p_i+q_i)^2}+\sum_{i=1}^n \frac{4m^2(p_i-q_i)^2}{m\cdot(p_i+q_i)}+\sum_{i=1}^n\frac{8m^2 q_i^2}{\max\{m^2(p_i+q_i)^2,1\}}\\
&\leq 4m\sum_{i=1}^n |p_i-q_i|+4m\sum_{i=1}^n |p_i-q_i|+\sum_{i=1}^n\frac{8m^2 q_i^2}{\max\{m(p_i+q_i),1\}} \tag{$\max\{x^2,1\} \geq \max\{x,1\}$}\\
&\leq 8m\norm{p-q}_1+\sum_{i=1}^n8mq_i\\
&\leq  16m +8m=24m\,.
\end{align*}

Next, we prove the second inequality:
\begin{align*}
\sum_{i=1}^n \frac{m^2(p_i-q_i)^2}{f_i} &= \sum_{i=1}^n \frac{m^2|p_i-q_i|^2}{\max\{ m\cdot  (p_i+q_i), 1 \}}  
\leq  \sum_{i=1}^n m|p_i-q_i| = m\norm{p-q}_1.
\end{align*}
Finally, we prove the last inequality:
\begin{align*}
\frac{m^2\norm{p-q}_1^2}{\sum_{i=1}^n{f}_i} &= \frac{m^2\norm{p-q}_1^2}{\sum_{i=1}^n \max\{ m\cdot  (p_i+q_i),1 \}}
\geq \frac{m^2\norm{p-q}_1^2}{\sum_{i=1}^n \mleft( m\cdot  (p_i+q_i) + 1\mright)} 
= \frac{m^2\norm{p-q}_1^2}{ 2m + n} 
\geq  \frac{m^2{\norm{p-q}_1^2}}{3n}.
\end{align*}
\end{proof}

\noindent It only remains to establish the last part of  Lemma~\ref{lem:mainmgn2}, which we do next.
\begin{align*}
 \Var\mleft[ Z\;\middle|\;\widehat{f}_i \text{ for }i\in[n] \mright] 
&=\sum_{i=1}^n \frac{\Var(Z_i)}{\widehat{f}_i^2}
= \sum_{i=1}^n \frac{4m^3(p_i-q_i)^2(p_i+q_i)+2m^2(p_i+q_i)^2}{\widehat{f}_i^2}\\ 
&\overset{\rm(a)}\le 4m^3 \mleft(\sum_{i=1}^n \frac{(p_i-q_i)^4}{\widehat{f}^2_i} \mright)^{1/2}\mleft(\sum_{i=1}^n \frac{(p_i+q_i)^2}{\widehat{f}^2_i} \mright)^{1/2}+\sum_{i=1}^n\frac{ 2m^2(p_i+q_i)^2}{\widehat{f}_i^2}\\ 
&\overset{\rm(b)}\le 4m^3 \mleft(\sum_{i=1}^n \frac{(p_i-q_i)^2}{\widehat{f}_i} \mright)\mleft(\sum_{i=1}^n \frac{(p_i+q_i)^2}{\widehat{f}^2_i} \mright)^{1/2}+\sum_{i=1}^n\frac{ 2m^2(p_i+q_i)^2}{\widehat{f}_i^2}\\ 
&= 4 \mleft(m^2\sum_{i=1}^n \frac{(p_i-q_i)^2}{\widehat{f}_i} \mright)\mleft(m^2\sum_{i=1}^n \frac{(p_i+q_i)^2}{\widehat{f}^2_i} \mright)^{1/2}+\sum_{i=1}^n\frac{ 2m^2(p_i+q_i)^2}{\widehat{f}_i^2}
\end{align*}
where step (a) is the Cauchy--Schwarz inequality, and (b) is monotonicity of $\ell_p$ norms: for any vector $u$, $\norm{u}_2\le \norm{u}_1$. We can then continue as follows, making the expectation appear:
\begin{align*}
 \Var\mleft[ Z\;\middle|\;\widehat{f}_i \text{ for }i\in[n] \mright] 
&= 4 \mleft( \E\mleft[ Z\;\middle|\;\widehat{f}_i \text{ for }i\in[n] \mright] \mright)\mleft(\sum_{i=1}^n \frac{m^2(p_i+q_i)^2}{\widehat{f}^2_i} \mright)^{1/2}+\sum_{i=1}^n\frac{ 2m^2(p_i+q_i)^2}{\widehat{f}_i^2}\\
&\overset{\rm(c)}\le \frac{1}{40}\mleft( \E\mleft[ Z\;\middle|\;\widehat{f}_i \text{ for }i\in[n] \mright] \mright)^2+(160+2)\sum_{i=1}^n\frac{ m^2(p_i+q_i)^2}{\widehat{f}_i^2}\\
&\overset{\rm(d)}\le \frac{1}{40}\mleft( \E\mleft[ Z\;\middle|\;\widehat{f}_i \text{ for }i\in[n] \mright] \mright)^2+162\sum_{i=1}^n \frac{2m^2(p_i-q_i)^2}{\widehat{f}^2_i}+162\sum_{i=1}^n \frac{4m^2q_i^2}{\widehat{f}^2_i}\\
&\overset{\rm(e)}\le \frac{1}{40}\mleft( \E\mleft[ Z\;\middle|\;\widehat{f}_i \text{ for }i\in[n] \mright] \mright)^2+324\sum_{i=1}^n \frac{m^2(p_i-q_i)^2}{\widehat{f}_i}+648\sum_{i=1}^n {m^2q_i^2}\\
&= \frac{1}{40}\mleft( \E\mleft[ Z\;\middle|\;\widehat{f}_i \text{ for }i\in[n] \mright] \mright)^2+324\E\mleft[ Z\;\middle|\;\widehat{f}_i \text{ for }i\in[n] \mright]+648m^2\norm{q}_2^2, 
\end{align*}
where step (c) uses $2ab\le a^2+b^2$, (d) uses $(a+b)^2\le 2(a-b)^2+4b^2$, and finally (e) uses $\widehat{f_i}\ge 1$. \qedsymbol

\section{Lower Bounds for Tolerant Testing}
    \label{sec:lb}
In this section, we derive our lower bounds on the ``price of tolerance,'' i.e., on the increase in the sample complexity as a function of the parameters $\ve_1,\ve_2$.
The main technical result is a lower bound for tolerant uniformity testing, from which the results for identity and equivalence will follow. 
In particular, we show:
\begin{theorem}[The price of tolerance for uniformity testing]\label{thm:lb-main}
For any $n$ and $\ve_1<\ve_2 <c$, for some universal constant $c>0$, any tester which for any unknown distribution $p$ over $[n]$ distinguishes between $\|p-\unif_n\|_1\le \ve_1$ and $\|p-\unif_n\|_1\ge \ve_2$ with probability at least $4/5$ must use $\Omega\bigg(\frac{n}{\log n}\Big(\frac{\ve_1}{\ve_2^2}\Big)+\frac{n}{\log n}\Big(\frac{\ve_1}{\ve_2^2}\Big)^2\bigg)$ samples from $p$.
\end{theorem}
\noindent
By combining the above lower bound with previously known lower bounds for non-tolerant uniformity/identity testing~\cite{Paninski08}, we obtain:
\begin{corollary}[Tolerant uniformity testing lower bound]
    \label{cor:lbU}
For any $n$ and $0\leq \ve_1<\ve_2 <c$, for some universal constant $c>0$, any tester which for any unknown distribution $p$ over $[n]$ distinguishes between $\|p-\unif_n\|_1\le \ve_1$ vs $\|p-\unif_n\|_1\ge \ve_2$ with probability $\ge 4/5$ needs at least 
\[
\Omega\bigg(\frac{n}{\log n}\Big(\frac{\ve_1}{\ve_2^2}\Big)+\frac{n}{\log n}\Big(\frac{\ve_1}{\ve_2^2}\Big)^2+\frac{\sqrt{n}}{\ve_2^2}\bigg)
\]samples from $p$.
\end{corollary}
\noindent
Similarly, by combining our lower bound with previously known lower bounds for non-tolerant equivalence testing~\cite{Valiant11,ChanDVV14}, we obtain:
\begin{corollary}[Tolerant equivalence testing lower bound]
    \label{cor:lbC}
For any $n$ and $0\leq \ve_1<\ve_2 <c$, for some universal constant $c>0$, any tester which for any unknown distributions $p$ and $q$, both over $[n]$,  distinguishes between $\|p-q\|_1\le \ve_1$ and $\|p-q\|_1\ge \ve_2$ with probability at least $4/5$ must use 
\[
\Omega\bigg(\frac{n}{\log n}\Big(\frac{\ve_1}{\ve_2^2}\Big)+\frac{n}{\log n}\Big(\frac{\ve_1}{\ve_2^2}\Big)^2+\frac{\sqrt{n}}{\ve_2^2}+\frac{ n^{2/3}}{\ve_2^{4/3}}\bigg)
\]
samples.
\end{corollary}

\subsection{The moment matching technique}
The starting point for our proof of Theorem~\ref{thm:lb-main} is the moment matching technique first used in~\cite{WuY16}.
We briefly review this technique here.
The first step is to consider the Poissonized version of the problem.
Namely, given $\Theta (m)$ samples from a distribution $p = (p_1, \ldots, p_n)$, then with high probability, we can simulate a set of $\poi(m)$ samples from the same distribution.
Thus, without loss of generality, we may assume that we are given $\poi(m)$ samples from $p$, and our goal is to distinguish with high probability given these samples whether $\| p - \unif_n \|_1 \leq \eps_1$ or $\| p - \unif_n \|_1 \geq \eps_2$.
A classical fact is that the result of sampling $\poi(m)$ samples from $p$ is identical in distribution to a draw from $(X_1, \ldots, X_n)$, where now the $X_i \sim \poi(m p_i)$ are independent. 

The high-level idea of the moment matching technique to construct two priors $\calU, \calU'$ over distributions on $n$ elements so that with high probability two conditions hold. First, if $p \sim \calU$ and $p' \sim \calU'$, then with high probability, $\Norm{p - \unif_n}_1 \leq \ve_1$ and $\Norm{p' - \unif_n}_1 > \ve_2$.
Second, the result of (i)~sampling a distribution $p \sim \calU$ then (ii)~sampling $\poi(m)$ elements from $p$ is close in TV distance to applying the same process to $\calU'$.
Specifically, the priors we construct will be product distributions, that is, $\calU = \mathcal{P}^n$ and $\calU' = (\mathcal{P}')^n$ for some positive univariate distributions $\mathcal{P}, \mathcal{P}'$.
Then, ignoring some technical issues which we will address momentarily, the problem becomes: find distributions $\mathcal P, \mathcal P'$ supported on nonnegative values such that (1)~$n \E_\mathcal P \Abs{p_i - 1/n} \leq \ve_1$ and $n \E_{\mathcal P'} \Abs{p'_i - 1/n} > \ve_2$, and (2)~the following distance is small:
\[
    \operatorname{TV}\mleft(\E_{\mathcal{P}^n} \Paren{\poi (mp_1), \ldots, \poi(m p_n)}, \  \E_{(\mathcal{P}')^n} \Paren{\poi (mp_1'), \ldots, \poi(mp_n')}\mright) = o(1) \; .
\]
Note that, in view of the subaditivity of TV distance, this condition can be relaxed to the condition 
\begin{equation}\label{eq:poitv}
     \operatorname{TV}\mleft(\E_{\mathcal{P}} \poi(mp_i),\  \E_{\mathcal{P'}} \poi(mp_i')\mright) = o(1/n) \; .
\end{equation}

While this will make later calculations much simpler, this introduces a couple of minor complications here.
First, the vectors in the domain of $\calU, \calU'$ may not sum to 1, that is, $\calU$ and $\calU'$ may not actually be priors over \textit{bona fide} distributions.
However, if we additionally enforce that $\E_{V \sim \mathcal{P}} [V] = \E_{V' \sim \mathcal{P}'} [V'] = 1/n$, then under some mild conditions on the $\mathcal{P}, \mathcal{P}'$, by standard concentration arguments, the resulting vectors are very close to summing to 1 and thus form ``approximate'' distributions.
One can then show that by slightly changing the construction, we can create priors over distributions that satisfy the desired properties.
Second, the vectors $p, p'$ in the domain of $\calU, \calU'$ may not deterministically satisfy the properties that $\Norm{p - \unif_n}_1 \leq \ve_1$ and $\Norm{p' - \unif_n}_1 > \ve_2$.
However, again by standard concentration inequalities, with high probability these random variables will not exceed their expectation by too much, and thus will satisfy these same constraints with high probability, perhaps relaxed by constant factors.
We make this discussion more precise in the following theorem, whose (rather technical) proof is deferred to Section~\ref{sec:main-lb-conversion}:
\begin{theorem}
\label{thm:testinglb-to-moments}
Let $0\le \eps_1 < \eps_2 \le 1$, and let $n, m$ be positive integers and $m\ge c$, where $c>0$ is an absolute constant. 
Suppose there exist random variables $U, U'$ supported on the domain $[a, b]$ so that $b-a\le \frac{\ve_2^2}{1000}$, $\E [U] = \E [U'] = 1/n$, and
\begin{equation}
    \E \left[ \left|U - \frac{1}{n} \right| \right] \le \frac{\eps_1}{n} \; , \qquad \mbox{and} \qquad \E \left[ \left|U' - \frac{1}{n} \right| \right] \ge \frac{\eps_2}{n} \; .
\end{equation}
Moreover, assume 
\begin{equation}\label{eq:poitv-1}
     \operatorname{TV}\mleft(\, \E  \poi(m U),\, \E \poi(m U')\mright) \le \frac1{20n} \; .
\end{equation}
Then, any tester which for any unknown distribution $p$ distinguishes between $\|p-\unif_n\|_1\le 25\ve_1$ and $\|p-\unif_n\|_1\ge \ve_2/2$ with probability at least $4/5$ requires at least $m/2$ samples from $p$.
\end{theorem}

\noindent

Thus, for given $m$, $n$ and $\ve_1$ the problem reduces to finding the maximum value of $\ve_2$ for which we can construct a pair of random variables $U$ and $U'$ for which the assumptions of Theorem~\ref{thm:testinglb-to-moments} hold. 
The next key insight is that we can further reduce the condition in~\eqref{eq:poitv-1} to designing two random variables with matching moments:
\begin{lemma}[{\cite[Lemma~32]{JiaoHW18}; see also~\cite{WuY16}}]
\label{lem:mm}
    For any $\kappa\ge M\ge 0$, let $Y, Y'$ be two random variables over $[\kappa-M, \kappa+M]$ so that $\E Y^i = \E Y'^i$ for all $i = 1, \ldots, L$.
    Then, we have
    \[
      \operatorname{TV}\mleft(\E \poi(Y), \E \poi (Y')\mright) \leq 2\Paren{\frac{e M}{\sqrt{\kappa (L+1)}}}^{L+1} \; .
    \]
\end{lemma}
\noindent
With this lemma in place, our goal can be restated as follows: maximize $n\cdot  \E \Abs{U'_i - 1/n}$ such that $n \cdot \E_{\mathcal P} \Abs{U - 1/n}\le \ve_1$ and the first $L$ moments of $U$ and $U'$ match, where the support of $U$ and $U'$ is over $[\frac{\kappa-M}{m},\frac{\kappa+M}{m}]$ for some $\kappa\ge M\ge 0$.
The value of this maximum is a function of the parameters $\kappa$, $M$ and $L$, whose values we choose later appropriately so that this function is maximized, while
\begin{equation}\label{eq:tvbound}
    2\Paren{\frac{eM}{\sqrt{\kappa(L+1)}}}^{L+1} \le \frac{1}{20n}
\end{equation}
holds, so that~\eqref{eq:poitv-1} is satisfied.

We formulate the problem of maximizing $n\cdot  \E \Abs{U' - 1/n}$ for any given choice of parameters as the following linear program over infinitely many variables, where we have used random variables $V$ and $V'$ to denote $n\cdot U$ and $n\cdot U'$, respectively, 
\begin{align*}
    \max\,  \E |V'-1|\mbox{ s.t. } & \E |V-1|\le \ve_1 \mbox{ and } \\
    & \E V =\E V' = 1, \mbox{ and }  \\
    & \E V^i = \E V'^i, i = 2, \ldots, L, \mbox{ and }\\
    & V, V' \in \Big[\frac{n(\kappa-M)}{m}, \frac{n(\kappa+M)}{m}\Big] \;. \numberthis \label{eq:lp-1}
\end{align*}
Let $\mathcal L(\ve_1,n,m,M,\kappa,L)$ denote the value of the optimal solution of the above optimization problem. 
Observe that we do not need to find the exact solution to the above linear program: instead, any reasonable lower bound on the solution of the above optimization problem suffices.
The next theorem gives one such lower bound. To state the theorem, we define
\begin{align}\label{eq:def2AB}
    A := \frac{n(M+\kappa)}{m}-1-\ve_1, \text{ and } B:= \frac{n(M-\kappa)}{m}+1-\ve_1. 
\end{align}

\begin{theorem}\label{th:optboundfinal}
For any $\kappa$, $M$, $n$, $m$, $L$, and $\ve_1$, if for $A,B$, defined in~\eqref{eq:def2AB}, $0< \ve_1\le \min\big\{\frac{B}{4},\frac{A}{4}\big\}$, then the value of optimal solution of~\eqref{eq:lp-1} is lower bounded by
\begin{align}
\mathcal L(\ve_1,n,m,M,\kappa,L) \geq \frac{1}{12}\max\Bigg\{\sqrt{\ve_1\cdot \frac{A+B}{32L^2}},
 \sqrt{\ve_1\cdot \frac{\sqrt{AB}}{16L}}\Bigg\}.
\end{align}
\end{theorem}

We prove the theorem later in Section~\ref{sec:optboundfinal}. First in Section~\ref{sec:lowerb} we use this theorem to prove the distribution testing lower bounds.

\subsection{Proof of Theorem~\ref{thm:lb-main}}\label{sec:lowerb}
Theorem~\ref{th:optboundfinal} implies that for any $\ve_2$ smaller than $\mathcal L(\ve_1,n,m,M,\kappa,L)$, there exist random variables $U= \frac{V}{n}$ and $U' =\frac{V'}{n}$ such that
\begin{align*}
    &\E |U'-1/n|\ge \ve_2/n  \mbox{ and }\\
    & \E |U-1/n|\le \ve_1/n \mbox{ and } \\
    & \E U =\E U' = 1/n, \mbox{ and }  \\
    & \E U^i = \E U'^i, i = 2, \ldots, L, \mbox{ and }\\
    & U, U' \in \Big[\frac{(\kappa-M)}{m}, \frac{(\kappa+M)}{m}\Big] \;.
\end{align*}
Next, we choose the values of parameters $L$, $\kappa$ and $M$ so that $\mathcal L(\ve_1,n,m,M,\kappa,L)$ is maximized while~\eqref{eq:tvbound} hold,
which by Lemma~\ref{lem:mm} will imply $\operatorname{TV}\mleft(\, \E  \poi(m U),\, \E \poi(m U')\mright) \le \frac1{20n}$. 
The choice of the parameters differs for different regimes of $m$. 

\begin{itemize}
    \item First we consider the regime $m < \frac{1}{4}n\log n$. Consider any such $m$ and $\ve_1\le 1/8$. Choose $\kappa = M =\log n$ and $L=4e^2\log n$. One can check that the desired bound on TV distance in~\eqref{eq:tvbound} is  satisfied for these choices of the parameters.
Further, we have $A = \frac{2n\log n}{m}-1-\ve_1\ge \frac{n\log n}{m} $, where we used $\frac{n\log n}m> 4 >1+\ve_1$ in the above parameter range; and $B = 1-\ve_1\ge 1/2$. Finally, $\ve_1< 1/8$ we have that $\ve_1\le \min\big\{\frac{B}{4},\frac{A}{4}\big\}$. 
Then, invoking Theorem~\ref{thm:testinglb-to-moments} we get that for any
\[\ve_2\le \mathcal L(\ve_1,n,m,M = \log n,\kappa = \log n,L = 4e^2\log n)
\] 
and $\ve_2^2\ge 1000\frac{M}{m} = 1000\frac{\log n}{m}$, one cannot distinguish between $25\ve_1$-close and $\ve_2/2$-far using $m/2$ samples.

Equivalently, by rescaling the parameters, for any $m< c_0 n\log n$,
\begin{align}\label{eq:somebound}
\ve_2\le \frac{1}{2}\mathcal L\mleft(\frac{\ve_1}{25},n,2m,M = \log n,\kappa = \log n,L = 4e^2\log n\mright),    
\end{align}
and $\frac{\ve_2^2}{2^2}\ge 1000\frac{\log n}{m}$, we can not distinguish between $\ve_1$-close and $\ve_2$-far using $m$ samples.

Next, we show that the above statement holds even without the constrain $\frac{\ve_2^2}{2^2}\ge 1000\frac{\log n}{m}$. We do so by showing that even when $\frac{\ve_2^2}{2^2}< 1000\frac{\log n}{m}$ or equivalently $m< 4000\frac{\log n}{\ve_2^2}$, we can not distinguish between $\ve_1$-close and $\ve_2$-far using $m$ samples, even for $\ve_1=0$.
From the known uniformity testing lower bound for non tolerant case, we know that there is an absolute constant $c_3>0$ such that $m \ge  c_3\frac{\sqrt{n}}{\ve_2^2}$ samples are needed to distinguish correctly between $p=\unif_n$ and $\norm{p-\unif_n}_1\ge \ve_2$  with probability $\ge 4/5$.
Since for $n$ larger than an absolute constant $c_4$, we have $c_3  \frac{\sqrt{n}}{\ve_2^2}>4000\frac{\log n}{\ve_2^2} > m$, hence $m$ samples are insufficient and the claim follows.

From Theorem~\ref{th:optboundfinal} we get:
\begin{align*}
\frac{1}{2}\mathcal L\mleft(\frac{\ve_1}{25},n,2m,M = \log n,\kappa = \log n,L = 4e^2\log n\mright) \ge \max\Bigg\{c_1\sqrt{\frac{\ve_1 n }{m\log n}},c_2
 \sqrt{ \frac{\ve_1 \sqrt{n}}{\sqrt{m\log n}}}\Bigg\} ,
\end{align*}
where $c_1$ and $c_2$ are some absolute positive constants.

From~\eqref{eq:somebound} and the above lower bound on $\mathcal{L}$ it follows that for any $n>c_4$, for any $m< (n\log n)/4$ and $\ve_1< 1/8$, such that
\[
m <  \max\mleft\{c_1^2\frac{n}{\log n}\Big(\frac{\ve_1}{\ve_2^2}\Big),c_2^3\frac{n}{\log n}\Big(\frac{\ve_1}{\ve_2^2}\Big)^2\mright\}
,
\]
then using $m$ samples from $p$ one can not distinguish correctly with probability $\ge 4/5$ between $\norm{p-\unif_n}_1\le \ve_1$ and $\norm{p-\unif_n}_1\ge \ve_2$.

Observe given $c_1$ and $c_2$, there exist an universal constant $c_5$ such that for any $\Big(\frac{\ve_1}{\ve_2^2}\Big)\le c_5\log n$, we have
\[
\max\mleft\{c_1^2\frac{n}{\log n}\Big(\frac{\ve_1}{\ve_2^2}\Big),c_2^3\frac{n}{\log n}\Big(\frac{\ve_1}{\ve_2^2}\Big)^2\mright\}\le \frac{n\log n}{4}.
\]
Therefore, for any $\ve_1\le \frac{1}{8}$, $n>c_4$ and $\Big(\frac{\ve_1}{\ve_2^2}\Big)\le c_5\log n$, then using $\max\mleft\{c_1^2\frac{n}{\log n}\Big(\frac{\ve_1}{\ve_2^2}\Big),c_2^3\frac{n}{\log n}\Big(\frac{\ve_1}{\ve_2^2}\Big)^2\mright\} = \Omega\mleft(\frac{n}{\log n}\Big(\frac{\ve_1}{\ve_2^2}\Big)+\frac{n}{\log n}\Big(\frac{\ve_1}{\ve_2^2}\Big)^2\mright)$ samples from $p$ one can not distinguish correctly with probability $\ge 4/5$ between $\norm{p-\unif_n}_1\le \ve_1$ and $\norm{p-\unif_n}_1\ge \ve_2$.

\item Next, we choose the parameters for the regime $ m> 4 {n\log n}$. 
Note that since the theorem statement makes no claim for the setting where $m \geq \frac{n}{\ve_2^2 \log n}$, we can restrict our attention to the case where $m \leq \frac{n}{\ve_2^2\log n}$.
Furthermore, because $\frac{n}{\ve_2^2\log n}\ll \frac{n\log n}{64\ve_1^2}$ (using the fact that $\ve_1 \leq \ve_2$), we only need to consider $m \leq \frac{n\log n}{64\ve_1^2}$. 
This condition on $m$ implies that $\ve_1\le \frac{1}{8}\sqrt{\frac{n \log n}{m}}$. Consider any such $m$. Choose $\kappa = \frac{m}{n}$ and $M = \sqrt{\frac{m\log n}{n}}$ and $L=4e^2\log n$. Observe that for this choice $ M< \frac{m}{n} = \kappa$, hence $\kappa-M> 0$.
Then $A = B = \frac{nM}{m}-\ve_1 = \sqrt{\frac{n \log n}{m}}-\ve_1$. 
Since $\ve_1\le \frac{1}{8}\sqrt{\frac{n \log n}{m}}$, observe that $A,B \ge   \frac{1}{2}\sqrt{\frac{n \log n}{m}}$ and $\ve_1\le \min\big\{\frac{B}{4},\frac{A}{4}\big\}$.
Then, invoking Theorem~\ref{thm:testinglb-to-moments} and rescaling the parameters as for the previous case, we get
\[
\ve_2\le \frac{1}{2}\mathcal L\mleft(\frac{\ve_1}{25},n,2m,M = \sqrt{\frac{m\log n}{n}},\kappa = \frac{m}{n},L = 4e^2\log n\mright),
\] 
and $\frac{\ve_2^2}{2^2}\ge 1000\frac{M}{m}=1000\sqrt{\frac{ \log n}{mn}}$, we can not distinguish between $\ve_1-$close vs $\ve_2-$far using $m$ samples.
From Theorem~\eqref{th:optboundfinal} we get:
\[
\mathcal L\mleft(\ve_1,n,m,M = \sqrt{\frac{m\log n}{n}},\kappa = \frac{m}{n},L =4e^2\log n\mright)  \ge c_6
 \sqrt{\ve_1\cdot \frac{\sqrt{n\log n}}{\sqrt{m}\log n}},
\]
where $c_6$ is some absolute constant.

Following the similar steps as before it can be shown that for any $\ve_1\le \frac{1}{8}$, $n>c_8$ and $\Big(\frac{\ve_1}{\ve_2^2}\Big)\ge c_7\log n$, then using ${c_6^4}\frac{n}{\log n}\Big(\frac{\ve_1}{\ve_2^2}\Big)^2$ samples from $p$ one can not distinguish correctly with probability $\ge 4/5$ between $\norm{p-\unif_n}_1\le \ve_1$ and $\norm{p-\unif_n}_1\ge \ve_2$. Using $\Big(\frac{\ve_1}{\ve_2^2}\Big)\ge c_7\log n$, we get ${c_6^4}\frac{n}{\log n}\Big(\frac{\ve_1}{\ve_2^2}\Big)^2 = \Omega\mleft(\frac{n}{\log n}\Big(\frac{\ve_1}{\ve_2^2}\Big)+\frac{n}{\log n}\Big(\frac{\ve_1}{\ve_2^2}\Big)^2\mright)$ bound on the sample complexity.

\end{itemize}

We have so far shown the target lower bound on the sample complexity for both regimes $\frac{\ve_1}{\ve_2^2} < c_5 \log n$ and $\frac{\ve_1}{\ve_2^2} > c_7 \log n$ for some absolute positive constants $c_5$ and $c_7$. To conclude for the intermediate cases, observe that the sample complexity is an increasing function of $\ve_1$ (more tolerance makes the problem harder) and a decreasing function of $\ve_2$. Thus, by monotonicity, the lower bound for $\frac{\ve_1}{\ve_2^2} = c_5 \log n$ still applies to $c_5 \log n \leq \frac{\ve_1}{\ve_2^2} \leq c_7 \log n$, by relaxing the problem to $\ve_1' = \frac{c_5\ve_1}{c_7}$. This only affects the resulting lower bound by a constant factor, and allows us to conclude with the desired
\[
\Omega\Big(\frac{n}{\log n}\Big(\frac{\ve_1}{\ve_2^2}\Big)^2+\frac{n}{\log n}\Big(\frac{\ve_1}{\ve_2^2}\Big)\Big)
\]
sample complexity lower bound for the full range of parameters.

\subsection{Proof of Theorem~\ref{th:optboundfinal}}\label{sec:optboundfinal}

We break the proof of Theorem~\ref{th:optboundfinal} into two parts.
First, we convert the primal form of the problem into a more convenient representation via a few helpful transformations, and then take the dual (Section~\ref{sec:primal}).
We then lower bound the value of the dual using tools from approximation theory (Section~\ref{sec:dual}). 

\subsubsection{Transforming the primal}
\label{sec:primal}
First, to simplify the optimization problem~\eqref{eq:lp-1}, notice that moment matching of all degree-$L$ or less moments is unaffected by translation.
So if one introduces the random variables $X = V - 1$ and $X' = V' - 1$, we see that these are supposed to be mean zero random variables over $[\frac{n(\kappa-M)}{m}-1, \frac{n(\kappa+M)}{m}-1]$ with matching $L$-th and below moments, and that distance to uniformity corresponds to $\E |X|$ and $\E |X'|$.  
\begin{align}
    \max\,  \E |X'|\mbox{ s.t. } & \E |X|\le \ve_1 \mbox{ and } \nonumber\\
    & \E X = \E X' = 0\nonumber\\
    & \E X^i = \E X'^i, i = 2, \ldots, L, \mbox{ and }\nonumber\\
    & X, X' \in \Big[\frac{n(\kappa-M)}{m}-1, \frac{n(\kappa+M)}{m}-1\Big] \;.\label{eq:oporig}
\end{align}
It will be useful to remove the constraint that $\E [X] = \E [X'] = 0$.
To do so, we propose the following optimization problem without this constraint.
\begin{align}
    \max\,  \E |Y'|\mbox{ s.t. } & \E |Y|\le \frac{\ve_1}{2} \mbox{ and } \nonumber\\
    & \E Y^i = \E Y'^i, i = 1, \ldots, L, \mbox{ and }\nonumber\\
    & Y, Y' \in [-B,A] \;,\label{eq:opnew}
\end{align}
{where $A = \frac{n(\kappa+M)}{m}-1-\ve_1$ and $B = -\Big(\frac{n(\kappa-M)}{m}-1+\ve_1\Big)$, first defined in Equation~\eqref{eq:def2AB}.}
We show the following claim.
\begin{lemma}\label{lem:relsol}
The value of the solution of~\eqref{eq:oporig} is at least half the value of the solution of~\eqref{eq:opnew}. 
\end{lemma}
\begin{proof}
Let $Y$ and $Y'$ be the random variables that achieve the maximum in~\eqref{eq:opnew}. To prove the claim, first we show that random variables $X = Y-\E  Y$ and $X' = Y'-\E  Y$ satisfy the constrains in~\eqref{eq:oporig}.

First note that $\E |X| = \E |Y-\E  Y|\le 2 \E|Y| \le \ve_1$.
Next, using $\E Y = \E Y'$, we get $\E X = \E [ Y-\E  Y] = 0 $ and $\E X' =  \E[Y'-\E  Y] = 0$.
Since the moment matching of all degree-$L$ or less moments is unaffected by translation, all $L$ moments of $X$ and $X'$ will match.
Finally, $|\E[Y]|\le \E[|Y|]\le \frac{\ve_1}2$, and $Y, Y' \in [\frac{n(\kappa-M)}{m}-1+\ve_1, \frac{n(\kappa+M)}{m}-1-\ve_1]$ ensures $X, X' \in [\frac{n(\kappa-M)}{m}-1, \frac{n(\kappa+M)}{m}-1]$.

Now to complete the proof we show that the solution of the optimization problem~\eqref{eq:oporig} is at least $\ge \frac{ \E |Y'|}{2}$. 

\noindent If $\E |Y'| \le \ve_1$, then this is obviously true as the maximum in~\eqref{eq:oporig} is at least $\ve_1$.
If $\E |Y'| \ge \ve_1$ then $\E|X'| = \E |Y'-\E Y|\ge \E |Y'|-\E |Y| \ge \E |Y'|-\frac{\ve_1}2> \frac{\E |Y'|}{2} $, where we used $\E |Y'| \ge \ve_1$ in the last step.
\end{proof}

Mechanically taking the dual of~\eqref{eq:opnew}, we obtain that the dual is:
\begin{align*}
    \min \frac{\ve_1}2 \alpha + z_1 + z_2 \mbox{ s.t. } & z_1 + \sum_{i = 1}^L c_i p_i (x) \geq \Abs{x} \mbox{ for all $x \in [-B, A]$,}\\
    & \alpha \Abs{x} \geq \sum_{i = 1}^L c_i p_i (x) - z_2 \mbox{ for all $x \in [-B, A]$}, \mbox{ and } \\
    & \alpha \geq 0 \;. \numberthis \label{eq:dual}
\end{align*}
By weak duality, we know that the value of the optimal solution to~\eqref{eq:opnew} is upper bounded by the value of the optimal solution to~\eqref{eq:dual}.
However, since we seek to prove a lower bound on the value of the optimal solution to~\eqref{eq:opnew}, this is insufficient for our purposes.
However, we show that in this case, strong duality still holds:
\begin{lemma}\label{lemma:value:dual}
The value of the optimal solution to~\eqref{eq:opnew} is equal to the value of the optimal solution to~\eqref{eq:dual}.
\end{lemma}
\noindent
This follows from the classical theory of convex duality~\cite{Rockafellar74}, however, for completeness we also give a self-contained proof of this fact in Section~\ref{sec:strong-duality}.

We now make a couple of final simplifications before we lower bound the value of~\eqref{eq:dual}.
Let $\mathcal{P}_L$ denote the collection of all degree-$L$ polynomials. We reparametrize the above dual by letting $\ve_2 \eqdef \alpha{\ve_1}$, $p(x) \eqdef \frac{\ve_1}{\ve_2}\sum_{i = 1}^L c_i p_i (x)-z_2$ and $z \eqdef z_1+z_2$.
\begin{align*}
    \min \frac{\ve_2}2+z \mbox{ s.t. } & \text{ for some }p(x)\in\mathcal{P}_L,\\
    &p(x) \geq \frac{\ve_1}{\ve_2}\Abs{x}-\frac{\ve_1}{\ve_2}{z} \mbox{ for all $x \in [-B, A]$, and }\\
    &  \Abs{x} \geq p(x) \mbox{ for all $x \in [-B,A]$, and}\\
    &\ve_2\ge 0.
    \numberthis \label{eq:intermdual}
\end{align*}
Imposing the additional constraint that $z = \ve_2$, we get the following program:
\begin{align*}
    \min \ve_2 \mbox{ s.t. } & \text{ for some }p(x)\in\mathcal{P}_L,\\
    &p(x) \geq \frac{\ve_1}{\ve_2}\Abs{x}-{\ve_1} \mbox{ for all $x \in [-B, A]$, and }\\
    &  \Abs{x} \geq p(x) \mbox{ for all $x \in [-B, A]$, and}\\
    &\ve_2\ge 0.\numberthis \label{eq:finaldual}
\end{align*}
Suppose $\ve_2=a$ and $z=b$ achieve the optimal solution of~\eqref{eq:intermdual}, which is therefore $a/2+b$.
It is easy to verify that for $z<0$ the optimization problem~\eqref{eq:intermdual} is infeasible, hence $b\ge 0$. Further, observe that $z=\ve_2 = \max\{a,b\}$ is also a feasible solution of~\eqref{eq:intermdual}, and $z+\ve_2/2 = \max\{3a/2,3b/2\}$, which is at most 3 times the optimal solution of~\eqref{eq:intermdual}. Since~\eqref{eq:finaldual} scales it down by $2/3$, it follows that the solution of the new dual~\eqref{eq:finaldual} is at most twice the solution of the previous dual~\eqref{eq:intermdual}.

\begin{lemma}\label{lem:firatandlastoptimizationrelation}
The value of the solution of~\eqref{eq:lp-1} is at least $1/4$ times the value of the solution of~\eqref{eq:finaldual}.
\end{lemma}
\begin{proof}
From the preceding discussion, the value of the solution of~\eqref{eq:lp-1} is equal to the value of the solution  of~\eqref{eq:oporig}. 
From Lemma~\ref{lem:relsol}, this value is at least is at least $1/2$ times the value of the solution of~\eqref{eq:opnew}.
From Lemma~\ref{lemma:value:dual} the values of the solutions of~\eqref{eq:opnew} and~\eqref{eq:dual} are equal. 
Furthermore, the value of the solution of~\eqref{eq:dual} is equal to the value of the solution of~\eqref{eq:intermdual}. Finally, the value of the solution of~\eqref{eq:intermdual} is at least $1/2$ times the value of the solution of~\eqref{eq:finaldual}. 
Hence, the  value of the solution of~\eqref{eq:lp-1} is at least $1/4$ times the value of the solution of~\eqref{eq:finaldual}.
\end{proof}

\subsubsection{Lower bounding the dual}
\label{sec:dual}

We now establish a lower bound on the solution of the dual (Equation~\eqref{eq:finaldual}).
We assume that parameters $A,B$ are such that $\ve_1\ll A, B$. First note that when $\ve_2\le\frac{\ve_1}2$, then the two conditions are contradictory and can not be met simultaneously. Therefore, we get the lower bound:
\begin{equation}\label{eq:trivduallb}
    \ve_2 \ge \frac{\ve_1}2.
\end{equation}

\noindent This implies that $\forall\, x \in [-B, A]$
\[
p(x)\ge 2|x|-\ve_1.
\]

\noindent Combining this with constraints in the dual imply that for all $x \in [-B, A]$
\[
|p(x)| \le 2|x|+{\ve_1}.
\]
From the above equation, for $x=0$, we get $|p_0|\le \ve_1$. Then $\forall\, x \in [-B, A]$,
\begin{align*}
|p(x)-p_0|&\le |p(x)|+|p_0|\nonumber\\
&\le 2|x|+\ve_1+\ve_1\nonumber\\
&= 2|x|+2\ve_1.
\end{align*}
Let $p(x)\in\mathcal{P}_L$ be any polynomial satisfying the constrains of the dual. Let $p_0=p(0)$, $p_1$ be the coefficients of $x$ in $p(x)$, and $\tilde p(x) = p(x)-p_0-p_1 x$.

\noindent Using the above equation,
\begin{align}
    |\tilde p(x)| \le  |p(x)-p_0| + |p_1x|\le (2+|p_1|)|x|+2\ve_1.\label{eq:con2}
\end{align}

\noindent The constraints in the dual also imply $p_0\le 0$ and $p({3\ve_2}) \ge 2\ve_1$. Combining these,
\begin{equation*}
    p(2\ve_2) - p_0 \ge 2\ve_1.
\end{equation*}
Using $\tilde p(x) = p(x)-p_1x-p_0$ for $x=3\ve_2$ in the above equation
\begin{equation*}
    \tilde p(3\ve_2) \ge 2\ve_1- p_13\ve_2.
\end{equation*}

\noindent Similarly, one can get
\begin{equation*}
    \tilde p(-3\ve_2) \ge 2\ve_1+p_13\ve_2.
\end{equation*}
Combining the two equations we get
\begin{equation}\label{eq:maincon}
    \max\{\tilde p(3\ve_2),\tilde p(-3\ve_2)\} \ge 2\ve_1+|p_1|3\ve_2. 
\end{equation}

\noindent Our lower bound on the dual is the consequence of the following key lemma, which we prove in Section~\ref{sec:proof-polgenboun}.
\begin{lemma}\label{lem:polgenboun}
Let $g$ be a degree-$L$ polynomial such that $g(0)=g'(0)=0$. For some $a< 0 < b$, $\gamma\ge 1$ and $0\le \delta\le \min\big\{-\frac{a}{4},\frac{b}{4}\big\}$ suppose $|g(x)|\le \gamma |x|+\delta$ for all $x\in [a,b]$
then 
\[
\min\{|x|: g(x) \ge (\gamma-1)|x|+\delta\} 
\ge \begin{cases}
  \sqrt{\delta\cdot \frac{b-a}{32L^2}} & \text{if $\delta< \frac{(b-a)}{32L^2}$}\\
 \sqrt{\delta\cdot \frac{\sqrt{|ab|}}{16L}} & \text{if $\delta<\frac{\sqrt{|ab|}}{16L}$.}
\end{cases}
\]
\end{lemma}

 The next theorem that establishes the lower bound on the dual follows by combining Equation~\eqref{eq:maincon}, Equation~\eqref{eq:con2} and Lemma~\ref{lem:polgenboun}.
\begin{theorem}[Lower bound on the solution to the dual]\label{th:duallb}
For any $A,B>0$ and $0< \ve_1\le \min\big\{\frac{B}{4},\frac{A}{4}\big\}$, the value of optimal solution of~\eqref{eq:finaldual} is lower bounded by
\[
\ve_2 \ge  \max\Bigg\{\frac{1}{3}\sqrt{\ve_1\cdot \frac{A+B}{32L^2}},\frac{1}{3}
 \sqrt{\ve_1\cdot \frac{\sqrt{AB}}{16L}},\frac{\ve_1}2\Bigg\}
\]
\end{theorem}
\begin{proof}
Applying Lemma~\ref{lem:polgenboun} for $g=\tilde p/2$, $\gamma = 1+|p_1|/2$, $\delta = \ve_1$, $b=A$, and $a=-B$ gives
\[
3\ve_2 \ge  \begin{cases}
  \sqrt{\ve_1\cdot \frac{A+B}{32L^2}} & \text{if $\ve_1< \frac{(A+B)}{32L^2}$}\\
 \sqrt{\ve_1\cdot \frac{\sqrt{AB}}{16L}} & \text{if $\ve_1<\frac{\sqrt{AB}}{16L}$.}
 \end{cases}
\]
Combining the above bound with the upper bound $\ve_2\ge \frac{\ve_1}2$ in~\eqref{eq:trivduallb}, and using the observations that if $\ve_1\ge \frac{(A+B)}{32L^2}$ then $\sqrt{\ve_1\cdot \frac{A+B}{32L^2}}\le \ve_1$, and if $\ve_1\ge \frac{\sqrt{AB}}{16L}$ then $\sqrt{\ve_1\cdot \frac{\sqrt{AB}}{16L}}\le \ve_1$, completes the proof.
\end{proof}

\noindent Combining Lemma~\ref{lem:firatandlastoptimizationrelation} and Theorem~\ref{th:duallb} proves Theorem~\ref{th:optboundfinal}.

\subsection{Proof of Lemma~\ref{lem:polgenboun}}
\label{sec:proof-polgenboun}
First, we recall some useful results from approximation theory, which we then leverage to derive a few auxiliary lemmas. Finally, using these lemmas we establish Lemma~\ref{lem:polgenboun}.

We will use the two following results, both of which bound the absolute value of derivatives of bounded polynomials. The first is the Markov Brothers' inequality, which gives a bound on all derivatives of bounded polynomials.
\begin{theorem}[Markov Brothers' inequality~\cite{Markoff1892}]\label{th:mbi}
For any real polynomial $g$ of degree-$L$, $0\le k\le t$, and real numbers $a$ and $b$ s.t. $-\infty < a < b < \infty $, the $k$\textsuperscript{th} derivative of $g$ satisfies
\[
\max_{x\in [a,b]}|g^{(k)}(x)| \le 2^k\cdot\frac{\max_{x\in [a,b] }|g(x)|}{(b-a)^{k}} \cdot\prod_{i=0}^{k-1} \frac{(L^2-i^2)}{2i+1} .
\]
\end{theorem}
\noindent
We next need Bernstein's inequality, which provides a bound on the first derivative, which for some values of $x$ is stronger than the Markov Brothers' inequality.
\begin{theorem}[Bernstein's inequality~\cite{Bernstein1912}]
For any real polynomial $g$ of degree-$L$ and real numbers $-\infty < a < x < b < \infty $, we have
\[
|g'(x)| \le \frac{ L\cdot \max_{x\in [a,b] }|g(x)|}{\sqrt{(x-b)(a-x)}}.
\]
\end{theorem}
\noindent
Using Bernstein's inequality above we can obtain the following bounds on the first and the second derivatives, which for some range of parameters improve upon Markov Brothers' inequality. 
Note that in the following lemma we have assumed $a< 0< b$.

\begin{lemma}\label{lem:bernder}
For any real polynomial $g$ of degree-$L$, and real numbers $a$ and $b$ s.t. $-\infty < a < 0 < b < \infty $, the following hold
\[
\max_{x\in [a/2,b/2]}
|g'(x)| \le \frac{L\cdot \max_{x\in [a,b] }|g(x)|}{\sqrt{|ba|/2}},
\]
and 
\[
\max_{x\in [a/4,b/4]} |g''(x)| \le \frac{4 L(L-1)\cdot \max_{x\in [a,b] }|g(x)|}{|ab|}.
\]
\end{lemma}
\begin{proof}
From Bernstein's inequality we get
\[
\max_{x\in [a/2,b/2]}
|g'(x)| \le \frac{L\cdot \max_{x\in [a,b] }|g(x)|}{\sqrt{|ba|/2}},
\]
where we used $(x-b)(a-x)\ge |ab|/2$ if $a< 0 < b$ and $x\in[a/2,b/2]$.

\noindent Replacing $g\to g'$, $a\to a/2$ and $b\to b/2$ in the above equation, we get
\[
\max_{x\in [a/4,b/4]}
|g''(x)| \le \frac{(L-1)\cdot \max_{x\in [a/2,b/2] }|g'(x)|}{\sqrt{|ba|/8}}.
\]
Combining the two equations proves the lemma.
\end{proof}

Using the bounds on the derivative of the bounded functions in Theorem~\ref{th:mbi} and Lemma~\ref{lem:bernder}, and a Taylor expansion, we derive the following lemma.
\begin{lemma}\label{lem:lbcruc}
Let $g$ be a degree-$L$ polynomial such that $g(0)=g'(0)=0$. Suppose that, for some $a< 0 < b$, we have $|g(x)|\le |x|$ for all $x\in [a,b]$. 
Then the following bounds hold:
\[
|g(x)|\le 
\begin{cases}
  \frac{8L^2}{(b-a)}x^2 & \text{if $x\in [a/4,b/4]$ and $|x|\le \frac{3(b-a)}{L^2}$} \\
  \frac{4L}{\sqrt{|ab|}}x^2 & \text{if $x\in [a/4,b/4]$ and $|x|\le\frac{\sqrt{|ab|}}{4L}$.}
\end{cases}
\]
\end{lemma}
\begin{proof}
By Taylor's theorem, we know that
\[
g(x) = g(0)+ g'(0) x + \frac{g''(c)}{2}x^2,
\]
where $c$ is some number between $0$ and $x$.

\noindent Then using $g(0)=g'(0)=0$, we obtain
\begin{equation}\label{eq:taylor}
|g(x)| \le x^2\cdot \max_{z\le|x|}\Big|\frac{g''(z)}{2}\Big|.
\end{equation}
Next, let $h(x)= g(x)/x$. Note that $h(x)$ is a degree-$(L-1)$ polynomial and $\max_{[a,b]}|h(x)|\le 1$. From Theorem~\ref{th:mbi} and Lemma~\ref{lem:bernder} it follows
\[
\max_{x\in [a/2,b/2]} |h'(x)|\le \min\Big\{\frac{2L^2}{b-a},\frac{\sqrt{2}\cdot  L}{\sqrt{|ab|}}\Big\},
\]
and
\[
\max_{x\in [a/4,b/4]} |h''(x)|\le \min\Big\{\frac{4 L^4}{3(b-a)^2},\frac{4L^2}{|ab|}\Big\}.
\]

\noindent The following relation between the second derivative of $g$ and the derivatives of $h$ can be obtained by differentiation $x\cdot h(x)$ twice using the chain rule:
\[
g''(x) = 2 h'(x) + x h''(x).
\]
Using the bounds on the derivatives of $h$, for any $x\in [a/4,b/4]$,
\[
|g''(x)|\le \min\Big\{\frac{4L^2}{(b-a)}+|x|\frac{4 L^4}{3(b-a)^2},\frac{2\sqrt{2}\cdot  L}{\sqrt{|ab|}}+|x|\frac{4L^2}{|ab|}\Big\}.
\]

\noindent Note that for any $x\in [a/4,b/4]$ s.t. $|x|\le \frac{3(b-a)}{L^2}$,
\[
|g''(x)|\le \frac{8L^2}{(b-a)},
\]
and similarly for any $x\in [a/4,b/4]$ s.t. $|x|\le \frac{\sqrt{|ab|}}{4L}$,
\[
|g''(x)|\le \frac{4L}{\sqrt{|ab|}}.
\]
Combining the above bounds with Equation~\eqref{eq:taylor} proves the lemma.
\end{proof}

Using this, we derive the following lemma.
\begin{lemma}\label{lem:linerbound}
Let $g$ be a degree-$L$ polynomial such that $g(0)=g'(0)=0$. Supoose that for some $a< 0 < b$, $\gamma\ge 1$ and $0< \delta\le \min\Big\{-\frac{a}{4},\frac{b}{4}\Big\}$ such that 
\[
\delta<  \max\bigg\{\frac{(b-a)}{16L^2}, \frac{\sqrt{|ab|}}{8L} \bigg\},
\] 
we have $|g(x)|\le \gamma |x|+\delta$ for all $x\in [a,b]$. Then $g(x)\le 2\gamma |x|$ for all $x\in[a,b]$.
\end{lemma}
\begin{proof}
Note that $|g(x)|\le \gamma |x|+\delta$ implies that $|g(x)|\le 2\gamma|x|$ for $|x|\ge \frac{\delta}{\gamma}$.
We proceed by contradiction. For contradiction, assume 
\[
\max_{|x|\le \frac{\delta}{\gamma}}\frac{|g(x)|}{|x|} = r\ge 1.
\]

\noindent Consider the polynomial $\tilde g(x) \eqdef g(x)/r$. 
Then from the definition of $r$, for any $x$ such that $|x|\le \frac{\delta}{\gamma}$ it follows that
\begin{equation}\label{eq:assump}
    |\tilde g(x)| = |g(x)|/r \le |x|,
\end{equation}
and
\begin{equation}\label{eq:contradiction}
|\tilde g(x)| =|x|
\end{equation}
for at least one $x$ such that $|x|\le \frac{\delta}{\gamma}$.
To show the contradiction, in the reminder of the proof, we show that this is impossible.
For $x\in [a,b]\setminus[-\frac{\delta}{\gamma},\frac{\delta}{\gamma}]$,
\[
|\tilde g(x)| = \frac{|g(x)|}{r}\le |g(x)|\le  \gamma|x|+\delta \le 2\gamma |x|,  
\]
where the last inequality uses the fact that for $|x|\ge \frac{\delta}{\gamma}$ we get $\delta\le |x|$. Combining with Equation~\eqref{eq:assump}, we get $|\tilde g(x)|\le 2\gamma|x|$ for all $x\in [a,b]$.

\noindent Applying Lemma~\ref{lem:lbcruc} on $\frac{\tilde g(x)}{2\gamma}$, we get
\[
|\tilde g(x)|\le 
\begin{cases}
  2\gamma\frac{8L^2}{(b-a)}x^2 & \text{if $x\in [a/4,b/4]$ and $|x|\le \frac{3(b-a)}{L^2}$} \\
  2\gamma\frac{4L}{\sqrt{|ab|}}x^2 & \text{if $x\in [a/4,b/4]$ and $|x|\le\frac{\sqrt{|ab|}}{4L}$.}
\end{cases}
\]
Recall $0\le \delta\le \min\big\{-\frac{a}{4},\frac{b}{4}\big\}$ and $\gamma\ge 1$.
First we show a contradiction when $\delta < \frac{b-a}{16L^2}$.
In this case,  $\frac{\delta}{\gamma}\le\delta \le \frac{b-a}{16L^2}<  \frac{3(b-a)}{L^2}$, hence for any $x$ such that $|x|\le \frac{\delta}{\gamma}$ the above equation implies that
\[
|\tilde g(x)|\le 2\gamma\frac{8L^2}{(b-a)}|x|^2 \le  \gamma\frac{16L^2}{(b-a)}\cdot |x|\cdot \max |x| \le \gamma\frac{16L^2}{(b-a)}\cdot\frac{\delta}{\gamma}=  \frac{16L^2}{(b-a)}\cdot |x|\cdot \delta < |x|,
\]
where the last inequality uses $\delta < \frac{b-a}{16L^2}$.
This contradicts Equation~\eqref{eq:contradiction}. 

Next we show a contradiction when $\delta < \frac{\sqrt{|ab|}}{8L}$ by using similar steps.
In this case,  $\frac{\delta}{\gamma}\le\delta \le \frac{\sqrt{|ab|}}{8L}<  \frac{\sqrt{|ab|}}{4L}$, hence for any $x$ such that $|x|\le \frac{\delta}{\gamma}$ we get
\[
|\tilde g(x)|\le 2\gamma\frac{4L}{\sqrt{|ab|}}|x|^2 \le  \gamma\frac{8L}{\sqrt{|ab|}}\cdot |x|\cdot \max |x| \le \gamma\frac{8L}{\sqrt{|ab|}}\cdot\frac{\delta}{\gamma}=  \gamma\frac{4L}{\sqrt{|ab|}}\cdot |x|\cdot \delta < |x|,
\]
where the last inequality uses $\delta < \frac{\sqrt{|ab|}}{8L}$.
This contradicts Equation~\eqref{eq:contradiction}. 
\end{proof}
Finally, we use Lemmas~\ref{lem:lbcruc} and~\ref{lem:linerbound} to derive Lemma~\ref{lem:polgenboun}.
\\
\\
\noindent\textit{Proof of Lemma~\ref{lem:polgenboun}. }
First consider the case $\delta <\frac{(b-a)}{16L^2} $. Then 
Lemma~\ref{lem:linerbound} implies
$
g(x)\le 2\gamma|x|.
$
Applying Lemma~\ref{lem:lbcruc} gives
\[
|g(x)|\le 2\gamma\frac{8L^2}{(b-a)}x^2.
\]
To prove the first bound in the lemma we show that if $|x|< \sqrt{\delta\cdot \frac{b-a}{32L^2}}$ then $g(x)< (\gamma-1)|x|+\delta$.

\noindent First we show it when $\gamma \le 2$. In this case, for all $x$ such that $|x|< \sqrt{\delta\cdot \frac{b-a}{32L^2}}$,
\[
|g(x)|\le 2\gamma\frac{8L^2}{(b-a)}x^2<\frac\gamma2\delta<\delta.
\]

\noindent Then, we turn to the case $\gamma \ge 2$. In this case, for all $x$ such that $|x|< \sqrt{\delta\cdot \frac{b-a}{32L^2}}$,
\[
|g(x)|\le 2\gamma\frac{8L^2}{b-a}x^2<\frac\gamma2|x|<(\gamma-1)|x|,
\]
where step the second inequality uses $|x|< \frac{b-a}{32L^2}$, which follows since $|x|< \sqrt{\delta\cdot \frac{b-a}{32L^2}}$ and $\delta< \frac{b-a}{32L^2}$, and the last inequality uses $\gamma/2<(\gamma-1)$, since $\gamma\ge 2$.
This completes the proof of the first upper bound in the lemma.
The second upper bound in the lemma can be shown using similar steps.\hfill\qedsymbol

\section*{Acknowledgements}
The authors would like to thank Sivaraman Balakrishnan for helpful pointers on the Statistics literature.

\bibliographystyle{alpha-all}
\bibliography{biblio}

\begin{thebibliography}{BFFKRW01}

\bibitem[Aba20]{Abadie20}
Alberto Abadie.
\newblock Statistical nonsignificance in empirical economics.
\newblock {\em American Economic Review: Insights}, 2(2):193--208, 2020.

\bibitem[ADK15]{AcharyaDK15}
Jayadev Acharya, Constantinos Daskalakis, and Gautam Kamath.
\newblock Optimal testing for properties of distributions.
\newblock In {\em Advances in Neural Information Processing Systems 28}, NIPS
  '15, pages 3577--3598. Curran Associates, Inc., 2015.

\bibitem[ADOS17]{AcharyaDOS17}
Jayadev Acharya, Hirakendu Das, Alon Orlitsky, and Ananda~Theertha Suresh.
\newblock A unified maximum likelihood approach for estimating symmetric
  properties of discrete distributions.
\newblock In {\em Proceedings of the 34th International Conference on Machine
  Learning}, ICML '17, pages 11--21. JMLR, Inc., 2017.

\bibitem[AJOS14]{AcharyaJOS14c}
Jayadev Acharya, Ashkan Jafarpour, Alon Orlitsky, and Ananda~Theertha Suresh.
\newblock Sublinear algorithms for outlier detection and generalized closeness
  testing.
\newblock In {\em Proceedings of the 2014 IEEE International Symposium on
  Information Theory}, ISIT '14, pages 3200--3204, Washington, DC, USA, 2014.
  IEEE Computer Society.

\bibitem[BCG17]{BlaisCG17}
Eric Blais, Cl{\'e}ment~L. Canonne, and Tom Gur.
\newblock Distribution testing lower bounds via reductions from communication
  complexity.
\newblock In {\em Proceedings of the 32nd Computational Complexity Conference},
  CCC '17, pages 28:1--28:40, Dagstuhl, Germany, 2017. Schloss
  Dagstuhl--Leibniz-Zentrum fuer Informatik.

\bibitem[BD00]{Bhatia00}
Rajendra Bhatia and Chandler Davis.
\newblock A better bound on the variance.
\newblock {\em The american mathematical monthly}, 107(4):353--357, 2000.

\bibitem[Ber12]{Bernstein1912}
Serge Bernstein.
\newblock {\em Sur l'ordre de la meilleure approximation des fonctions
  continues par des polyn{\^o}mes de degr{\'e} donn{\'e}}, volume~4.
\newblock Hayez, imprimeur des acad{\'e}mies royales, 1912.

\bibitem[BFFKRW01]{BatuFFKRW01}
Tu\u{g}kan Batu, Eldar Fischer, Lance Fortnow, Ravi Kumar, Ronitt Rubinfeld,
  and Patrick White.
\newblock Testing random variables for independence and identity.
\newblock In {\em Proceedings of the 42nd Annual IEEE Symposium on Foundations
  of Computer Science}, FOCS '01, pages 442--451, Washington, DC, USA, 2001.
  IEEE Computer Society.

\bibitem[BFRSW00]{BatuFRSW00}
Tu\u{g}kan Batu, Lance Fortnow, Ronitt Rubinfeld, Warren~D. Smith, and Patrick
  White.
\newblock Testing that distributions are close.
\newblock In {\em Proceedings of the 41st Annual IEEE Symposium on Foundations
  of Computer Science}, FOCS '00, pages 259--269, Washington, DC, USA, 2000.
  IEEE Computer Society.

\bibitem[BFRSW13]{BatuFRSW13}
Tu\u{g}kan Batu, Lance Fortnow, Ronitt Rubinfeld, Warren~D. Smith, and Patrick
  White.
\newblock Testing closeness of discrete distributions.
\newblock {\em Journal of the ACM}, 60(1):4:1--4:25, 2013.

\bibitem[BS87]{BergerS87}
James~O. Berger and Thomas Sellke.
\newblock Testing a point null hypothesis: The irreconcilability of p values
  and evidence.
\newblock {\em Journal of the American Statistical Association},
  82(397):112--122, 1987.

\bibitem[BV15]{BhattacharyaV15}
Bhaswar Bhattacharya and Gregory Valiant.
\newblock Testing closeness with unequal sized samples.
\newblock In {\em Advances in Neural Information Processing Systems 28}, NIPS
  '15, pages 2611--2619. Curran Associates, Inc., 2015.

\bibitem[BW18]{BalakrishnanW18}
Sivaraman Balakrishnan and Larry Wasserman.
\newblock Hypothesis testing for high-dimensional multinomials: A selective
  review.
\newblock {\em The Annals of Applied Statistics}, 12(2):727--749, 2018.

\bibitem[Can20a]{Canonne20}
Cl{\'e}ment~L. Canonne.
\newblock A survey on distribution testing: Your data is big. but is it blue?
\newblock {\em Theory of Computing}, (9):1--100, 2020.

\bibitem[Can20b]{CanonneLearningNote20}
Clément~L. Canonne.
\newblock A short note on learning discrete distributions.
\newblock {\em CoRR}, abs/2002.11457, 2020.

\bibitem[CDVV14]{ChanDVV14}
Siu~On Chan, Ilias Diakonikolas, Gregory Valiant, and Paul Valiant.
\newblock Optimal algorithms for testing closeness of discrete distributions.
\newblock In {\em Proceedings of the 25th Annual ACM-SIAM Symposium on Discrete
  Algorithms}, SODA '14, pages 1193--1203, Philadelphia, PA, USA, 2014. SIAM.

\bibitem[CFGMS21]{ChakrabortyFGMS21}
Sourav Chakraborty, Eldar Fischer, Arijit Ghosh, Gopinath Mishra, and Sayantan
  Sen.
\newblock Exploring the gap between tolerant and non-tolerant distribution
  testing.
\newblock {\em CoRR}, abs/2110.09972, 2021.

\bibitem[DGKPP21]{DiakonikolasGKPP21}
Ilias Diakonikolas, Themis Gouleakis, Daniel~M Kane, John Peebles, and Eric
  Price.
\newblock Optimal testing of discrete distributions with high probability.
\newblock In {\em Proceedings of the 53nd Annual ACM Symposium on the Theory of
  Computing}, STOC '21, New York, NY, USA, 2021. ACM.

\bibitem[DGPP18]{DiakonikolasGPP18}
Ilias Diakonikolas, Themis Gouleakis, John Peebles, and Eric Price.
\newblock Sample-optimal identity testing with high probability.
\newblock In {\em Proceedings of the 45th International Colloquium on Automata,
  Languages, and Programming}, ICALP '18, pages 41:1--41:14, 2018.

\bibitem[DGPP19]{DiakonikolasGPP19}
Ilias Diakonikolas, Themis Gouleakis, John Peebles, and Eric Price.
\newblock Collision-based testers are optimal for uniformity and closeness.
\newblock {\em Chicago Journal of Theoretical Computer Science}, 1:1--21, 2019.

\bibitem[DK16]{DiakonikolasK16}
Ilias Diakonikolas and Daniel~M. Kane.
\newblock A new approach for testing properties of discrete distributions.
\newblock In {\em Proceedings of the 57th Annual IEEE Symposium on Foundations
  of Computer Science}, FOCS '16, pages 685--694, Washington, DC, USA, 2016.
  IEEE Computer Society.

\bibitem[DKN15]{DiakonikolasKN15a}
Ilias Diakonikolas, Daniel~M. Kane, and Vladimir Nikishkin.
\newblock Testing identity of structured distributions.
\newblock In {\em Proceedings of the 26th Annual ACM-SIAM Symposium on Discrete
  Algorithms}, SODA '15, pages 1841--1854, Philadelphia, PA, USA, 2015. SIAM.

\bibitem[DKW18]{DaskalakisKW18}
Constantinos Daskalakis, Gautam Kamath, and John Wright.
\newblock Which distribution distances are sublinearly testable?
\newblock In {\em Proceedings of the 29th Annual ACM-SIAM Symposium on Discrete
  Algorithms}, SODA '18, pages 2747--2764, Philadelphia, PA, USA, 2018. SIAM.

\bibitem[Gol16]{Goldreich16}
Oded Goldreich.
\newblock The uniform distribution is complete with respect to testing identity
  to a fixed distribution.
\newblock {\em Electronic Colloquium on Computational Complexity (ECCC)},
  23(15), 2016.

\bibitem[GR00]{GoldreichR00}
Oded Goldreich and Dana Ron.
\newblock On testing expansion in bounded-degree graphs.
\newblock {\em Electronic Colloquium on Computational Complexity (ECCC)},
  7(20), 2000.

\bibitem[HJW16]{HanJW16}
Yanjun Han, Jiantao Jiao, and Tsachy Weissman.
\newblock Minimax rate-optimal estimation of divergences between discrete
  distributions.
\newblock In {\em Proceedings of the 2016 International Symposium on
  Information Theory and Its Applications}, ISITA '16, pages 256--260,
  Washington, DC, USA, 2016. IEEE Computer Society.

\bibitem[Ing94]{Ingster94}
Yuri~Izmailovich Ingster.
\newblock Minimax detection of a signal in $\ell_p$ metrics.
\newblock {\em Journal of Mathematical Sciences}, 68(4):503--515, 1994.

\bibitem[Ing00]{Ingster00}
Yuri~Izmailovich Ingster.
\newblock On testing a hypothesis which is close to a simple hypothesis.
\newblock {\em Teoriya Veroyatnostei i ee Primeneniya}, 45:356--368, 2000.

\bibitem[JHW18]{JiaoHW18}
Jiantao Jiao, Yanjun Han, and Tsachy Weissman.
\newblock Minimax estimation of the $ {L}_1$ distance.
\newblock {\em IEEE Transactions on Information Theory}, 64(10):6672--6706,
  2018.

\bibitem[JVHW17]{JiaoVHW17}
Jiantao Jiao, Kartik Venkat, Yanjun Han, and Tsachy Weissman.
\newblock Minimax estimation of functionals of discrete distributions.
\newblock {\em IEEE Transactions on Information Theory}, 61(5):2835--2885,
  2017.

\bibitem[Kam18]{Kamath18}
Gautam Kamath.
\newblock {\em Modern Challenges in Distribution Testing}.
\newblock PhD thesis, Massachusetts Institute of Technology, September 2018.

\bibitem[KOPS15]{KamathOPS15}
Sudeep Kamath, Alon Orlitsky, Venkatadheeraj Pichapati, and Ananda~Theertha
  Suresh.
\newblock On learning distributions from their samples.
\newblock In {\em Proceedings of the 28th Annual Conference on Learning
  Theory}, COLT '15, pages 1066--1100, 2015.

\bibitem[Mar92]{Markoff1892}
Wladimir Markoff.
\newblock Ober polynome, die in einem gegebenen intervalle moglichst wenig von
  null abweichen.
\newblock {\em Ann., 77, 213-258 (1892)(translation and condenstation by J.
  Grossman of Russian article published}, 1892.

\bibitem[OSW16]{OrlitskySW16}
Alon Orlitsky, Ananda~Theerta Suresh, and Yihong Wu.
\newblock Optimal prediction of the number of unseen species.
\newblock {\em Proceedings of the National Academy of Sciences},
  113(47):13283--13288, 2016.

\bibitem[Pan08]{Paninski08}
Liam Paninski.
\newblock A coincidence-based test for uniformity given very sparsely sampled
  discrete data.
\newblock {\em IEEE Transactions on Information Theory}, 54(10):4750--4755,
  2008.

\bibitem[PRR06]{ParnasRR06}
Michal Parnas, Dana Ron, and Ronitt Rubinfeld.
\newblock Tolerant property testing and distance approximation.
\newblock {\em Journal of Computer and System Sciences}, 72(6):1012--1042,
  2006.

\bibitem[RL16]{RaoL16}
Calyampudi~Radhakrishna Rao and Miodrag~M Lovric.
\newblock Testing point null hypothesis of a normal mean and the truth: 21st
  century perspective.
\newblock {\em Journal of Modern Applied Statistical Methods}, 15(2):2--21,
  2016.

\bibitem[Roc74]{Rockafellar74}
R~Tyrrell Rockafellar.
\newblock {\em Conjugate Duality and Optimization}.
\newblock SIAM, 1974.

\bibitem[Rub12]{Rubinfeld12}
Ronitt Rubinfeld.
\newblock Taming big probability distributions.
\newblock {\em XRDS}, 19(1):24--28, 2012.

\bibitem[Val11]{Valiant11}
Paul Valiant.
\newblock Testing symmetric properties of distributions.
\newblock {\em SIAM Journal on Computing}, 40(6):1927--1968, 2011.

\bibitem[VV10a]{ValiantV10a}
Gregory Valiant and Paul Valiant.
\newblock A {CLT} and tight lower bounds for estimating entropy.
\newblock {\em Electronic Colloquium on Computational Complexity (ECCC)},
  17(179), 2010.

\bibitem[VV10b]{ValiantV10b}
Gregory Valiant and Paul Valiant.
\newblock Estimating the unseen: A sublinear-sample canonical estimator of
  distributions.
\newblock {\em Electronic Colloquium on Computational Complexity (ECCC)},
  17(180), 2010.

\bibitem[VV11a]{ValiantV11a}
Gregory Valiant and Paul Valiant.
\newblock Estimating the unseen: An $n/\log n$-sample estimator for entropy and
  support size, shown optimal via new {CLT}s.
\newblock In {\em Proceedings of the 43rd Annual ACM Symposium on the Theory of
  Computing}, STOC '11, pages 685--694, New York, NY, USA, 2011. ACM.

\bibitem[VV11b]{ValiantV11b}
Gregory Valiant and Paul Valiant.
\newblock The power of linear estimators.
\newblock In {\em Proceedings of the 52nd Annual IEEE Symposium on Foundations
  of Computer Science}, FOCS '11, pages 403--412, Washington, DC, USA, 2011.
  IEEE Computer Society.

\bibitem[VV14]{ValiantV14}
Gregory Valiant and Paul Valiant.
\newblock An automatic inequality prover and instance optimal identity testing.
\newblock In {\em Proceedings of the 55th Annual IEEE Symposium on Foundations
  of Computer Science}, FOCS '14, pages 51--60, Washington, DC, USA, 2014. IEEE
  Computer Society.

\bibitem[VV17]{ValiantV17a}
Gregory Valiant and Paul Valiant.
\newblock An automatic inequality prover and instance optimal identity testing.
\newblock {\em SIAM Journal on Computing}, 46(1):429--455, 2017.

\bibitem[WY16]{WuY16}
Yihong Wu and Pengkun Yang.
\newblock Minimax rates of entropy estimation on large alphabets via best
  polynomial approximation.
\newblock {\em IEEE Transactions on Information Theory}, 62(6):3702--3720,
  2016.

\bibitem[WY18]{WuY18}
Yihong Wu and Pengkun Yang.
\newblock {C}hebyshev polynomials, moment matching, and optimal estimation of
  the unseen.
\newblock {\em The Annals of Statistics}, 2018.

\bibitem[WY20]{WuY20}
Yihong Wu and Pengkun Yang.
\newblock Polynomial methods in statistical inference: Theory and practice.
\newblock {\em Foundations and Trends{\textregistered} in Communications and
  Information Theory}, 17(4):402--586, 2020.

\end{thebibliography}

\newpage

\appendix
\section{Details on the ``splitting'' operation}~\label{app:splitting}
Given an explicit reference distribution $q$ over $[n]$, we describe the splitting operation with respect to $q$, as introduced in~\cite{DiakonikolasK16}. 
For $i\in [n]$, let $a_i \eqdef 1+\lfloor{nq_i} \rfloor$ and $D \eqdef \{(i,j) : i\in [n], j\in[a_i]\}$.
The splitting operation with respect to $q$ maps any given distribution $p$ over $[n]$ to a new distribution $p^{S(q)}$ over the new domain $D$ such that the new distribution $p^{S(q)}$ assigns the probability $\frac{p_i}{a_i}$ to element $(i,j)$. 

We note a few properties of the splitting operation: 
\begin{enumerate}
    \item The new domain is at most twice as large: indeed, $|D|= \sum_{i=1}^n (1+\lfloor nq_i\rfloor) \leq \sum_{i=1}^n (nq_i+1) =2n$.
    \item If $q$ is known, then $m$ i.i.d.\ samples from an unknown distribution $p$ can be used to simulate the $m$ i.i.d.\ from $p^{S(q)}$, by (independently for each) mapping a sample $i\in[n]$ to a $(i,j)$, for $j$ chosen uniformly at random in $[a_i]$.
    \item The resulting distribution obtained by applying splitting operation w.r.t. $q$ on itself has small $\lp{2}$ norm,
    \[
        \norm{q^{S(q)}}_2^2 = \sum_{(i,j)\in D}(q^{S(q)}_{i,j})^2 = \sum_{i=1}^n \sum_{j\in[a_i]} \mleft(\frac{q_i}{a_i}\mright)^2
        = \sum_{i=1}^n \frac{q_i^2}{a_i}
        \leq \sum_{i=1}^n \frac{q_i}{n} = \frac{1}{n} \leq \frac{2}{|D|}\,.
    \]
    \item The pairwise $\lp{1}$ (and thus total variation) distances between any two distributions $p$ and $p'$ are preserved after the splitting operation, namely for any distributions $p,p'$ over $[n]$,
    \[
    \norm{p-p'}_1 = \norm{p^{S(q)}-p'^{S(q)}}_1\,;
    \]
    this follows from observing that
    \[
      \norm{p^{S(q)}-p'^{S(q)}}_1  = \sum_{i=1}^n \sum_{j\in[a_i]} |p^{S(q)}_{i,j}-p'^{S(q)}_{i,j}|
      = \sum_{i=1}^n \sum_{j\in[a_i]} \Big|\frac{p(i)}{a_i}-\frac{p'(i)}{a_i}\Big|
      = \sum_{i=1}^n |p_i-p'_i|.
    \]
\end{enumerate}

\subsection{Proofs of Lemmas~\ref{lem:ficon} and~\ref{lem:ficon3}}\label{app:conlem}

The following standard bound on the concentration of Poisson random variables will be useful:
\begin{theorem}\label{th:poicon}
Let $X\sim \operatorname{Poi}(\lambda)$ be a Poisson random variable for some $\lambda> 0 $. Then for any $x>0$,
\[
\Pr\big[ |X-\lambda| \ge x \big] \le \exp\!\bigg(-\Omega\!\Big(\min\Big\{x,\frac{x^2}{\lambda}\Big\}\Big)\bigg).
\]
\end{theorem}

Next, we prove Lemma~\ref{lem:ficon}.
\begin{proof}
First we prove the lemma for the simpler of the two cases when $m< n$ and then later for $m\ge n$.

\paragraph{Proof for the regime $m< n$:} 
\begin{align*}
\Pr\big[ \widehat f_{i} > t f_i \big] &\overset{\rm(a)}=  \Pr\big[ \max\{X_i+Y_i,1\} > t \max\{m(p_i+q_i),1\}  \big] \\
&=  \Pr\big[ X_i+Y_i > t \max\{m(p_i+q_i),1\}  \big] \\
&=  \Pr\big[ X_i+Y_i - m(p_i+q_i) > t \max\{m(p_i+q_i),1\} -m(p_i+q_i) \big] \\
&=  \Pr\big[ X_i+Y_i - m(p_i+q_i) > (t-1) \max\{m(p_i+q_i),1\}\big] \\
&\overset{\rm(b)}\le \exp\mleft(-\Omega\mleft(\min\mleft\{(t-1)\max\{m(p_i+q_i),1\},\frac{(t-1)^2\max\{m(p_i+q_i)^2,1\}}{\max\{m(p_i+q_i),1\}}\mright\}\mright)\mright)\\
&\le \exp\mleft(-\Omega\mleft(\min\mleft\{(t-1),(t-1)^2\mright\}\mright)\mright)\\
&\overset{\rm(c)}=\exp\mleft(-\Omega\mleft(t\mright)\mright),
\end{align*}
here (a) uses definition of $f_i$, (b) uses the fact that $\tilde X+\tilde Y\sim \operatorname{Poi}(mp_i+mq_i)$ and the Poisson concentration bound in Theorem~\ref{th:poicon}, and (c) uses $t>1$. 

Next we prove the lemma for the other case when $m\ge n$.

\paragraph{Proof for the regime $m\ge n$:} 
We first bound $\Pr\big[ \widehat f_{i} > t f_i \big]$ by sum of three different terms, and then later we bound each term one by one.
  \begin{align*}
       &\Pr\big[ \widehat f_{i} > t f_i \big] \\
       &=\Pr\Bigg[ \max\bigg\{\frac{|\tilde{X}_i-\tilde{Y}_i|}{\sqrt{m/n}},\frac{\tilde{X}_i+\tilde{Y}_i}{m/n}, 1 \bigg\} > t f_i\Bigg]\\
       &\overset{\rm(a)}\le \Pr\Bigg[ \frac{|\tilde{X}_i-\tilde{Y}_i|}{\sqrt{m/n}} > t f_i\Bigg]+\Pr\Bigg[ \frac{\tilde{X}_i+\tilde{Y}_i}{m/n} > tf_i \Bigg]\\
       &\overset{\rm(b)}\le \Pr\Bigg[ \frac{|\tilde{X}_i- mp_i|+ |\tilde{Y}_i-m q_i|+|mp_i-mq_i|}{\sqrt{m/n}} > tf_i \Bigg]+\Pr\Bigg[ \frac{\tilde{X}_i+\tilde{Y}_i}{m/n} > tf_i \Bigg]\\
       &= \Pr\Big[{|\tilde{X}_i- mp_i|+ |\tilde{Y}_i-m q_i|} > t f_i{\sqrt{m/n}} -m |p_i-q_i| \Big]+\Pr\Big[ {\tilde{X}_i+\tilde{Y}_i} > t f_i \frac{m}{n} \Big]\\
       &\le \Pr\Big[{\max\{|\tilde{X}_i- mp_i|, |\tilde{Y}_i-m q_i|\}} > \frac{t f_i{\sqrt{m/n}} -m |p_i-q_i| }2\Big]+\Pr\Big[ {\tilde{X}_i+\tilde{Y}_i} > t f_i \frac{m}{n} \Big]\\
       &    \overset{\rm(c)}\le \Pr\Big[ {|\tilde{X}_i- mp_i|} > \frac{tf_i\sqrt{m/n}-m |p_i-q_i|}{2} \Big]+\Pr\Big[ {|\tilde{Y}_i-m q_i|} > \frac{tf_i\sqrt{m/n}+-m |p_i-q_i|}{2} \Big]\\
       &+\Pr\Big[ {\tilde{X}_i+\tilde{Y}_i} > tf_i\frac{m}{n} \Big],\numberthis \label{eq:upcon1}
   \end{align*}
where inequalities (a) and (c) use union bound and (b) uses triangle inequality.

\noindent To obtain an upper bound we bound each term in the above equation. 
Next, we bound the first term. 
\begin{align*}
&\Pr\Big[ {|\tilde{X}_i- mp_i|} > \frac{t}{2}f_i  \sqrt{m/n}-\frac{m |p_i-q_i|}{2} \Big] \\
&    \overset{\rm(a)}= 
\Pr\Big[ {|\tilde{X}_i- mp_i|} > \frac{t}{2} \max\{m|p_i-q_i|, \sqrt{mn}(p_i+q_i), \sqrt{m/n}\}-\frac{m |p_i-q_i|}{2} \Big]\\
&    \overset{\rm(b)}\le  
\Pr\Big[ {|\tilde{X}_i- mp_i|} > \frac{tm|p_i-q_i| + t \sqrt{mn} (p_i+q_i)+ t \sqrt{m/n}}{6} -\frac{m |p_i-q_i|}{2} \Big]\\
&= 
\Pr\Big[ {|\tilde{X}_i- mp_i|} > \frac{(t-3)m|p_i-q_i| + t \sqrt{mn} (p_i+q_i)+ t \sqrt{m/n}}{6} \Big]\\
&    \overset{\rm(c)}\le 
\Pr\Big[ {|\tilde{X}_i- mp_i|} > \frac{t \sqrt{mn} (p_i+q_i)+ t \sqrt{m/n}}{6} \Big]\\
&    \overset{\rm(d)}\le 2\exp\Bigg(- \Omega\bigg(\min\bigg\{\frac{t \sqrt{mn} (p_i+q_i)+ t \sqrt{m/n}}{6}, \frac{(t \sqrt{mn} (p_i+q_i)+ t \sqrt{m/n})^2}{36 mp_i}\bigg\}\bigg)\Bigg)\\
& \le 2\exp\Bigg(- \Omega\bigg(\min\bigg\{\frac{ t \sqrt{m/n}}6, \frac{t^2 mn (p_i+q_i)^2+t^2\cdot (m/n)}{36 mp_i}\bigg\}\bigg)\Bigg)\\
&\le 2\exp\Bigg(- \Omega\bigg(\min\bigg\{\frac{ t \sqrt{m/n}}6, \frac{t^2 n p_i}{36}+\frac{t^2}{ 36 np_i}\bigg\}\bigg)\Bigg)\\
&    \overset{\rm(e)}\le \exp(-\Omega(t)),
\end{align*}
where (a) uses the definition of $f_i$, (b) follows from the fact that $\frac{a+b+c}{3}\le\max\{a,b,c\}$, inequality (c) uses $t\ge 3$, inequality (d) uses the fact that $\tilde X\sim \operatorname{Poi}(mp_i)$ and the Poisson concentration bound in Theorem~\ref{th:poicon}, and finally (e) uses $m\ge n$ and the fact that $x+1/x \ge 2$ for any $x\ge 0$.

\noindent Note that because of the symmetry the above bound will also apply on the second term, namely 
\[
Pr\Big[ {|\tilde{X}_i- mq_i|} > \frac{t}{2}f_i  \sqrt{m/n}-\frac{m |p_i-q_i|}{2} \Big]\le \exp(-\Omega(t)).
\]

\noindent Next we bound the last term in Equation~\eqref{eq:upcon1} to complete the proof of the first concentration inequality. 
\begin{align*}
&\Pr\Big[ {\tilde{X}_i+\tilde{Y}_i} > tf_i\frac{m}{n} \Big]  \\
&\overset{\rm(a)}=\Pr\Big[ {\tilde{X}_i+\tilde{Y}_i} > t\max\big\{\frac{m^{3/2}}{n^{1/2}}|p_i-q_i|,m(p_i+q_i),\frac{m}{n}\big\} \Big]  \\
&\le \Pr\Big[ {\tilde{X}_i+\tilde{Y}_i} > t\max\big\{m(p_i+q_i),\frac{m}{n}\big\} \Big]  \\
&\overset{\rm(b)}\le \Pr\Big[ {\tilde{X}_i+\tilde{Y}_i}-m(p_i+q_i) > \frac{tm(p_i+q_i)+t\frac{m}{n}}{2}-m(p_i+q_i) \Big]\\
&= \Pr\Big[ {\tilde{X}_i+\tilde{Y}_i}-m(p_i+q_i) > \frac{(t-2)m(p_i+q_i)+t\frac{m}{n}}{2} \Big] \\
&\overset{\rm(c)}\le 2\exp\Bigg(- \Omega\bigg(\min\bigg\{\frac{(t-2)m(p_i+q_i)+t\frac{m}{n}}{2} , \frac{((t-2)m(p_i+q_i)+t\frac{m}{n})^2}{4m(p_i,q_i)} \bigg\}\bigg)\Bigg)\\
&\overset{\rm(d)}\le 2\exp\Bigg(- \Omega\bigg(\min\bigg\{t\frac{m}{2n} , \frac{((t-2)m(p_i+q_i)+(t-2)\frac{m}{n})^2}{4m(p_i+q_i)} \bigg\}\bigg)\Bigg)\\
&\overset{\rm(e)}\le 2\exp\Bigg(- \Omega\bigg(\min\bigg\{t\frac{m}{2n} , \frac{(t-2)^2m(p_i+q_i)}4+ \frac{(t-2)^2}{4 m (p_i+q_i)} \bigg\}\bigg)\Bigg)\\
&\overset{\rm(f)}\le 2\exp\Bigg(- \Omega\bigg(\min\bigg\{\frac{t}{2} , \frac{(t-2)^2 }2 \bigg\}\bigg)\Bigg)\\
&\le \exp(-\Omega(t)),
\end{align*}
where (a) uses the definition of $f_i$, (b) uses the fact that $\frac{a+b}{2}\le\max\{a,b\}$, inequality (c) uses the fact that $\tilde X+\tilde Y\sim \operatorname{Poi}(mp_i+mq_i)$ and the Poisson concentration bound in Theorem~\ref{th:poicon}, inequality (d) uses $t\ge 3$, inequality (e) uses the fact that for $a,b>0$, $(a+b)^2\le a^2+b^2$, and finally (f) uses $m\ge n$ and the fact that $x+1/x \ge 2$ for any $x\ge 0$.

Combining the bounds on all three terms in Equation~\eqref{eq:upcon1} proves the Lemma.
\end{proof}

Finally, we prove Lemma~\ref{lem:ficon3}.
\begin{proof}

First we prove the lemma for the simpler of the two cases when $m< n$ and then later for $m\ge n$.

\paragraph{Proof for the regime $m< n$:} Based on the value of $f_i$, we further divide in two cases, and for both cases we show one by one that the concentration inequality holds.
\begin{enumerate}
    \item \textbf{Case 1:} $f_i=1\ge m(p_i+q_i)$.\\
    Since $\widehat f_i\ge 1$ then $\widehat f_i\ge f_i$, hence for any $t>1$ we have $\Pr\big[ \widehat f_{i} < \frac{1}{t} f_i \big]=0$.
    
    \item \textbf{Case 2:} $f_i=m(p_i+q_i)\ge 1$.\\
    Note for $t\ge f_i$ the inequality $ \Pr\Big[\widehat f_i <  \frac{f_i}t\Big]=0$ trivially holds as $\widehat f_i \ge 1$.\\
    For $t< f_i$ 
    \begin{align*}
    \Pr\mleft[ \widehat f_{i} < \frac{f_i}{t} f_i \mright] &\overset{\rm(a)}=  \Pr\mleft[ \max\{X_i+Y_i,1\} < \frac{f_i}{t} \mright] \\
    &\le  \Pr\mleft[ X_i+Y_i< \frac{f_i}{t}   \mright] \\
    &=  \Pr\mleft[ m(p_i+q_i) - X_i+Y_i  < m(p_i+q_i)     -\frac{f_i}{t} \mright] \\
    &\overset{\rm(b)}=  \Pr\mleft[ m(p_i+q_i) - X_i+Y_i  < \mleft(1-\frac{1}{t}\mright) f_i  \mright] \\
    &\overset{\rm(c)}=  \Pr\mleft[ m(p_i+q_i) - X_i+Y_i  < \frac{f_i}2  \mright] \\
    &\overset{\rm(d)}\le     \exp\mleft(-\Omega\mleft(\min\mleft\{\frac{f_i}2,\frac{f^2_i}{4(m(p_i+q_i)}\mright\}\mright)\mright)\\
    &\overset{\rm(e)}= \exp\mleft(-\Omega\mleft(\min\mleft\{\frac{f_i}2,\frac{f_i^2}{4}\mright\}\mright)\mright)\\
    &\overset{\rm(f)}= \exp\mleft(-\Omega\mleft(t\mright)\mright)\\
    \end{align*}
    where (a) uses the definition of $\widehat f_i$, (b) uses the fact that $f_i =m(p_i+q_i)$, inequality (c) uses $t\ge 2$, inequality (d) uses the fact that $\tilde X+\tilde Y\sim \operatorname{Poi}(mp_i+mq_i)$ and the Poisson concentration bound in Theorem~\ref{th:poicon}, inequality (e) uses $f_i = m(p_i+q_i)$, and finally (f) uses $f_i\ge  t$ and $t>1$.
\end{enumerate}

Next we prove the lemma for the other case when $m\ge n$.

\paragraph{Proof for the regime $m\ge n$:} Based on the value of $f_i$, we further divide in three cases, and for each of the three cases we show one by one that the concentration inequality holds.
\begin{enumerate}
    \item \textbf{Case 1:} $f_i= 1 \ge \max\{n \cdot  (p_i+q_i), \sqrt{mn}\cdot |p_i-q_i|]\}$.\\ 
    In this case $p_i+q_i\le \frac{1}{n} $ and $|p_i-q_i|\le \frac{1}{\sqrt{mn}}$.\\
    Since $\widehat f_{i}\ge 1$, then
    \[
    \Pr\Big[\widehat f_i <  \frac{f_i}t \Big]\le \Pr\Big[\widehat f_i <  f_i \Big] = 0. 
    \]
    
    \item \textbf{Case 2:} $f_i =n \cdot  (p_i+q_i)  \ge \max\{1, \sqrt{mn}\cdot |p_i-q_i|]\}$.\\
    In this case $p_i+q_i=\frac{f_i}{n}$ and $|p_i-q_i|\le \frac{f_1}{\sqrt{mn}}$.\\
    Note for $t\ge f_i$ the inequality $ \Pr\Big[\widehat f_i <  \frac{f_i}t\Big]=0$ trivially holds as $\widehat f_i \ge 1$. For $t< f_i$ 
    \begin{align*}
        \Pr\Big[\widehat f_i <  \frac{f_i}t\Big] &\overset{\rm(a)}\le \Pr\Big[ \frac{n}{m}(X_i+Y_i) \le \frac{f_i}{t} \Big]\\
        &= \Pr\Big[(X_i+Y_i) \le \frac{mf_i}{nt} \Big]\\
        &= \Pr\Big[m(p_i+q_i)- (X_i+Y_i) \ge m(p_i+q_i)-\frac{mf_i}{nt} \Big]\\
        &\overset{\rm(b)}= \Pr\Big[m(p_i+q_i)- (X_i+Y_i) \ge \frac{mf_i}{n}\big(1-\frac{1}{t}\big) \Big]\\
        &\overset{\rm(c)}\le  \Pr\Big[m(p_i+q_i)- (X_i+Y_i) \ge \frac{mf_i}{2n} \Big]\\
        &\overset{\rm(d)}\le 2\exp\Bigg(- \Omega\bigg(\min\bigg\{\frac{mf_i}{2n} , \frac{m^2f_i^2}{4n^2\cdot m(p_i+q_i)} \bigg\}\bigg)\Bigg)\\
        &\overset{\rm(e)}\le 2\exp\Bigg(- \Omega\bigg(\min\bigg\{\frac{mf_i}{2n} , \frac{m^2 f_i n(p_i+q_i)}{4n^2\cdot m(p_i+q_i)} \bigg\}\bigg)\Bigg)\\
        &\le 2\exp\Bigg(- \Omega\bigg(\frac{mf_i}{4n} \bigg)\Bigg)\\
        &\overset{\rm(f)}\le \exp\big(- \Omega(t)\big),
    \end{align*}
    where (a) uses the definition of $\widehat f_i$, (b) uses the fact that $p_i+q_i=f_i/n$ for Case 2, inequality (c) uses $t\ge 2$, inequality (d) uses the fact that $\tilde X+\tilde Y\sim \operatorname{Poi}(mp_i+mq_i)$ and the Poisson concentration bound in Theorem~\ref{th:poicon}, inequality (e) uses $f_i = n(p_i+q_i)$, and finally (f) uses $m\ge n$ and $f_i>t$.
    
    \item \textbf{Case 3:} $f_i = \sqrt{mn}\cdot |p_i-q_i| \ge \max\{1, n \cdot  (p_i+q_i)\}$.\\ 
    In this case $ |p_i-q_i|= f_i/\sqrt{mn}$ and $(p_i+q_i)\le f_i/n$.\\
    Note for $t\ge f_i$ the inequality $ \Pr\Big[\widehat f_i <  \frac{f_i}t\Big]=0$ trivially holds as $\widehat f_i \ge 1$. For $t< f_i$
    \begin{align*}
        \Pr\Big[\widehat f_i <  \frac{f_i}t\Big] &\overset{\rm(a)}\le \Pr\Big[ \frac{|\tilde X_i-\tilde Y_i|}{\sqrt{m/n}} \le \frac{f_i}{t} \Big]\\
        &\overset{\rm(b)}\le \Pr\Big[ \frac{|mp_i-mq_i|-|\tilde X_i-mp_i-\tilde Y_i+mq_i|}{\sqrt{m/n}} \le \frac{f_i}{t} \Big]\\
        &= \Pr\Big[ \frac{|\tilde X_i-mp_i-\tilde Y_i+mq_i|}{\sqrt{m/n}} \ge \sqrt{mn}|p_i-q_i|- \frac{f_i}{t} \Big]\\
        &\overset{\rm(c)}= \Pr\Big[ \frac{|\tilde X_i-mp_i-\tilde Y_i+mq_i|}{\sqrt{m/n}} \ge f_i\big(1- \frac{1}{t}\big) \Big]\\
        &\overset{\rm(d)}\le \Pr\Big[ {|\tilde X_i-mp_i-\tilde Y_i+mq_i|} \ge \sqrt{\frac{m}{n}}\cdot\frac{ f_i}2 \Big]\\
        &\overset{\rm(e)}\le \Pr\Big[ \max\{|\tilde X_i-mp_i|,|\tilde Y_i-mq_i|\} \ge \sqrt{\frac{m}{n}}\cdot\frac{ f_i}4 \Big]\\
        &\overset{\rm(f)}\le \Pr\Big[ |\tilde X_i-mp_i|\ge \sqrt{\frac{m}{n}}\cdot\frac{ f_i}4 \Big]+ \Pr\Big[ |\tilde Y_i-mq_i|\ge \sqrt{\frac{m}{n}}\cdot\frac{ f_i}4 \Big]\\
        &\overset{\rm(g)}\le 2\exp\Bigg(- \Omega\bigg(\min\bigg\{\sqrt{\frac{m}{n}}\cdot\frac{ f_i}4 , \frac{mf_i^2}{16n\cdot mp_i} \bigg\}\bigg)\Bigg)+2\exp\Bigg(- \Omega\bigg(\min\bigg\{\sqrt{\frac{m}{n}}\cdot\frac{ f_i}4 , \frac{mf_i^2}{16n\cdot mq_i} \bigg\}\bigg)\Bigg)\\
        &\overset{\rm(h)}\le \exp\big(- \Omega(t)\big),
    \end{align*}
    where (a) uses the definition of $\widehat f_i$, inequality (b) follows from the triangle inequality, inequality (c) uses the fact that $\sqrt{mn}|p_i-q_i|=f_i$ for Case 3, inequality (d) uses $t\ge 2$, inequality (e) uses the fact that $|x-y|\le 2 \max\{|x|,|y|\}$, inequality (f) uses union bound, inequality (g)  uses the Poisson concentration bound in Theorem~\ref{th:poicon}, and finally (h) uses $m\ge n$, $f_i> n(p_i+q_i)$ and $f_i> t$.\qedhere
\end{enumerate}
\end{proof}

\section{Missing proofs from Section~\ref{sec:lb}}
    \label{app:lb:details}
\subsection{Proof of Theorem~\ref{thm:testinglb-to-moments}}
\label{sec:main-lb-conversion}

Let $\mathcal P^{(1)}$ be the joint distribution of $n$ independent copies of $U$ and $\mathcal P^{(2)}$ be the joint distribution of $n$ independent copies of $U'$, respectively. 
For $j=1,2$, let $(U_1^{(j)},\dots ,U_n^{(j)})\sim  \mathcal P^{(j)}$.
Define random vectors $D^{(j)}$ as the $\ell_1$-normalizations of these vectors:
\[D^{(j)} \eqdef \mleft(\frac{U^{(j)}_1}{\sum_{i=1}^n U_i^{(j)}},\dots,\frac{U_n^{(j)}}{\sum_{i=1}^n U_i^{(j)}}\mright).\]
Since $U_i^{(j)}\ge 0$, the vectors $D^{(j)}$ are distributions.

Let $N^{(j)}\sim\poi(m\sum_i U^{(j)}_i)$.
Let $(C^{(j)}_1,\dots,C_n^{(j)})$ be the collection of random variables, whose joint distribution conditioned on $(U^{(j)}_1,\dots,U^{(j)}_n)$ and $N^{(j)}$ is the following multinomial distribution,
\[
(C^{(j)}_1,\dots,C^{(j)}_n)\Big| \mleft(\mleft(U^{(j)}_1,\dots,U^{(j)}_n\mright), N^{(j)}\mright) \sim \text{Mult}\mleft(\mleft(\frac{U^{(j)}_1}{\sum U^{(j)}_i},\dots ,\frac{U^{(j)}_n}{\sum U^{(j)}_i}\mright), N^{(j)}\mright)\equiv \text{Mult}\mleft(D^{(j)}, N^{(j)}\mright).
\]
It follows that we can use $(C^{(j)}_1,\dots,C^{(j)}_n)$ to generate up to $N^{(j)}$ samples from $D^{(j)}$.
Observe that for $i\in [n]$, conditioned on the $U_i^{(j)}$'s, $C_i^{(j)}\sim\poi(m U_i^{(j)})$ are independent Poisson random variables.

Define the events
\[
E^{(1)} = \mleft\{ \mleft|\sum_{i=1}^n U^{(1)}_i -1\mright|\le \frac{1}{10}\mright\}\,\bigcap\, \mleft\{\sum_{i=1}^n \mleft|U^{(1)}_i -\frac{1}{n}\mright|\le 10\ve_1\mright\},
\]
and 
\[
E^{(2)} = \mleft\{ \mleft|\sum_{i=1}^n U^{(2)}_i -1\mright|\le \frac{\ve_2}{10}\mright\}\,\bigcap\, \mleft\{\sum_{i=1}^n \mleft|U^{(2)}_i -\frac{1}{n}\mright|\ge \frac{9\ve_2}{10}\mright\}.
\]
We bound the probability of the complement events $\bar{E^{(1)}}$ and $\bar{E^{(2)}}$. The following general lemma will be useful, which we prove using Chebyshev's inequality.
\begin{lemma}\label{lem:chebtigbou}
Let $X_1,\dots,X_n$ be $n$ i.i.d.\ random variables over $[a,b]$ for some $0\le a < b$ and $\E[X_i]=\mu$. Then 
\[
\Pr\mleft[ \mleft|\sum_i X_i - n\mu \mright| \ge \sqrt{10 n\mu(b-a)} \mright]\le \frac{1}{10}
\]
\end{lemma}
\begin{proof} 
The following bound on the variance of a random variable will be useful. This bound has been proved in many previous works including~\cite{Bhatia00}. We provide the proof for completeness.
\begin{theorem}
\label{thm:var-bound}
Let $X$ be any random variable over $[a,b]$, then $\Var(X) \le (b-\E[X])(\E[X]-a)$.
\end{theorem}
\begin{proof}
Let $Y = \frac{X-a}{b-a}$. Note $Y\in [0,1]$.
Then 
\[
\Var(Y) = \E[Y^2]-(\E[Y])^2 \le \E[Y]-(\E[Y])^2=\E[Y](1-\E[Y]),
\]
where we used the fact $Y^2\le Y$, since $Y\in [0,1]$.
Then using the relations $\Var(Y) = \frac{\Var(X)}{(b-a)^2}$ and $\E[Y] = \frac{\E [X]-a}{b-a}$ completes the proof.
\end{proof}
\noindent From Theorem~\ref{thm:var-bound} 
\[
\Var(X_i) \le (b-\mu)(\mu-a)\le (b-a)\mu,
\]
where we used $0\le a\le\mu\le b$.
Then $\Var(\sum X_i)\le n\mu(b-a)$. From Chebyshev's inequality
\[
\Pr\mleft[ \mleft|\sum_i X_i - \sum_i \E[ X_i ] \mright| \ge \sqrt{10 n\mu(b-a)} \mright]\le \frac{1}{10}.%
\]
\end{proof}

\noindent First we bound the probability of $\bar{E^{(1)}}$. Using the union bound
\begin{align*}
\Pr[\bar{E^{(1)}}]\le \Pr\mleft[\mleft|\sum_{i=1}^n U^{(1)}_i -\frac{1}{n}\mright|\ge \frac{1}{10}\mright] + \Pr\mleft[\sum_{i=1}^n \mleft|U^{(1)}_i -\frac{1}{n}\mright|\ge 10\ve_1\mright]. 
\end{align*}
We next upper bound both these terms, starting with the former. 
Note that $U^{(1)}_i\in[a,b]$ and  $\E\mleft[U^{(1)}_i\mright] =\frac{1}{n}$. Applying Lemma~\ref{lem:chebtigbou}, we obtain
\[
\Pr\mleft[ \mleft|\sum_i U^{(1)}_i- 1\mright|\ge \sqrt{10(b-a)} \mright]\le \frac{1}{10},
\]
which, since $10(b-a)\le \frac{\ve_2^2}{100}\le \frac{1}{100}$, upper bounds the first term in the expression.

\noindent We now bound the second term using linearity of expectations and Markov's inequality,
\[
\Pr\mleft[\sum_{i=1}^n \mleft|U^{(1)}_i -\frac{1}{n}\mright|\ge 10\ve_1\mright] \le \frac{\sum_i\E\mleft[\mleft|U^{(1)}_i -\frac{1}{n}\mright|\mright]}{10\ve_1}= \frac{n\E\mleft[\mleft|U -\frac{1}{n}\mright|\mright]}{10\ve_1}\le \frac{1}{10}.
\]
Combining the bounds on both terms we get:
\[
\Pr[\bar{E^{(1)}}]\le\frac{2}{10}.
\]

\noindent Next, we bound the probability of $\bar{E^{(2)}}$.
\begin{align*}
\Pr[\bar{E^{(2)}}]\le \Pr\mleft[\mleft|\sum_{i=1}^n U^{(2)}_i -1\mright|\ge \frac{\ve_2}{10}\mright] + \Pr\mleft[\sum_{i=1}^n \mleft|U^{(2)}_i -\frac{1}{n}\mright|\le \frac{9\ve_2}{10}\mright]. 
\end{align*}
Again, note that $U^{(2)}_i\in[a,b]$ and  $\E\mleft[U^{(2)}_i\mright] =\frac{1}{n}$. Applying Lemma~\ref{lem:chebtigbou}, we obtain
\[
\Pr\mleft[ \mleft|\sum_i U^{(2)}_i- 1\mright|\ge \sqrt{10(b-a)} \mright]\le \frac{1}{10},
\]
which, since $10(b-a)\le \frac{\ve_2^2}{100}$, upper bounds the first term in the expression.

Next, we bound the second term.
Recall that random variables $U^{(2)}_i$ are independent copies of $U'$. Since $U'\in[a,b]$ and $a\le \frac{1}{n}\le b$, then $0\le \mleft|U' -\frac{1}{n}\mright|\le (b-a)$. Applying Lemma~\ref{lem:chebtigbou}, we obtain
\[
\Pr\mleft[ \sum_i \mleft|U^{(2)}_i -\frac{1}{n}\mright| \le n\E\mleft[\mleft|U' -\frac{1}{n}\mright|\mright] - \sqrt{10 n (b-a) \E\mleft[\mleft|U' -\frac{1}{n}\mright|\mright] } \mright]\le \frac{1}{10}.
\]
Since $10(b-a)\le \frac{\ve_2^2}{100}\le \frac{\ve_2}{100}$ and $\E\mleft[\mleft|U' -\frac{1}{n}\mright|\mright]\ge  \frac{\ve_2}{n}$, then
\[
n \E\mleft[\mleft|U' -\frac{1}{n}\mright|\mright] - \sqrt{10 n (b-a) \E\mleft[\mleft|U' -\frac{1}{n}\mright|\mright] } \ge
n\E\mleft[\mleft|U' -\frac{1}{n}\mright|\mright] - \frac{n\E\mleft[\mleft|U' -\frac{1}{n}\mright|\mright]}{10}\ge \frac{9\ve_2}{10} .
\]
Combining the above two equations we get
\[
\Pr\mleft[ \sum_i\mleft|U^{(2)}_i -\frac{1}{n}\mright| \le \frac{9\ve_2}{10}  \mright]\le \frac{1}{10}.
\]

\noindent Combining the bounds on both terms:
\[
\Pr[\bar{E^{(2)}}]\le\frac{2}{10}.
\]

\noindent Note that the event $E^{(1)} $ implies that 
\begin{align*}
\|D^{(1)}-\unif_n\|_1 &= \sum_{i=1}^n \mleft|\frac{U^{(1)}_i}{\sum_{i} U^{(1)}_i} -\frac{1}{n}\mright|\\
&\le\sum_{i=1}^n \mleft(\mleft|\frac{U^{(1)}_i}{\sum_{i} U^{(1)}_i} -\frac{1}{n\sum_{i} U^{(1)}_i}\mright| + \mleft|\frac{1}{n}-\frac{1}{n\sum_{i} U^{(1)}_i}\mright|\mright)\\
&=\frac{\sum_i|U^{(1)}_i-1/n|}{\sum_{i} U^{(1)}_i} + \mleft|1-\frac{1}{\sum_{i} U^{(1)}_i}\mright|\\
&=\frac{\sum_i|U^{(1)}_i-1/n|}{\sum_{i} U^{(1)}_i} + \frac{\mleft|\sum_{i} U^{(1)}_i-1\mright|}{\sum_{i} U^{(1)}_i}\\
&\le 2\cdot \frac{\sum_i|U^{(1)}_i-1/n|}{\sum_{i} U^{(1)}_i} \\
&\le 2\cdot\frac{10\ve_1}{1-1/10}\le \frac{200\ve_1}9\le 25\ve_1.\numberthis \label{eq:appdis1}
\end{align*}

\noindent Similarly, event $E^{(2)}$ implies that
\begin{align*}
\|D^{(2)} -\unif_n\|_1 &= \sum_{i=1}^n \mleft|\frac{U^{(2)}_i}{\sum_{i} U^{(2)}_i} -\frac{1}{n}\mright|\\
&\ge\sum_{i=1}^n \mleft(\mleft|\frac{U^{(2)}_i}{\sum_{i} U^{(2)}_i} -\frac{1}{n\sum_{i} U^{(2)}_i}\mright| - \mleft|\frac{1}{n}-\frac{1}{n\sum_{i} U^{(2)}_i}\mright|\mright)\\
&=\frac{\sum_i|U^{(2)}_i-1/n|}{\sum_{i} U^{(2)}_i} - \mleft|1-\frac{1}{\sum_{i} U^{(2)}_i}\mright|\\
&=\frac{\sum_i|U^{(2)}_i-1/n|}{\sum_{i} U^{(2)}_i} - \frac{\mleft|\sum_{i} U^{(2)}_i-1\mright|}{\sum_{i} U^{(2)}_i}\\
&\ge\frac{(9\ve_2/10)-(\ve_2/10)}{1+\ve_2/10}\ge \frac{8\ve_2}{11}\ge \ve_2/2, \numberthis \label{eq:appdis2}
\end{align*}
where we used the triangle inequality and $\ve_2\le 1$.

For $j=1,2$, let ${\mathcal C}^{(j)}$ denote the distribution of $(C^{(j)}_1\dots C^{(j)}_n)$, and
let ${\mathcal P}^{(j)}_{|E^{(j)}}$, ${\mathcal D}^{(j)}_{|E^{(j)}}$ and ${\mathcal C}^{(j)}_{|E^{(j)}}$ denote the distributions of $(U^{(j)}_1\dots U^{(j)}_n)$,   $D^{(j)}$, and $(C^{(j)}_1\dots C^{(j)}_n)$, respectively conditioned on the event $E^{(j)}$.

In light of Equations~\eqref{eq:appdis1} and~\eqref{eq:appdis2}, to prove the lemma it suffices to show no tester using $m/2$ samples from $p$ correctly identifies whether $p= {\mathcal D}^{(1)}_{|E^{(1)}}$ or $p={\mathcal D}^{(2)}_{|E^{(2)}}$ with probability $\ge 4/5$. To prove by contradiction, suppose there is such a tester $\mathcal{T}$.

Event $E^{(j)}$ implies that $\sum U^{(j)}_i\ge \frac{9}{10}$.
Hence, for any given $(U^{(j)}_1\dots U^{(j)}_n)\sim {\mathcal P}^{(j)}_{|E^{(j)}}$, we have 
\[
\Pr(N^{(j)}\ge m/2) = \Pr[\poi(m\sum U^{(j)}_i) \ge m/2] \ge \Pr[\poi(0.9 m) \ge m/2], 
\]
which is at least $0.95$ for $m$ larger than some absolute constant $c$.

If $N^{(j)} \ge m/2$, then we can simulate $m/2$ samples from $D^{(j)}$ using $(C^{(j)}_1\dots C^{(j)}_n)$ and use the tester $\mathcal{T}$ on these $m/2$ samples.
Hence, using this tester we can correctly identify whether $(C_1\dots C_n) \sim {\mathcal C}^{(1)}_{|E^{(1)}}$ or $(C_1\dots C_n) \sim {\mathcal C}^{(2)}_{|E^{(2)}}$ with probability
$
\ge 0.95 \mleft( \frac{4}{5}\mright)= 0.76
$.

Next, we show that the TV distance between the distributions ${\mathcal C}^{(1)}_{|E^{(2)}}$ and ${\mathcal C}^{(2)}_{|E^{(2)}}$ is small. 
\begin{align*}
&\operatorname{TV}\mleft( {\mathcal C}^{(1)}_{|E^{(1)}}, {\mathcal C}^{(2)}_{|E^{(2)}}\mright)\le \Pr[\bar{E^{(1)}}] + \operatorname{TV}\mleft( {\mathcal C}^{(1)}, {\mathcal C}^{(2)}\mright) +\Pr[\bar{E^{(2)}}]\\
&=
\operatorname{TV}\mleft( \E_{\mathcal{P}^{(1)}} \Paren{\poi (mU_1^{(1)}), \ldots, \poi(m U_n^{(1)})}, \E_{\mathcal{P}^{(2)}} \Paren{\poi (mU_1^{(1)}), \ldots, \poi(mU_n^{(1)})}\mright)+\Pr[\bar{E^{(1)}}]+\Pr[\bar{E^{(2)}}]\\
&\le n \operatorname{TV}\mleft( \E \poi (mU), \,\E \poi (mU')\mright)+\Pr[\bar{E^{(1)}}]+\Pr[\bar{E^{(2)}}]\\
&\le n\cdot  \frac 1{20n} +\Pr[\bar{E^{(1)}}]+\Pr[\bar{E^{(2)}}]\le \frac{9}{20}. 
\end{align*}
This
implies that for any tester the probability of correctly distinguishing ${\mathcal C}^{(1)}_{|E^{(1)}}$ and  ${\mathcal C}^{(2)}_{|E^{(2)}}$ is at most $\frac{1+9/20}{2}= 29/40=0.725$, which is a contradiction since $0.725 < 0.76$.

\subsection{Proof of Lemma~\ref{lemma:value:dual}}
\label{sec:strong-duality}
For any finite subset $\mathcal S$ of $\subset [-B,A]$, consider the optimization problem
\begin{align}
    \max\,  \E |U'|\mbox{ s.t. } & \E |U|\le \frac{\ve_1}{2} \mbox{ and } \nonumber\\
    & \E U^i = \E U'^i, i = 1, \ldots, L, \mbox{ and }\nonumber\\
    & U, U' \in \mathcal S \;,\label{eq:opnewsub}
\end{align}
and its dual
\begin{align*}
    \min \frac{\ve_1}2 \alpha + z_1 + z_2 \mbox{ s.t. } & z_1 + \sum_{i = 1}^L c_i p_i (x) \geq \Abs{x} \mbox{ for all $x \in \mathcal S$,}\\
    & \alpha \Abs{x} \geq \sum_{i = 1}^L c_i p_i (x) - z_2 \mbox{ for all $x \in \mathcal S$}, \mbox{ and } \\
    & \alpha \geq 0 \;. \numberthis \label{eq:dualsub}
\end{align*}

For a given $\mathcal S$, let $P_{\mathcal S}$ and $D_{\mathcal{S}}$ be the optimal solution to the primal and dual, respectively.
Since $\mathcal{S}$ is finite, the distribution of both $U$ $U'$ is a finite vector of size $\mathcal{S}$, then from the strong duality for linear programming we have $P_{\mathcal S}=D_{\mathcal{S}}$. 

Let $P$ and $D$ denote the value of optimal solution of~\eqref{eq:opnew} and~\eqref{eq:dual}, respectively. From the weak duality we have $P\le D$. 

For any $\mathcal S$, the corresponding optimization problem~\eqref{eq:opnewsub} can be obtained by imposing the constraints $\Pr[U \in [-B,A]\setminus \mathcal S] = \Pr[U' \in [-B,A]\setminus \mathcal S] =0$ in~\eqref{eq:opnew}.
Since upon imposing the additional constrains, the value of the optimal solution in~\eqref{eq:opnew} would only decrease, hence $P_{\mathcal S}\le P$.
This implies for all finite subset $\mathcal S$ of $\subset [-B,A]$, the following holds
$P\ge P_{\mathcal S} = D_{\mathcal S}$. 

Let $\mathcal S_{\delta} = \{x: x= -B+k\delta \text{ for }k\in\{0,1,\dots,\lfloor\frac{A+B}{\delta} \rfloor \}\}$. Observe that for all $\delta>0$, $\mathcal S_\delta$ is a finite subset of $ [-B,A]$.
Taking the supremum over $\mathcal S_{\delta}$ as $\delta\to 0$ ,
\[
P \ge \sup_{\delta> 0} D_{\mathcal S_\delta}.
\]
Using the continuity of functions $x^i$ and $|x|$ and elementary real analysis it can be verified that
\[
\sup_{\delta> 0} D_{\mathcal S_\delta} = D.
\]
Hence, we get $P\ge D$. Combining this with $P\le D$ proves the lemma.\hfill\qedhere

\section{Instance-optimal tolerant testing}
    \label{app:io}
In this appendix, we establish our ``instance-optimal'' tolerant identity bounds (\autoref{th:identity:io:informal}); that is, sample complexity bounds parameterized by the reference distribution $q$ itself, instead of the domain size $n$. We do so by establishing separately the lower bound (\autoref{th:lbU:io}) and upper bound (\autoref{th:ubU:io}) parts of the statement, in~\autoref{app:lbU:io} and~\autoref{app:ubU:io}.\footnote{As mentioned earlier, we slightly abuse the $\tilde{\mathcal{O}}$ and $\tilde{\Omega}$ notation in those two statements to also hide logarithmic factors in $n$, not just in the argument.}\medskip

In order to formally state our results, a few definitions will be useful. For any distribution $q$ over a set $[n]$ and any subset $S\subseteq [n]$, let $\mas{q}{S} \eqdef  \max_{i\in S} q_i$, and $\rat{q}{S}\eqdef  \left\lfloor \frac{q(S)}{2\mas{q}{S}} \right\rfloor$, where as usual $q(S) = \sum_{i\in S}q_i$. Moreover, for any $x\ge 0$, let $q_{-x}$ denote the vector obtained by iteratively removing the smallest entries from $q$ and stopping just before the sum of the removed elements exceed $x$. Finally, recall that for any integer $t\geq 1$, $\unif_t$ denotes the uniform distribution over $[t]$.

\subsection{Lower bound}
    \label{app:lbU:io}

Given a reference distribution $q$, $0\le \ve_1< \ve_2 \le 1$ and $\delta>0$, let $\samcom{q}{\ve_1}{\ve_2}{\delta}$ denote the minimum number of samples (in the Poissonized sampling model) any tester requires from an unknown distribution $p$ to correctly distinguish between $\|p-q\|_1\le \ve_1$ and $\|p-q\|_1\ge \ve_2$ with probability at least $1-\delta$.

Our main tool will be the following theorem relating the lower bound of testing uniform distributions to the lower bound of testing for general $q$.

\begin{theorem}
For any distribution $q$ over $[n]$, subset $S\subseteq[n]$ such that $\rat{q}{S}\geq 1$, $0\le \ve_1< \ve_2 $ and $\delta>0$, 
\[
\samcom{q}{\ve_1}{\ve_2}{\delta} \ge \frac{4}{q(S)}\cdot\samcom{\unif_{ \rat{q}{S}}}{\,\tfrac{4\ve_1}{q(S)}}{\tfrac{4\ve_2}{q(S)}}{\delta}.
\]
\end{theorem}
\begin{proof}
For any $S\subseteq[n]$ and any distribution $p$ over $[\rat{q}{S}]$, we derive a distribution $p^{\text{new}}$ over $[n]$ such that $\norm{p^{\text{new}}-q}_1 = \frac{q(S)}{4} \|p-\unif_{\rat{q}{S}}\|_1 $ and, for any $m>0$, $\operatorname{Poi}(m)$ samples from $p$ can be used to generate $\operatorname{Poi}\left({4m}/{q(S)}\right)$ samples from $p^{\text{new}}$, with the knowledge of just $q$ and not of $p$ and $p^{\text{new}}$. Then the statement of the theorem follows, since to distinguish $\|p-\unif_{\rat{q}{S}}\|_1 \le \frac{4\ve_1}{q(S)}$ and $\|p-\unif_{\rat{q}{S}}\|_1 \le \frac{4\ve_2}{q(S)}$ one can use samples from $p$ to generate samples from $p^{\text{new}}$ and test $\norm{p^{\text{new}}-q}_1 \le \ve_1$ vs $\norm{p^{\text{new}}-q}_1 \ge \ve_2$ instead.\medskip

The rest of the proof focuses on obtaining such a distribution $p^{\text{new}}$.

Fix any $S$ is a subset of $[n]$ such that $\rat{q}{S}\geq 1$, and let $\mas{q}{S} = \max_{i\in S} q_i$.
Consider a partition of $S$ into $S_1,S_2,\dots,S_\ell$ (for some $\ell\geq 1$) such that $q(S_j)\in[\mas{q}{S}, 2\mas{q}{S})$ for every $j\in[\ell]$. Such partition exists, and can be obtained by a greedy construction.
Since the mass of each $S_j$ is less than $2\mas{q}{S}$, we also have that $\ell> \frac{q(S)}{2\mas{q}{S}}\ge \rat{q}{S}$.

Given a distribution $p$ on $[\rat{q}{S}]$, we define $p^{\text{new}}$ as follows. For every $1\leq j\leq \rat{q}{S}$ (i.e., the first $\rat{q}{S}$ subsets), each element  $i\in S_j$ is given the probability  
\begin{equation}
p^{\text{new}}_i= q_i+q_i\cdot \frac{q(S)}{4q(S_j)}\mleft(p_j-\frac{1}{\rat{q}{S}}\mright)
\end{equation}
while every $i\in \bigcup_{j > \rat{q}{S}}S_j$ is assigned the probability $p^{\text{new}}_i= q_i$. Next we show that $p^{\text{new}}$ is indeed a distribution, and can be sampled given samples from $p$ and knowledge of $q$ only. 
\begin{itemize}
    \item For $i\notin \bigcup_{j\le \rat{q}{S}}S_j$, since $p^{\text{new}}_i= q_i$, we have $p^{\text{new}}_i\ge 0$.
Moreover, the count $\operatorname{Poi}(mq_i)$ can clearly be generated with the knowledge of $q, m$ only.
    \item For $j\le \rat{q}{S}$ and any element $i\in S_j$, note that
\[
p^{\text{new}}_i= q_i+q_i\frac{q(S)}{4q(S_j)}\mleft(p_j-\frac{1}{\rat{q}{S}}\mright)\ge q_i-q_i\frac{q(S)}{4q(S_j)}\cdot \frac{1}{\rat{q}{S}} \ge q_i-q_i\frac{q(S)}{4\mas{q}{S}}\cdot\frac{1}{\rat{q}{S}}\ge q_i-q_i\frac{\rat{q}{S}+1}{2\rat{q}{S}}\ge 0.
\]
Using the standard properties of Poisson processes it is easy to see that for any $j\le \rat{q}{S}$, a sample from $\operatorname{Poi}(mp_j)$ and the knowledge of $q$ suffice to generate $\operatorname{Poi}(\frac{4m}{q(S)}p^{\text{new}}_i)$ samples for each $i\in S_j$.
    \item Finally,
\[
\sum_{i\in[n]}{p^{\text{new}}_i} 
= \sum_{i\in[n]} q_i +\sum_{j\le \rat{q}{S}}q(S_j)\frac{q(S)}{4q(S_j)}\mleft(p_j-\frac{1}{\rat{q}{S}}\mright)
=1 +\frac{q(S)}{4}\mleft(\sum_{j\le \rat{q}{S}}p_j-1\mright)=1.
\]
This shows that $p^{\text{new}}$ is indeed a distribution.
\end{itemize}
To complete the proof, it only remains to relate the $\ell_1$ distances, which we do now.
\[
\sum_{i\in[n]}|{p^{\text{new}}_i}-q_i| = \sum_{j\le \rat{q}{S}}\sum_{i\in S_j} q_i\frac{q(S)}{4q(S_j)}\mleft|p_j-\frac{1}{\rat{q}{S}}\mright|=\frac{q(S)}{4}\sum_{j\le \rat{q}{S}}\mleft|p_j-\frac{1}{\rat{q}{S}}\mright|=\frac{q(S)}{4}\|p-\unif_{\rat{q}{S}}\|_1 \,,
\]
as claimed.
\end{proof}
The lower bound from Theorem~\ref{thm:lb-main} implies that for any subset $S$ such that  $\rat{q}{S}\ge 2$ and $0\le \frac{4\ve_1}{q(S)}<\frac{4\ve_2}{q(S)}\le c$ (for some universal constant $c>0$), 
\[
\samcom{\unif_{ \rat{q}{S}}}{\,\frac{4\ve_1}{q(S)}}{\frac{4\ve_2}{q(S)}}{4/5} =\Omega\mleft(\frac{q(S)\cdot\rat{q}{S}}{\log n}\Big(\frac{\ve_1}{\ve_2^2}\Big)+\frac{q(S)^2\rat{q}{S}}{\log n}\Big(\frac{\ve_1}{\ve_2^2}\Big)^2\mright).
\]

Combining this bound with the above theorem and using the observation that if $\rat{q}{S}\le 2$ then $\rat{q}{S}-2\le 0$ and $\rat{q}{S}\cdot q(S)\le 0$, we get:
\begin{corollary}\label{cor:instoptlb}
For any distribution $q$ over $[n]$, $0\le \ve_1< \ve_2 < 1 $, and some universal constant $c>0$,
\[
\samcom{q}{\ve_1}{\ve_2}{4/5} \ge \Omega\mleft(\max_{S\subseteq[n]: q(S) \ge 4\ve_2/c,}\mleft((\rat{q}{S}-2)\cdot\frac{1}{\log n}\Big(\frac{\ve_1}{\ve_2^2}\Big)+(\rat{q}{S}-2)\cdot\frac{1}{\log n}\Big(\frac{\ve_1}{\ve_2^2}\Big)^2\mright)\mright).
\]
\end{corollary}
We would like to relate this bound, which involves a maximum over subsets $S$ and the quantity $\rat{q}{S}$, to a more interpretable expression involving the $0$- and $1/2$-quasinorms of $q$, as stated in~\autoref{th:identity:io:informal}. Our next two lemmas will allow us to do so.
\begin{lemma}
For any $x\in(0,1)$ such that $\norm{q_{-x}}_0> 1$, there exists some $i^*\in [n] $ such that for $A \eqdef  \{j:q_j\le q_i^*\}$ the following holds:
(i) $\max_{j\in A}q_j = q_{i^*}$,  (ii) $q(A)\ge x$ and (iii) $\frac{q(A)}{q_{i^*}}\ge \frac{\norm{q_{-x}}_0-1}{\ln (1/x)}$. 
\end{lemma}
\begin{proof}
Without loss of generality, we can assume that the distribution $q=(q_1,q_2,...,q_n)$ is non-increasing, that is that $q_1\ge q_2\ge....\ge q_n$. This in particular implies that $\norm{q_{-x}}_0 = \max\{i : \sum_{j\ge i }q_j  \ge x\}$.

Note that for any $i^*\in [n]$ property (i) holds trivially, and property (ii) holds for any $i^*\le \norm{q_{-x}}_0$, as $q(A)=\sum_{j\in A}q_j= \sum_{j:q_j\le q_{i^*}}q_j\ge  \sum_{j:j\ge {i^*}}q_j\ge  \sum_{j:j\ge \norm{q_{-x}}_0}q_j\ge x$.

To complete the proof, it thus suffices to establish (iii) for some $i^* \le \norm{q_{-x}}_0$; that is, to find $i^* \le \norm{q_{-x}}_0$ such that
\[
\frac{\sum_{j\in A}q_j}{q_{i^*}}\ge \frac{\sum_{j\ge i^*}q_j}{q_{i^*}}\ge \frac{\norm{q_{-x}}_0-1}{\ln (1/x)}.
\]
Since if $\frac{\norm{q_{-x}}_0-1}{\ln (1/x)}\le 1$ this trivially holds for every $i^* \le \norm{q_{-x}}_0$, in what follows we assume $\frac{\norm{q_{-x}}_0-1}{\ln (1/x)}> 1$.

Suppose by contradiction that for every $i\le  \norm{q_{-x}}_0$, we have $\frac{\sum_{j\ge i}q_j}{q_i}\le  \frac{\norm{q_{-x}}_0-1}{\ln (1/x)} $; equivalently, that $q_i\ge \paren{\sum_{j\ge i}q_j} \frac{\ln (1/x)}{\norm{q_{-x}}_0-1} $.
Hence, $\sum_{j\ge i+1}q_j = \paren{\sum_{j\ge i}q_j}-q_i\le \paren{\sum_{j\ge i}q_j}\mleft(1-\frac{\ln (1/x)}{\norm{q_{-x}}_0-1}\mright)$. By induction, this gives
\[
\sum_{j\ge \norm{q_{-x}}_0}q_j
\le 
\Paren{\sum_{j\ge 1}q_j}\mleft(1-\frac{\ln (1/x)}{\norm{q_{-x}}_0-1}\mright)^{\norm{q_{-x}}_0-1}=\mleft(1-\frac{\ln (1/x)}{\norm{q_{-x}}_0-1}\mright)^{\norm{q_{-x}}_0-1} < e^{-\ln (1/x)} =x\,,
\]
where we used that $ 0\le 1-u< e^{-u} $ for $u\in (0,1)$.
But, by definition $\sum_{j\ge \norm{q_{-x}}_0}q_j \geq x$: this is a contradiction, concluding the proof.
\end{proof}
\begin{lemma}
For any distribution $q$ over $[n]$ and $x\in(0,1)$, 
$
\max_{S\subseteq[n]: q(S) \ge x}\rat{q}{S}\ge \frac{\|q_{-x}\|_{1/2}}{(\log (n/x)+1)^2}-4.
$
\end{lemma}
\begin{proof}
Let $D=\max\{S:q(S)\le x\}$ be a largest subset that has mass $\le x$ under $q$. From the definition of $D$ it is not hard to see that we can choose $D$ such that $\min_{i\in [n]\setminus D}q_i\ge \max_{i\in D}q_i$, and 
\[x < \sum_{i\in D}q_i +\min_{i\in [n]\setminus D}q_i\le (D+1)\min_{i\in [n]\setminus D}q_i\le n\min_{i\in [n]\setminus D}q_i,\] 
therefore, $\min_{i\in [n]\setminus D}q_i> \frac{x}{n} $.

Next, we perform a ``bucketing'' of the remaining elements; that is, we partition $[n]\setminus D$ in subsets so that the probability assigned by $q$ to any two elements in the same subset differ by at most a factor 2.
Let $\ell=\lfloor \log(\frac{n}{x}) \rfloor+1$ and for $j\in[\ell]$, let 
\[
D_{j} \eqdef  \mleft\{i\in [n]\setminus D :  q_i\in\mleft( \frac{1}{2^{j}},\frac{1}{2^{j-1}}\mright] \mright\}.
\]
We can write
\begin{equation}
    \label{eq:12norm:1}
\|q_{-x}\|_{1/2} = \mleft(\sum_{j\in[\ell]}\sum_{i\in D_j} q_i^{1/2}\mright)^{2} \le \ell^2\max_{j\in[\ell]}\mleft(\sum_{i\in D_j} q_i^{1/2}\mright)^{2}\le \ell^2\max_{j\in[\ell]}|D_j|\cdot q(D_j),
\end{equation}
the last inequality being Cauchy--Schwarz. 
Let $j^*\in [\ell]$ be the index maximizing the term on the left, and choose $S=D_{j^*}\cup D$. Since $D_{j^*}$ is non-empty, from the definition of $D$, we have $q(S)\ge x$.
Further, 
\begin{equation}
    \label{eq:12norm:2}
\rat{q}{S} =\mleft\lfloor \frac{q(S)}{2\mas{q}{S}} \mright\rfloor= \mleft\lfloor \frac{\sum_{i\in D_{j^*}\cup D} q_i}{2\max_{i\in D_{j^*}\cup D}q_i} \mright\rfloor \ge \mleft\lfloor \frac{\sum_{i\in D_{j^*}} q_i}{2\max_{i\in D_{j^*}}q_i} \mright\rfloor\ge  \mleft\lfloor \frac{|D_{j^*}|\cdot 2^{-j^*}}{2\cdot 2^{-j^*+1}} \mright\rfloor\ge \frac{|D_{j^*}|}4-1. 
\end{equation}
Putting together~\eqref{eq:12norm:1} and~\eqref{eq:12norm:2}, we get
\[
\frac{\|q_{-x}\|_{1/2}}{\ell^2}\le |D_{j^*}|\cdot q(D_{j^*})\le |D_{j^*}|\cdot q(S) 
\le (\rat{q}{S}+4)\mleft(q(S)\mright)
\le \rat{q}{S}\cdot q(S)+4\,
\]
which concludes the proof.
\end{proof}

Combining Corollary~\ref{cor:instoptlb} and the above two lemmas for $x=4\ve_2/c$, for some universal constant $c>0$, any distribution $q$ over $[n]$, $0\le \ve_1< \ve_2 < c/4 $,
\[
\samcom{q}{\ve_1}{\ve_2}{4/5} \ge \Omega\mleft(\mleft(\frac{\norm{q_{-4\ve_2/c}}_0}{\ln (c/4\ve_2)}-3\mright)\cdot\frac{1}{\log n}\Big(\frac{\ve_1}{\ve_2^2}\Big)+\mleft(\frac{||q_{-4\ve_2/c}||_{1/2}}{\log^2 (nc/4\ve_2)}-6\mright)\cdot\frac{1}{\log n}\Big(\frac{\ve_1}{\ve_2^2}\Big)^2\mright).
\]
By combining the above lower bound with previously known lower bound $\Omega(\frac{||q_{-c'\ve_2}||_{2/3}-1}{\ve_2^2})$  for non-tolerant identity testing from~\cite{ValiantV14}, where $c'>0$ is an absolute constant, we obtain:
\begin{theorem}
    \label{th:lbU:io}
For any distribution $q$ over $[n]$, $0\le \ve_1< \ve_2 < c/4 $, for some universal constant $c>0$,
\[
\samcom{q}{\ve_1}{\ve_2}{4/5} \ge \tilde \Omega\mleft({\norm{q_{-4\ve_2/c}}_0}\Big(\frac{\ve_1}{\ve_2^2}\Big)+{||q_{-4\ve_2/c}||_{1/2}}\Big(\frac{\ve_1}{\ve_2^2}\Big)^2+\frac{||q_{-4\ve_2/c}||_{2/3}}{\ve_2^2}\mright)-\mathcal{O}\mleft(\frac{1}{\ve_2^2}\mright)\,.
\]
\end{theorem}
\noindent This establishes the lower bound part of~\autoref{th:identity:io:informal}.

\subsection{Upper bound}
    \label{app:ubU:io}
The proof of our instance-optimal upper bound follows the same outline as~\cite[Proposition 2.12]{DiakonikolasK16}, yet the extension to tolerant testing requires a significantly more detailed argument.

\begin{theorem}[Identity testing]\label{th:ubU:io}
Let $q$ be a known reference distribution and $p$ be an . There is a computationally efficient algorithm with the following guarantee. Given a known reference distribution $q$ over $[n]$, as well as parameters $\ve_1,\ve_2$ such that $0 \leq  \ve_2\leq 1$ and $0\leq \ve_1 \leq c \frac{\ve_2}{\log (n/\ve_2)}$ (where $c>0$ is an absolute constant), the algorithm takes
\[
       \tilde{\mathcal{O}}\mleft( \norm{q_{-\ve_2/20}}_{\frac{1}{2}}\Big(\frac{\ve_1}{\ve_2^2}\Big)^2+\norm{q_{-\ve_2/20}}_0\Big(\frac{\ve_1}{\ve_2^2}\Big)+\frac{ \norm{q_{-\ve_2/20}}_{\frac{2}{3}}}{\ve_2^{2}} \mright). 
\]
 samples from an unknown distribution $p$ over $[n]$ and distinguishes between $\norm{p-q}_1\leq \ve_1$ and $\norm{p-q}_1\geq \ve_2$ with probability at least $4/5$.
\end{theorem}
\begin{proof}
Let $D=\arg\!\max\{S:\sum_{i\in S}q_i\le \ve_2/20\}$ be a largest subset that has mass $\le \ve_2/20$.
From the definition of $D$ it is easy to see that we can choose $D$ such that $\min_{i\in [n]\setminus D}q_i\ge \max_{i\in D}q_i$, and thus
\[
\frac{\ve_2}{20} < \sum_{i\in D}q_i +\min_{i\in [n]\setminus D}q_i\le (D+1)\min_{i\in [n]\setminus D}q_i\le n\min_{i\in [n]\setminus D}q_i,
\] 
which implies that $\min_{i\in [n]\setminus D}q_i> \frac{\ve_2}{20n}$. Moreover, given the full description of $q$, this set $D$ can be efficiently computed. 
We then (as in the lower bound section) ``bucket'' the remaining elements $[n]\setminus D$ into disjoint subsets so that the probability assigned by $q$ to any two elements in the same subset differ by at most a factor 2. That is, for $\ell\eqdef \lfloor \log(\frac{20n}{\ve_2}) \rfloor+1$ and $j\in[\ell]$, we let 
\begin{equation}
    \label{eq:bucketing:dj}
D_{j} \eqdef \mleft\{i\in [n]\setminus D :  q_i\in\mleft( \frac{1}{2^{j}},\frac{1}{2^{j-1}}\mright] \mright\}.
\end{equation}
We denote by $p^j$ and $q^j$ the conditional distributions on $D_j$ induced by $p$ and $q$, respectively. With this in hand, we get the following:
\begin{claim}
    \label{claim:uB:io:claim1}
If $\norm{p-q}_1\le \ve_1$, then all three conditions below hold simultaneously: 
\begin{enumerate}
    \item $p(D) \le \ve_1+\frac{\ve_2}{20}$ and
    \item for every $j\in[\ell]$, $|p(D_j)-q(D_j)|\le \ve_1$, and 
    \item for every $j\in[\ell]$
    , $\norm{p^j-q^j}_1\le \frac{2\ve_1}{q(D_j)}$. 
\end{enumerate}
\end{claim}
\begin{proof}
We prove the claim by showing that if any of the three conditions fails to hold then we must have $\norm{p-q}_1\ge \ve_1$.
Note that, by the triangle inequality,
\[
\norm{p-q}_1 %
=\sum_{i\in D}|p_i-q_i|+\sum_{j=1}^\ell\sum_{i\in D_j}|p_i-q_i|\ge 
|p(D)-q(D)|+\sum_{j=1}^\ell|p(D_j)-q(D_j)|.
\]
Recalling that $q(D) \leq \frac{\ve_2}{20}$, it is easy to see that if either of the first two conditions fails to hold, the above inequality implies that $\norm{p-q}_1\ge \ve_1$.
Turning to the third condition, we can write
\begin{align*}
\norm{p^j-q^j}_1&=\sum_{i\in D_j}\mleft|\frac{p_i}{p(D_j)}-\frac{q_i}{q(D_j)}\mright|
\le \sum_{i\in D_j}\mleft|\frac{p_i}{q(D_j)}-\frac{q_i}{q(D_j)}\mright| +\sum_{i\in D_j}\mleft|\frac{p_i}{p(D_j)}-\frac{p_i}{q(D_j)}\mright|\\
&=\frac{\sum_{i\in D_j}|p_i-q_i|+|p(D_j)-q(D_j)|}{q(D_j)}\le \frac{2\sum_{i\in D_j}|p_i-q_i|}{q(D_j)} \le \frac{2\norm{p-q}_1}{q(D_j)}\,,
\end{align*}
which shows that if the third item fails to hold for some $j$, then $\norm{p-q}_1>\ve_1$.
\end{proof}

The next claim then provides a qualitatively converse statement.
\begin{claim}
    \label{claim:uB:io:claim2}
Suppose that $p$ satisfies all three conditions below:
\begin{enumerate}
    \item $p(D) \le \frac{\ve_2}{5}$,
    \item for every $j\in[\ell]$, $|p(D_j)-q(D_j)|\le \frac{\ve_2}{10\ell}$, and
    \item for every $j\in[\ell]$ such that $q(D_j)\ge \frac{\ve_2}{5\ell}$, $\norm{p^j-q^j}_1\le \frac{\ve_2}{5\ell q(D_j)}$\,.
\end{enumerate}
Then, we have $\norm{p-q}_1\le \ve_2$. 
\end{claim}
\begin{proof}
Suppose that the three conditions hold. By the triangle inequality, we have
\begin{align*}
  \norm{p-q}_1&=\sum_{i\in D}|p_i-q_i|+\sum_{j=1}^\ell\sum_{i\in D_j}|p_i-q_i|\\ &\le p(D)+q(D)+\sum_{j=1}^\ell\sum_{i\in D_j}q(D_j)\mleft|\frac{p_i}{q(D_j)}-\frac{q_i}{q(D_j)}\mright| \\ 
  &\le p(D)+\frac{\ve_2}{20}+\sum_{j=1}^\ell\sum_{i\in D_j}q(D_j)\mleft|\frac{p_i}{p(D_j)}-\frac{q_i}{q(D_j)}\mright|+\sum_{j=1}^\ell\sum_{i\in D_j}q(D_j)\mleft|\frac{p_i}{p(D_j)}-\frac{p_i}{q(D_j)}\mright| \\ 
   &\le p(D)+\frac{\ve_2}{20}+\sum_{j=1}^\ell q(D_j)\norm{p^j-q^j}_1+\sum_{j=1}^\ell{|p(D_j)-q(D_j)|} \\
  &\le \frac{\ve_2}{5}+\frac{\ve_2}{20}+\mleft(\ell\cdot\frac{\ve_2}{5\ell}\cdot 2 + \sum_{j=1}^\ell q(D_j) \cdot \frac{\ve_2}{5\ell q(D_j)} \mright)+\ell\cdot\frac{\ve_2}{10\ell} < \ve_2\,,
\end{align*}
where we used the three conditions for the second-to-last inequality.
\end{proof}

Given the above claims, we can describe our testing algorithm. First, the algorithm computes the set $D$, the value $\ell$, and the bucketing of $[n]\setminus D$ into $D_1,\dots, D_\ell$. Then, it runs a total of (at most) $2\ell+1$ sub-tests, which we will detail momentarily:
\begin{enumerate}
    \item[(1)] Distinguish $p(D)< \frac{\ve_2}{20}+\ve_1$ (\textsf{accept}) from $p(D)\ge \frac{\ve_2}{5}$ (\textsf{reject}),
    \item[(2)] For every  $j\in [\ell]$, distinguish $|p(D_j)-q(D_j)|\le \ve_1$ (\textsf{accept}) from $|p(D_j)-q(D_j)|\ge \frac{\ve_2}{10\ell}$ (\textsf{reject}), and
    \item[(3)] For every $j\in[\ell]$ such that $q(D_j)\ge \frac{\ve_2}{5\ell}$,  distinguish $\norm{p^j-q^j}_1\le \frac{2\ve_1}{\ell q(D_j)}$ (\textsf{accept}) from $\norm{p^j-q^j}_1\ge \frac{\ve_2}{5\ell q(D_j)}$ (\textsf{reject}).
\end{enumerate}
If all the above testers accept, the overall tester accepts (i.e., outputs $\norm{p-q}_1\le \ve_1$); otherwise, it rejects (i.e., outputs $\norm{p-q}_1\ge \ve_2$).

From Claims~\ref{claim:uB:io:claim1} and~\ref{claim:uB:io:claim2}, it is not hard to see that if $\ve_1\le \ve_2/(40\ell)$ and all the above testers give correct outputs with probability at least $1-\frac{1}{5(2\ell+1)}$ each, then we correctly distinguish $\norm{p-q}_1\le \ve_1$ and $\norm{p-q}_1\ge \ve_2$ with probability at least $1-1/5=4/5$ (by a union bound). We now proceed to describe how those tests are implemented.

\begin{itemize}
    \item  Using $\mathcal O\mleft(\frac{\log(1/\ell)}{(\ve_2/\ell)^2}\mright)=O\mleft(\frac{\ell^2log(1/\ell)}{\ve_2^2}\mright)$ samples from $p$ one can estimate $p(D)$ and $p(D_j)$ for every $j\in [\ell]$ to an additive $\ve_2/(20\ell)$ with probability at least $1-\frac{1}{5(2\ell+1)}$ each, which gives us the testers for (1) and (2).

\item Theorem~\ref{th:ubU} provides a tester that, for any fixed $j$, distinguishes between $\norm{p^j-q^j}_1\le \frac{2\ve_1}{\ell q(D_j)}$ and $\norm{p^j-q^j}_1\ge \frac{\ve_2}{5\ell q(D_j)}$ with probability of success $\ge 4/5$ and uses 
\[
\mathcal{O}\mleft( |D_j|\Big(\frac{\ve_1(\ell\cdot  q(D_j))}{\ve_2^2}\Big)^2+|D_j|\Big(\frac{\ve_1(\ell q(D_j)}{\ve_2^2}\Big)+\frac{(\ell\cdot  q(D_j))^2\sqrt{|D_j|}}{\ve_2^{2}} \mright)
\]
samples from $p^j$.
By standard amplification arguments one can achieve a probability of success of $1-\frac{1}{10(2\ell+1)}$ at the cost of a multiplicative $\mathcal O(\log(1/\ell))$ factor in the sample complexity.

To use this in order to obtain the tests required for (3), note that for any $j\in[\ell]$ such that $q(D_j)\ge \frac{\ve_2}{5\ell}$ if $|p(D_j)-q(D_j)|\ge \frac{\ve_2}{10\ell}$ then the corresponding test from (2) already outputs \textsf{reject} with high probability; so we can assume that $p(D_j)\ge q(D_j)-\frac{\ve_2}{10\ell} \ge q(D_j)/2$.
In this case for any $m>0$, using $m$ samples from $p$, we can get $\Omega(m p(D_j)/\log \ell )= \Omega(m q(D_j)/\log \ell )$ samples from $p^j$ with probability at least $1-1/(10\ell^2)$. Note that we can use the same overall set of $m$ samples from $p$ to obtain our $m_j = \Omega(m q(D_j)/\log \ell )$ samples from every $p^j$, $j\in[\ell]$. 

This gives us the testing algorithms for (3).
\end{itemize}

Combining these bounds, we get the following upper bound on the sample complexity:
\begin{align*}
\mathcal{O}&\mleft(\frac{\ell^2\log \ell}{\ve_2^2}\mright)+\max_{j\in \ell}\frac{\log\ell}{q(D_j)}\cdot \mathcal{O}\mleft( |D_j|\Big(\frac{\ve_1(\ell\cdot  q(D_j))}{\ve_2^2}\Big)^2+|D_j|\Big(\frac{\ve_1(\ell\cdot  q(D_j))}{\ve_2^2}\Big)+\frac{(\ell\cdot  q(D_j))^2\sqrt{|D_j|}}{\ve_2^{2}} \mright)\\
&=
\tilde{\mathcal O}\mleft(\frac{1}{\ve_2^2}\mright))+\max_{j\in \ell}\tilde{\mathcal{O}}\mleft( |D_j| q(D_j)\Big(\frac{\ve_1}{\ve_2^2}\Big)^2+|D_j|\Big(\frac{\ve_1}{\ve_2^2}\Big)+\frac{ q(D_j)\sqrt{|D_j|}}{\ve_2^{2}} \mright)\\
&=
\tilde{\mathcal O}\mleft(\frac{1}{\ve_2^2}\mright))+\max_{j\in \ell}\tilde{\mathcal{O}}\mleft( |D_j|^22^{-j}\Big(\frac{\ve_1}{\ve_2^2}\Big)^2+|D_j|\Big(\frac{\ve_1}{\ve_2^2}\Big)+\frac{ 2^{-j}\paren{|D_j|}^{3/2}}{\ve_2^{2}} \mright)\,,  
\end{align*}
where the last line uses the definition of $D_j$ in~\eqref{eq:bucketing:dj} (the ``bucketing'') to relate $q(D_j)$ to $|D_j|$.

To conclude, we observe that
\begin{align*}
\norm{q_{-\ve_2/20}}_0 &=n-|D|\ge \max_{j\in [\ell]} |D_j|\\
\norm{q_{-\ve_2/20}}_{\frac{1}{2}} &=\big(\sum_{j\in [\ell]}\sum_{i\in D_j}q_i^{1/2}\big)^2\ge  \max_{j\in [\ell]} (|D_j|2^{-j/2})^2.\\ 
\norm{q_{-\ve_2/20}}_{\frac{2}{3}} &=\big(\sum_{j\in [\ell]}\sum_{i\in D_j}q_i^{2/3}\big)^{3/2}\ge  \max_{j\in [\ell]} (|D_j|2^{-2j/3})^{3/2}.  
\end{align*}

Combining the above four equations, and using $\norm{q_{-\ve_2/20}}_{\frac{2}{3}}=\Omega(1)$, we get the following upper bound on the sample complexity: 
\begin{align*}
    \tilde{\mathcal{O}}\mleft( \norm{q_{-\ve_2/20}}_{\frac{1}{2}}\Big(\frac{\ve_1}{\ve_2^2}\Big)^2+\norm{q_{-\ve_2/20}}_0\Big(\frac{\ve_1}{\ve_2^2}\Big)+\frac{ \norm{q_{-\ve_2/20}}_{\frac{2}{3}}}{\ve_2^{2}} \mright).   
\end{align*}
This concludes the proof of the theorem.
\end{proof}

\end{document}